\documentclass[plain,biber]{nowfnt}

\usepackage[utf8]{inputenc}

\usepackage[ruled]{algorithm2e}
\usepackage{amssymb}
\usepackage{amsthm}
\usepackage{amsmath}
\usepackage{bm}
\usepackage[caption=false,subrefformat=parens]{subfig}
\usepackage{booktabs}
\usepackage{multirow}
\usepackage{threeparttable}
\usepackage{tabularx}
\usepackage{color}
\usepackage{stfloats}
\usepackage{multicol}
\usepackage{flushend}
\usepackage{nicematrix}
\usepackage{mathtools}

\DeclarePairedDelimiter\bbrac{\{}{\}}

\newcommand{\bbr}[1]{\bbrac*{#1}}

\def\R{\mathbb{R}}
\def\calX{\mathcal{X}}
\def\rmX{\mathrm{X}}
\def\Xcstr{\mathrm{X}_{\mathrm{cstr}}}
\def\Xinit{\mathrm{X}_{\mathrm{init}}}
\def\Xedls{\mathrm{X}_{\mathrm{edls}}}
\def\Xinitg{\mathrm{X}_{\mathrm{init}}^g}
\def\Xedlsg{\mathrm{X}_{\mathrm{edls}}^g}
\def\Xedlss{\mathrm{X}_{\mathrm{edls}}^*}
\def\calU{\mathcal{U}}
\def\pmb#1{\mathbf{#1}}
\def\st{\text{ s.t. }}

\newcommand{\change}[1]{#1}

\title{The Feasibility Theory of Constrained Reinforcement Learning: \\A Tutorial Study}

\maintitleauthorlist{
Yujie Yang \\
School of Vehicle and Mobility, Tsinghua University \\
yangyj21@mails.tsinghua.edu.cn
\and
Zhilong Zheng \\
School of Vehicle and Mobility, Tsinghua University \\
zheng-zl22@mails.tsinghua.edu.cn
\and
Masayoshi Tomizuka \\
Department of Mechanical Engineering, \\
University of California, Berkeley \\
tomizuka@berkeley.edu
\and
Changliu Liu \\
Robotics Institute, Carnegie Mellon University \\
cliu6@andrew.cmu.edu
\and
Shengbo Eben Li \\
School of Vehicle and Mobility \& College of AI, Tsinghua University \\
lishbo@tsinghua.edu.cn
}

\usepackage{mwe}

\author[1]{Yujie Yang}
\author[2]{Zhilong Zheng}
\author[3]{Masayoshi Tomizuka}
\author[4]{Changliu Liu}
\author[5]{Shengbo Eben Li}

\affil[1]{School of Vehicle and Mobility, Tsinghua University; yangyj21@mails.tsinghua.edu.cn}
\affil[2]{School of Vehicle and Mobility, Tsinghua University; zheng-zl22@mails.tsinghua.edu.cn}
\affil[3]{Department of Mechanical Engineering, University of California, Berkeley; tomizuka@berkeley.edu}
\affil[4]{Robotics Institute, Carnegie Mellon University; cliu6@andrew.cmu.edu}
\affil[5]{School of Vehicle and Mobility \& College of AI, Tsinghua University; lishbo@tsinghua.edu.cn}

\articledatabox{\nowfntstandardcitation\\ Y. Yang and Z. Zheng contributed equally to this work. All correspondence should be sent to S. E. Li.}

\begin{document}

\makeabstracttitle

\begin{abstract}
Satisfying safety constraints is a priority concern when solving optimal control problems (OCPs). Due to the existence of infeasibility phenomenon, where a constraint-satisfying solution cannot be found, it is necessary to identify a feasible region before implementing a policy. Existing feasibility theories in model predictive control (MPC) only apply to the case where the policy is a solution to a specific OCP, either optimal or suboptimal. Such feasibility is essentially feasibility of OCPs, not that of policies or states. However, reinforcement learning (RL), as another important control method, represents a policy as a mapping from state to action, which itself is decoupled with OCPs. An RL policy may be not a solution to a specific OCP, i.e., not satisfying the OCP's constraints, especially at an early stage of training. Feasibility analysis of such inadequately trained policies is necessary for safety improvement in RL; but that is not available under existing MPC feasibility theories. This paper proposes a feasibility theory that applies to both MPC and RL by decoupling states, constraints, and policies. Starting from a state, different constraints can be constructed, and different policies can be applied. We reveal that feasibility depends on the combination of these three elements, extending the traditional viewpoint of OCP-specific feasibility theories in MPC. The basis of our theory is to distinguish initial and endless, state and policy feasibility, and their corresponding feasible regions. Based on these concepts, we analyze the containment relationships between different feasible regions, which enables us to describe feasibility under arbitrary combinations of states, constraints, and policies. We further provide virtual-time constraint design rules along with a practical design tool called feasibility function that helps to achieve the maximum feasible region. The feasibility function either represents a control invariant set or aggregates infinite steps of constraints into a single one. We review most of existing constraint formulations and point out that they are essentially applications of feasibility functions in different forms. We demonstrate our feasibility theory by visualizing different feasible regions under both MPC and RL policies in an emergency braking control task.
\end{abstract}

% \begin{keyword}
% feasibility \sep constrained optimal control \sep model predictive control \sep reinforcement learning
% \end{keyword}

% \end{frontmatter}

\chapter{Introduction}
Optimal control is an important theoretical framework for sequential decision-making and control. The goal of solving an optimal control problem (OCP) is to find a policy that maximizes some performance index, usually measured through cumulative rewards. In many real-world control tasks, such as robotics \citep{brunke2022safe}, aerospace engineering \citep{ravaioli2022safe}, and autonomous driving \citep{guan2022integrated}, policies \change{must optimize performance while also guaranteeing safety}. \change{These} problems \change{are} formulated as constrained OCPs, where certain state constraints must be strictly satisfied at every time step.

% infeasibility phenomenon
In a constrained OCP, satisfying constraints at a single time step is \change{insufficient} to ensure long-term safety. \change{A} policy may still \change{reach} a state where \change{no future constraint-satisfying actions exist}. This issue is called the infeasibility phenomenon \citep{li2023reinforcement}. To avoid this problem, the feasible region of a policy, where there is always a constraint-satisfying solution, must be identified before real-world deployment. \change{Rigorously defining and analyzing feasible regions requires a general theory of feasibility for constrained OCPs.}

% existing feasibility theory (for MPC)
Existing feasibility theories are \change{primarily developed} for model predictive control (MPC), which \change{computes optimal actions} online through receding horizon optimization. A central concept in these theories is \textit{persistent feasibility} \citep{zhang2016switched, borrelli2017predictive}, also \change{known as} \textit{recursive feasibility} \citep{lofberg2012oops, boccia2014stability}, which describes long-term constraint satisfaction in a receding horizon control (RHC) problem. Specifically, an RHC problem is persistently feasible if its initially feasible set equals its maximal positive invariant set. \change{Under this condition, any initial constraint-satisfying solution ensures that subsequent receding horizon optimizations will not encounter infeasibility.}
Persistent feasibility only applies to the optimal policy of a constrained OCP because its maximal positive invariant set is defined \change{for} the closed-loop system under the optimal MPC controller. The concept of \textit{strong feasibility} \citep{gondhalekar2009controlled, gondhalekar2011mpc} and suboptimal MPC \citep{scokaert1999suboptimal, pannocchia2011conditions} expand persistent feasibility from the optimal policy to any initially feasible policy. \change{They guarantee long-term safety not by requiring an optimal solution, but merely a feasible solution to the RHC problem at each step.} Specifically, an RHC problem is strongly feasible if, from every state in the initially feasible set, the trajectory under any feasible solution remains in this set.

A \change{key} limitation of \change{MPC-based} feasibility theories is that they only apply to \change{policies that are solutions to a specific OCP---either optimal or suboptimal---meaning the policy inherently satisfies the OCP's constraints. In this  context, the policy is coupled with the OCP, and the feasibility analysis pertains to the problem itself, not to the policy or state in isolation}. However, reinforcement learning (RL), as another important control method, represents a policy as a mapping from state to action, \change{decoupling it from any single OCP}. Although RL \change{also solves an OCP, it does not obtain the optimal solution until training converges. During early training stages, an RL policy may not even be a feasible solution to the intended OCP.} Feasibility analysis of such inadequately trained policies is \change{crucial} for safety improvement in RL, \change{yet this falls outside the scope of existing MPC feasibility theories}.

\change{There are a few existing works on safe RL that mention the concept of feasibility and leverage it to guarantee long-term safety. For example, \citet{ma2021feasible} define feasible states as those from which the cumulative cost constraints can be satisfied under a given policy, and introduce the feasible region as the set of all states that are feasible under any policy. \citet{yu2022reachability} define the feasible set of a policy as the set of states where infinite-horizon state constraints can be satisfied, and introduce the largest feasible set as the union of feasible sets over all policies. \cite{zheng2024safe} follow the definition of \citet{yu2022reachability} and extend the application of feasibility to offline settings.
However, these works have not yet formed a systematic feasibility framework and use varying, often inconsistent, terminology. More importantly, their notion of feasibility refers to the satisfaction of real-world safety constraints during policy execution, which differs from the classical MPC notion of feasibility, i.e., the existence of a solution to the OCP. In the context of RL, feasibility should also consider the existence of solutions to the underlying OCP during the training phase, where algorithmic constraints may differ from those in deployment. This aspect remains under-explored in existing safe RL literature and constitutes a major focus of our work.}

\begin{figure}
    \centering
    \includegraphics[width=\linewidth, trim=0 20 0 10, clip]{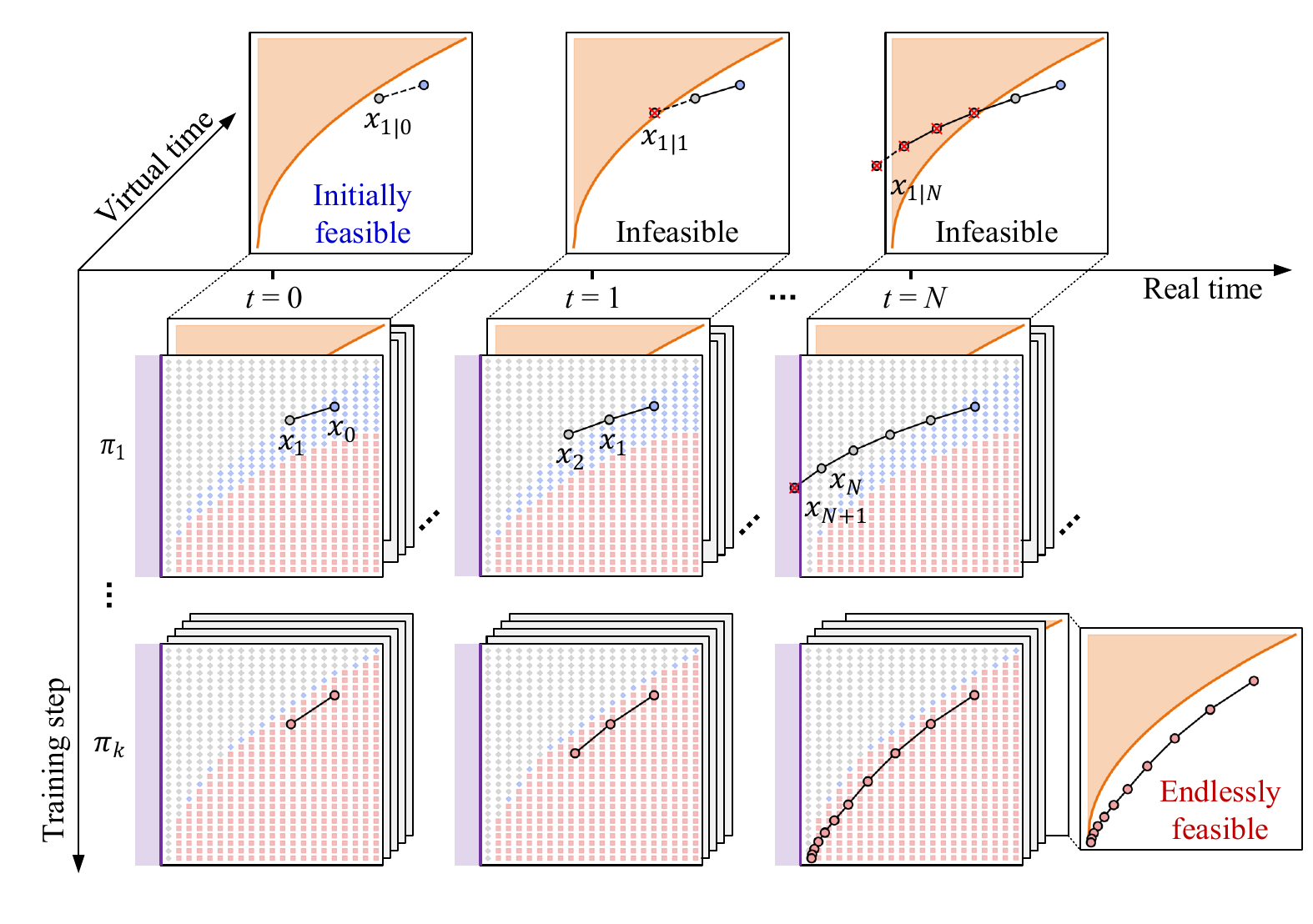}
    \caption{\change{Illustration of core concepts in feasibility theory. The second row shows state trajectories and feasible regions of an RL policy trained for 50 steps under a Hamilton-Jacobi reachability constraint, with the first row demonstrating the first virtual-time step corresponding to each real-time step. The rest grey shaded squares are stacked along the virtual-time axis. The third row shows the results of 10000 training steps. The purple-shaded areas represent real-time constraints, and the orange-shaded areas represent virtual-time constraints. The circles stand for the states at every step on a trajectory. The squares and diamonds stand for regular states so as to show different feasibility regions. Regardless of the shape, all blue marks represent initially feasible states, all red marks for endlessly feasible states, and all gray marks for infeasible states. Intuitively, ``initially feasible'' means that there exists a policy satisfying the virtual-time constraint at the current step, while ``endlessly feasible'' means that such a policy will always exist. A state with a red cross indicates violation of real-time or virtual-time constraints. The data in the figure comes from a numerical example of emergency braking control, and details can be found in Chapter~\ref{sec: experiments}.}}
    \label{fig: framework}
\end{figure}

In order to fill this gap, we propose a set of feasibility theories for constrained OCPs that apply to both MPC and RL. The key to our theory is a decoupling view of states, constraints, and policies in constrained OCPs. Starting from a state, different constraints can be constructed, and different policies can be applied. We reveal that feasibility depends on the combination of these three elements instead \change{of} just on \change{the} OCP itself, which extends the traditional viewpoint of OCP-specific feasibility theories in MPC. Figure \ref{fig: framework} illustrates the core concepts of our feasibility theory. Any optimal control method follows the practice of solving \change{an} OCP in \change{the} virtual-time domain and implementing \change{a} policy in \change{the} real-time domain. Each real-time step corresponds to a virtual-time constrained OCP, \change{solved either at each step online (as in MPC) or offline before deployment (as in RL)}.
\change{
In the real-time domain, the control task naturally has some intrinsic constraints, termed real-time constraints, that are immutable by human design, such as collision avoidance.
However, the original real-time constraints may be challenging to solve.
As a result, in the virtual-time domain, it is typically necessary to design the virtual-time constraints differing from the real-time ones to facilitate solving.
The satisfiability of virtual-time constraints over time defines different types of feasibility and feasible regions.
}
For a typical safe RL algorithm, the feasible region of its policy starts as a small set and gradually expands during training, \change{eventually approaching that} of the optimal policy.
\change{This is because the policy is typically initialized with random parameters, leading to unpredictable and often unsafe actions, which corresponds to a feasible region. As training proceeds, the policy is explicitly optimized to maximize reward while minimizing constraint violations. This process gradually guides the policy towards safer behaviors, effectively expanding the feasible region.}
\change{In contrast, MPC, which directly computes optimal control sequences, typically maintains a fixed feasible region.}
The main contributions of this paper are summarized as follows.
\begin{itemize}
    \item We propose a set of feasibility theories for constrained OCPs that apply to both MPC and RL. \change{By decoupling} states, constraints, and policies, we \change{define separate notions of initial and endless feasibility for both states and policies, characterizing short-term and long-term safety with and without a specific policy}. \change{Based on these}, we derive four kinds of feasible regions and \change{establish} the relationship between the maximum feasible region and real-time constraints.
    \item We \change{derive} and prove containment relationships \change{among} all feasible regions, which provides a tool for analyzing feasibility \change{for} arbitrary combinations of states, constraints, and policies. In particular, we analyze the relationship between a policy-specific feasible region and the maximum feasible region, which is a \change{primary} concern when solving a constrained OCP. \change{We also show that for any policy, its feasible region is bounded above and below by known feasible regions.}
    \item We provide \change{practical rules for designing virtual-time constraints and introduce a tool called feasibility function to help achieve the maximum feasible region}. Specifically, the maximum feasible region \change{is achieved by} designing a maximum initially feasible region when initially and endlessly feasible regions are equal. The feasibility function \change{guarantees} this equality by either constructing a control invariant set or aggregating \change{infinite-horizon} constraints into a \change{single step}. We review existing constraint formulations and \change{show} that they are essentially \change{specific} applications of feasibility functions.
\end{itemize}

\change{It is worth emphasizing that what we intend to propose in this paper is a feasibility theory regarding how different designs of virtual-time constraint influence the feasibility property of the solution of an constrained OCP, rather than a specific solving method.
Hence, this paper does not cover discussions on how to obtain or to what extent we can obtain the optimal solution of an constrained OCP.}

The rest of this paper is organized as follows.
Chapter \ref{two temporal domains} introduces a decoupling view of two temporal domains.
Chapter \ref{sec: infeasibility phenomenon} illustrates the infeasibility phenomenon.
We define feasibility and feasible regions in Chapter \ref{feasibility and feasible region} and \change{analyze their containment relationships} in Chapter \ref{feasible region property}, \change{we then introduce the} feasibility function for constraint design in Chapter \ref{feasibility function} and review existing constraint formulations in Chapter \ref{constraint review}.
Finally, Chapter \ref{sec: experiments} demonstrates feasible regions with experiments, and Chapter \ref{conclusion} concludes the paper.

\chapter{Constrained Optimal Control Problems in Two Temporal Domains}
\label{two temporal domains}
\section{Constrained optimal control problems}
A constrained OCP can be described by a deterministic Markov decision process $(\mathcal{X},\mathcal{U},f,r,\gamma,d_{\mathrm{init}})$, where $\mathcal{X}\subseteq\mathbb{R}^n$ is the state space, $\mathcal{U}\subseteq\mathbb{R}^m$ is the action space, $f: \mathcal{X}\times\mathcal{U}\to\mathcal{X}$ is the dynamic model, $r: \mathcal{X}\times\mathcal{U}\to\mathbb{R}$ is the reward function, $\gamma\in(0,1]$ is a discounting factor and $d_{\mathrm{init}}$ is the initial state distribution.

The constraint is specified through inequalities
\begin{equation}
\label{real-time constraint}
    h(x_{t+i})\le0, i=0,1,2,\dots,\infty,
\end{equation}
where $h: \mathcal{X}\to\mathbb{R}$ is the constraint function. The constrained set is defined as $\Xcstr=\{x\in\mathcal{X}|h(x)\le0\}$. \change{Its complement, $\bar{\mathrm{X}}_\mathrm{cstr}=\mathcal{X}\setminus \Xcstr$, contains all states violating the constraints.} Our aim is to find a policy $\pi: \mathcal{X}\to\mathcal{U}$ that maximizes the expected cumulative rewards under the state constraints,
\begin{equation}
\label{real-time OCP}
\begin{aligned}
    \max_\pi \quad& \mathbb{E}_{x_t\sim d_{\mathrm{init}}(x)}\left\{\sum_{i=0}^\infty \gamma^i r(x_{t+i},u_{t+i})\right\}, \\
    \st \quad& x_{t+i+1}=f(x_{t+i},u_{t+i}), \\
    & h(x_{t+i})\le0, i=0,1,2,\dots,\infty.
\end{aligned}
\end{equation}

\section{Real-time domain and virtual-time domain}
To explain feasibility, it is necessary to recall the working mechanism of optimal control. In essence, any optimal controller works in two separated temporal domains: virtual-time domain and real-time domain. In virtual-time domain, an OCP is solved to obtain an optimal policy, and in real-time domain, the optimal policy is implemented. Virtual-time domain and real-time domain have their own OCP. The above OCP \eqref{real-time OCP} is actually defined in real-time domain. The motivation for defining another OCP in virtual-time domain is that OCP \eqref{real-time OCP} is difficult to solve due to its infinite-step constraint. Introducing virtual-time OCP allows us to replace this constraint with a finite-step one:
% The separation of the virtual-time domain from the real-time domain inspires us that a different constraint can be used in virtual-time domain from that in real-time domain since constrained OCP is not defined in the real-time domain. Therefore, one can choose a new constraint function in the virtual-time domain:
\begin{equation}
\label{virtual-time constraint}
    g(x_{i|t})\le0, i=0,1,2,\dots,n,
\end{equation}
where $g(\cdot)$ is the virtual-time constraint function, and $n$ is its \change{finite} horizon length. 
% This length can be either finite or infinite. 
With virtual-time constraint \eqref{virtual-time constraint}, we can define the virtual-time OCP as follows.
\begin{equation}
\label{virtual-time OCP}
\begin{aligned}
    \max_\pi \quad& \mathbb{E}_{x_{0|t}\sim d_{\mathrm{init}}(x)}\left\{\sum_{i=0}^N \gamma^i r(x_{i|t},u_{i|t})\right\}, \\
    \st \quad& x_{i+1|t}=f(x_{i|t},u_{i|t}), \\
    & g(x_{i|t})\le0, i=0,1,2,\dots,n,
\end{aligned}
\end{equation}
where $N$ is the horizon length of virtual-time objective function, which can be either finite (as in MPC) or infinite (as in RL).

It must be noted that OCP \eqref{real-time OCP} is defined in the real-time domain, and OCP \eqref{virtual-time OCP} is defined in the virtual-time domain. The two domains are distinguished by their subscripts: $t+i$ represents the $(t+i)$-th step in the real-time domain, and $i|t$ represents the $i$-th point in the virtual-time domain starting from time $t$. The virtual-time constraint \eqref{virtual-time constraint} can be different from the real-time constraint \eqref{real-time constraint} in two aspects: 1) constraint function and 2) time horizon. The real-time constraint is determined by the control task itself and cannot be modified by algorithm designers. Moreover, it always has an infinite horizon because any real-world control task is continuous without termination. In contrast, the virtual-time constraint is constructed by algorithm designers and has a finite horizon \change{$n$, which is different from the horizon length of virtual-time objective function $N$}.
% does not need to be posed at every step. It even does not need to have the same function form at all steps. 
% That is to say, except at some special time instances, the constraints in virtual-time domain can be changed or even removed. 
There is only one requirement for virtual-time constraint: it must be not weaker than the real-time constraint in the current step, i.e., if $g(x_{0|t})\le0$, then $h(x_t)\le0$. 
% Besides building the basic structure of defining feasibility, this domain-separating perspective provides us with great flexibility to build various virtual-time constraints, as well as the formulation of their constrained OCPs.
% As an example of virtual-time constraint, a commonly used design in MPC is that the first $n-1$ constraints are the same as the real-time constraint $h$, and the last constraint $g(x_{n|t})$, called the terminal constraint, is designed differently, typically stronger than $h$ to ensure that $h$ is satisfied at every real-time step.

\change{Recall the commonly used constrained optimal control formulation in MPC:
\begin{equation}
\label{eq: MPC formulation}
\begin{aligned}
    \min_{u_{0|t},u_{1|t},\dots,u_{N|t}} \quad& \sum_{i=0}^{N-1} l(x_{i|t},u_{i|t})+l_T(x_{N|t},u_{N|t}), \\
    \st \quad& x_{i+1|t}=f(x_{i|t},u_{i|t}), \\
    & h(x_{i|t})\le0, i=0,1,2,\dots,N-1, \\
    & g_T(x_{N|t})\le0,
\end{aligned}
\end{equation}
where $l(\cdot,\cdot)$ is the utility function. Problem \eqref{eq: MPC formulation} is actually a variant of the virtual-time OCP \eqref{virtual-time OCP}. Specifically, 1) the horizon length of the constraint function in problem \eqref{eq: MPC formulation} equals that of the  objective function, i.e., $n=N$, 2) the maximization of rewards is changed to the minimization of cost, 3) the first $n$ steps of the virtual-time constraints equals the real-time constraints, i.e., $g(x_{i|t})=h(x_{t+i}),i=0,1,2,\dots,N-1$, 4) the last utility is a specially designed terminal cost $l_T(\cdot)$, and 5) the last constraint is a specially designed terminal constraint $g_T(\cdot)$, typically stronger than $h$ to ensure that $h$ is satisfied at every real-time step.}
%? g(x_{i|t})=h(x_{i|t})?

\change{We visualize the trajectories of an MPC controller obtained by solving problem \eqref{eq: MPC formulation} in an emergency braking task in Figure \ref{fig: real virtual trajectory}. Details of the task can be found in Chapter \ref{sec: experiments}. The prediction horizon is 5 and there is no terminal constraint. It can be observed that the two trajectories always coincide at the first two steps since the first virtual-time action is implemented in the real-time domain. Beyond that, the two trajectories diverge when the virtual-time trajectory starts from $t=2$ and $t=3$. This is because the prediction horizon is not long enough to take into account all future constraints, leading to a short-sighted objective-seeking policy. The virtual-time trajectory starting from $t=4$ and $t=5$ overlap with the real-time trajectory because in both cases the vehicle is doing best-effort control to avoid constraint violation, which is already impossible given the current state and control limits. This illustrates the fundamental difference between virtual-time and real-time domains in constrained OCPs.}

\begin{figure}
    \centering
    \subfloat{
        \includegraphics[trim=10 20 10 20, width=0.22\linewidth]{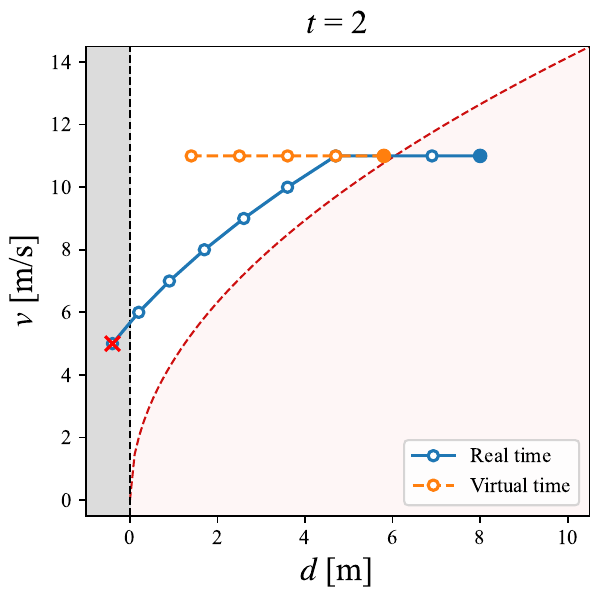}
    }
    \subfloat{
        \includegraphics[trim=10 20 10 20, width=0.22\linewidth]{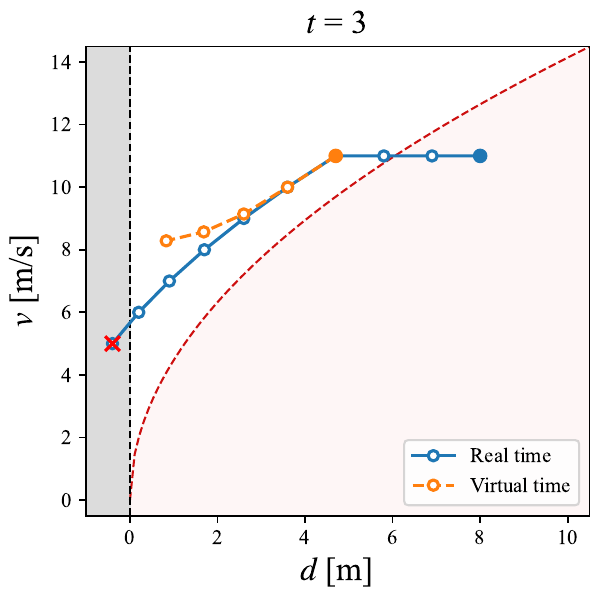}
    }
    \subfloat{
        \includegraphics[trim=10 20 10 20, width=0.22\linewidth]{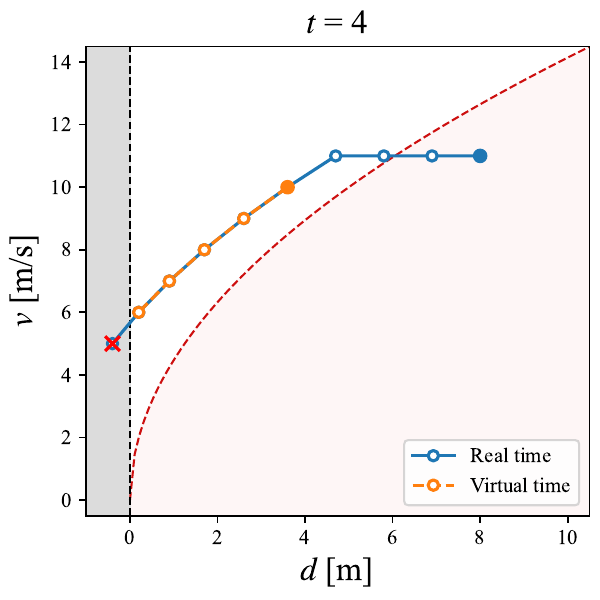}
    }
    \subfloat{
        \includegraphics[trim=10 20 10 20, width=0.22\linewidth]{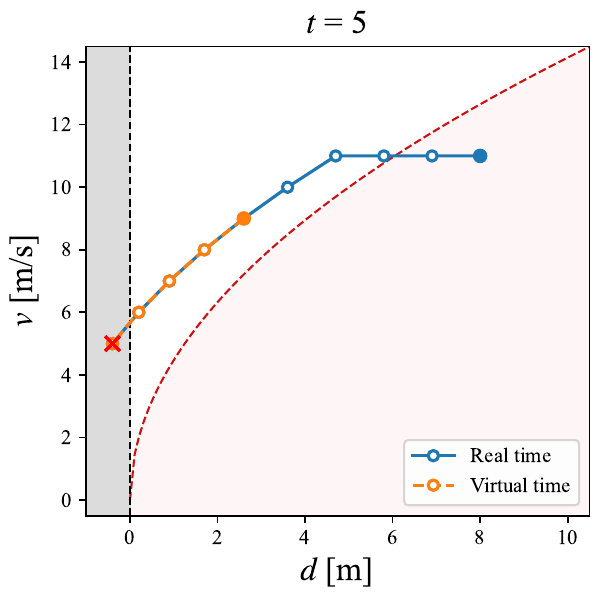}
    }
    \caption{\change{Real-time and virtual-time trajectories of an MPC controller in an emergency braking task. The real-time trajectories are the same across the four figures. The virtual-time trajectories starts at different real-time steps. The gray-shaded area represents the constraint-violating set. The red-shaded area represents the set where there exists a policy to ensure infinite-horizon safety. A state with a red cross indicates violation of the real-time constraint.}}
    \label{fig: real virtual trajectory}
\end{figure}

\chapter{Illustration of the Infeasibility Phenomenon}
\label{sec: infeasibility phenomenon}
Infeasibility is the most important concept in constrained OCPs. 
It describes the phenomenon that the constraint of a virtual-time OCP cannot be satisfied.
Formally, we define infeasibility as follows.
\change{
\begin{definition}[Infeasibility]
\label{def:infeasibility}
Consider a virtual-time OCP with virtual-time constraint $g(\cdot)\le0$ over horizon $i=0,\dots,n$.
\begin{enumerate}
    \item A state $x$ is infeasible (w.r.t. $g$) if there exists no policy $\pi$ that is a solution to the virtual-time OCP.
    \item A policy $\pi$ is infeasible in state $x$ (w.r.t. $g$) if it is not a solution to the virtual-time OCP.
\end{enumerate}
\end{definition}
}
There are two reasons for the infeasibility phenomenon: 1) imperfect policy solving and 2) improper constraint design.

To better explain how the infeasibility phenomenon occurs, let us consider the following example. Figure \ref{fig: framework} shows an infeasible emergency braking control of an RL policy under a one-step constraint that requires the state not to enter the orange region. 
% Here, the environment model is assumed to be perfect, i.e., the action and state trajectories are the same in the virtual-time and real-time domains.
In the first two rows, at $t=0$, policy $\pi_1$ satisfies the virtual-time constraint, i.e., $x_{1|0}$ is not in the orange region. The state is transferred to $x_1$ in the real-time domain, which equals $x_{1|0}$. At $t=1$, $\pi_1$ cannot satisfy the virtual-time constraint, i.e., $x_{1|1}$ enters the orange region. This is how the infeasibility phenomenon occurs. In this situation, we say $\pi_1$ is infeasible in state $x_1$. Note that constraint violation in virtual-time domain does not mean that there is no admissible action in real-time domain, where ``admissible'' means leading to a constraint-satisfying state. At $t=1$, the action is still admissible in real-time domain because the resulting next state $x_2$ is still inside the real-time constrained set. However, since the virtual-time constraint is already violated, the violation of the real-time constraint is inevitable sometime in the future. As shown in the figure, at time $t=N$, the next state $x_{N+1}$ finally goes out of the real-time constrained set.

The above infeasibility phenomenon is mainly caused by imperfect policy solving. Extending the training step from 50 to 10000 basically solves the infeasibility problem. As shown in the third row in Figure \ref{fig: framework}, the state trajectory is entirely contained in the virtual-time constrained set.
\change{More detailed examples can be found in Chapter \ref{sec: experiments}, where we demonstrate through an emergency braking task that the feasible regions of RL policies expand as training proceeds, indicating that higher solution quality leads to better feasibility.}
Another reason for the infeasibility phenomenon, i.e., improper constraint design, is illustrated in Figure \ref{fig: infeasibility PW}. Here, the virtual-time constraint function is a pointwise constraint, i.e., $g=h$ with a finite horizon $n$. The policy is an MPC controller. It shows that the infeasibility phenomenon occurs when $n$ is too small and disappears when $n$ is large enough.
\change{More detailed examples on the effect of constraint design can also be found in Chapter \ref{sec: experiments}, where we show that with an MPC controller, different virtual-time constraints leads to different feasible regions.}
\change{Intuitively, weaker and fewer steps of virtual-time constraints have worse feasibility in the long run because they cannot provide sufficient confinement to future state trajectories.
% , as can be corroborated by the results of the examples in Chapter \ref{sec: experiments}.
This intuition is formally stated and proved in Section \ref{sec: containment relationship}.}

Regarding virtual-time constraint design, we provide a collection of design rules and a practical design tool called feasibility function in Chapter \ref{feasibility function}. This tool helps avoid the infeasibility phenomenon and achieve the maximum feasible region.

\begin{figure}
    \centering
    \subfloat[$n=2$]{
        \includegraphics[trim=10 20 10 20, clip, width=0.3\linewidth]{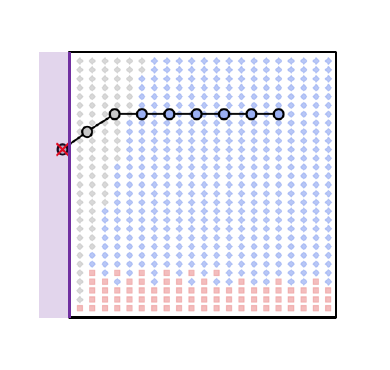}
    }
    \quad
    \subfloat[$n=10$]{
        \includegraphics[trim=10 20 10 20, clip, width=0.3\linewidth]{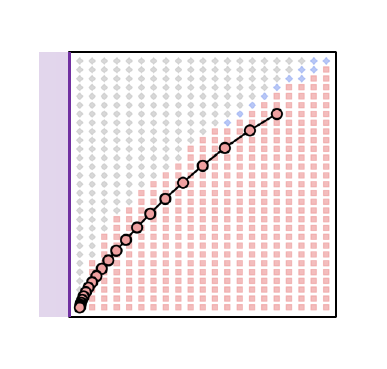}
    }
    \caption{\change{State trajectories and state feasibility of MPC under pointwise constraints. The circles stand for the states at every step on the trajectories, which are obtained by solving problem \eqref{eq: MPC formulation} in an receding-horizon manner. The squares and diamonds stand for regular states so as to show different feasibility regions. Regardless of the shape, all red marks represent endlessly feasible states, all blue marks for initially feasible states, and all gray marks for infeasible states. The purple-shaded areas represent real-time constraint.}}
    \label{fig: infeasibility PW}
\end{figure}

\chapter{Feasibility and Feasible Regions}
\label{feasibility and feasible region}
We have seen that improper design of virtual-time constraints will result in the infeasibility phenomenon. Now, let us move to the formal definition of feasibility, which allows us to analyze what kind of virtual-time constraint can guarantee safety. The aforementioned example of pointwise constraint \change{suggests} us that depending on the choice of constraint, virtual-time OCPs may suddenly lose feasibility at some time points or keep feasibility forever. These two results bring about the need to distinguish between 1) initial feasibility and 2) endless feasibility. As suggested by their names, initial feasibility is only temporary, while endless feasibility regards all virtual-time OCPs in the everlasting future.

\section{Initial feasibility}
The property \change{that a state yields} a feasible OCP at the current step is termed as initial feasibility. This property is just a shortsighted feature and may not give any guarantee in the future (even just one step later!). To understand this, recall the evolution of infeasibility phenomenon in Figure \ref{fig: framework}, where $x_t$ is a feasible state but still evolves into an infeasible state $x_{t+2}$. The definition of initial feasibility is as follows.

\begin{definition}[Initial feasibility of a state]
\label{def: initial feasibility of state}~
\begin{enumerate}
    \item A state $x$ is initially feasible if there exists a policy that satisfies the virtual-time constraint starting from $x$, i.e., $\exists\pi$, s.t. $g(x_{i|t})\le0, i=0,1,2,\dots,n$, where $x_{0|t}=x$. 
    \item The initially feasible region (shortened as IFR), denoted as $\Xinitg$, is the set of all states that are initially feasible.
\end{enumerate}
\end{definition}

The subscript ``init'' of initial feasible region $\Xinitg$ stands for ``initial'', and the superscript ``$g$'' emphasizes that it is related to the virtual-time constraint. For a state to be initially feasible, the only requirement is that the virtual-time OCP starting from it has a solution. 
% Outside the IFR, there is no policy satisfying the virtual-time constraint, and hence the virtual-time OCP has no solution. In other words, IFR is a region restricted by the chosen virtual-time constraint, identifying those states that are meaningful under the virtual-time constraint. 
In this definition, the policy $\pi$ can be arbitrarily chosen to meet the existence statement. Supposing that we have already been given a policy, it is also a natural need to check whether the policy satisfies the virtual-time constraint, which leads to the definition of initial feasibility of a policy. 

\begin{definition}[Initial feasibility of a policy]
\label{def: initial feasibility of policy}~
\begin{enumerate}
    \item A policy $\pi$ is initially feasible in a state $x$ if $\pi$ satisfies the virtual-time constraint starting from $x$, i.e., $g(x_{i|t})\le0, i=0,1,2,\dots,n$, where $x_{0|t}=x$.
    \item The set of all policies that are initially feasible in $x$ is denoted as $\Pi_\mathrm{init}^g(x)$.
    \item The initially feasible region of a policy $\pi$, denoted as $\Xinitg(\pi)$, is the set of all states in which $\pi$ is initially feasible.
\end{enumerate}
\end{definition}

The notation $\Xinitg(\pi)$ means that this region is a function of policy $\pi$. The initial feasibility of a policy is closely related to the initial feasibility of a state. If a state is initially feasible, then there must exist a policy that is initially feasible in this state. On the other hand, if a policy is initially feasible in a state, then there exists at least one solution to virtual-time OCP starting from this state, which indicates that this state is initially feasible. 
% The IFR of a policy $\Xinitg(\pi)$ is the region where one is approved to take policy $\pi$ under the virtual-time constraint. Similar to the initial feasibility of a state, the initial feasibility of a policy only describes the property of a single virtual-time OCP starting from the current real-time step. It does not provide any guarantee on the virtual-time OCPs in future real-time steps. A policy that is initially feasible at the current time step may become infeasible after it takes a step forward in the real-time domain. 
Examples of IFR of a policy are the union of blue and red regions in Figure \ref{fig: framework} and Figure \ref{fig: infeasibility PW}. Within these regions, the virtual-time constraints are satisfied by the policy, but there is no guarantee that they will still be satisfied in successive states.

\section{Endless feasibility}
With initial feasibility, we are able to describe what we actually care about, i.e., a long-term property taking infinite horizon constraints into account. This property is termed as endless feasibility. Intuitively, endless feasibility guarantees that all future virtual-time OCPs are feasible. Here is an example of the distinction between initial and endless feasibility. In a car-following scenario, the ego vehicle is required to follow a leading vehicle while keeping a safe distance from it. Suppose that the virtual-time constraint requires a positive distance between vehicles only at the next time step. Then, a state with an overly high speed may be initially feasible since the vehicles will not collide immediately. However, due to the limited braking capability, the car will inevitably reach a state where collision is doomed at the next time step no matter what action is taken. At this time, the virtual-time OCP becomes infeasible. Like initial feasibility, endless feasibility is also separately defined for a state and a policy.

\begin{definition}[Endless feasibility of a state]
\label{def: endless feasibility of state}~
\begin{enumerate}
    \item A state $x$ is endlessly feasible if it is initially feasible and its successive states under any initially feasible policy in each time step are initially feasible, i.e., $x_t\in\Xinitg$ and $\forall\pi_{t+i}\in\Pi_\mathrm{init}^g(x_{t+i}), x_{t+i+1}\in\Xinitg, i=0,1,2,\dots,\infty$, where $x_t=x$.
    \item The endlessly feasible region (shortened as EFR), denoted as $\Xedlsg$, is the set of all states that are endlessly feasible.
\end{enumerate}
\end{definition}

The subscript ``edls'' in the above symbols is an abbreviation of ``endless''. 
% The endless feasibility of a state is related to not only the current step in the real-time domain but also all future time steps. The states in future time steps can be obtained by arbitrary initially feasible policies in each time step. 
The arbitrariness of policy selection is critical to the definition of endless feasibility. This is because in any initially feasible state, there may exist more than one initially feasible policy. Which one is the solution to the virtual-time OCP depends on the objective function. The arbitrariness of policy selection ensures that the solution can always ensure the successive states to be initially feasible, regardless of the objective function. 
% As a result, the EFR identifies those states that can be definitely rendered safe by the virtual-time constraint. 
% Similar to the initial feasibility of a policy, we can also define the endless feasibility of a policy.

\begin{definition}[Endless feasibility of a policy]
\label{def: endless feasibility of policy}~
\begin{enumerate}
    \item A policy $\pi$ is endlessly feasible in a state $x$ if $\pi$ is initially feasible in both $x$ and the successive states of $x$ under $\pi$, i.e., $\pi\in\Pi_\mathrm{init}^g(x_{t+i}),i=0,1,2,\dots,\infty$, where $x_t=x$.
    \item The endlessly feasible region of a policy $\pi$, denoted as $\Xedlsg(\pi)$, is the set of all states in which $\pi$ is endlessly feasible.
\end{enumerate}
\end{definition}

% The ``$\pi$'' in brackets serves the same purpose as in $\Xinitg(\pi)$. 
Note that the definition of endless feasibility is built upon initial feasibility, requiring all successive states in a trajectory to be initially feasible. For endless feasibility of a policy, this trajectory is naturally induced by the given policy. For endless feasibility of a state, this trajectory can be induced by any initial feasible policy. With such a definition, it naturally holds that for an endlessly feasible state, there must also be an endlessly feasible policy in that state, i.e., EFR is a subset of any policy's EFR. This conclusion will be formally stated and proved in Theorem \ref{thm: Feasible region containment}. Examples of EFR of a policy are the red regions in Figure \ref{fig: framework} and Figure \ref{fig: infeasibility PW}. Within these regions, the policy is initially feasible in both the current and all successive states.

\change{Readers who are familiar with MPC may find the concept of endless feasibility similar to the persistent feasibility in MPC.
A persistently feasible set is a set that is forward invariant with respect to the MPC feedback law $\pi_{\text{MPC}}$ \citep{boccia2014stability}.
The relation between these two concepts lies in that, given a constrained OCP in MPC, the endlessly feasible region of a $\pi_{\text{MPC}}$, i.e., $\Xedls^g(\pi_{\text{MPC}})$, equals the maximum persistently feasible set.
This is because $\Xedls^g(\pi_{\text{MPC}})$ is forward invariant under $\pi_{\text{MPC}}$, and for any $x$ in the maximum persistently feasible set, $\pi_{\text{MPC}}$ must be initially feasible on $x$ and all the subsequent states.
However, this feedback law is a specific policy that is made up with the solutions of a series of OCPs, which makes it closely tied to the constrained OCP. Hence, traditional feasibility analysis pertains to the problem itself, not to the state or policy in isolation, which is inadequate in the RL cases.
Note the subtle distinction: persistent feasibility in MPC is typically an existential property, i.e., there exists a control action or policy that keeps successive states initially feasible, whereas our notion of endless feasibility is stronger and is achieved only when the conditions hold for all initially feasible policies.
The proposed endless feasibility (and the whole feasibility theory) therefore offers a more fine-grained discussion of feasibility at both state and policy levels. At the state level, the definition of feasibility is independent of any specific policy. At the policy level, the policy under discussion is not constrained to be the optimal one or the (suboptimal) solution of an OCP.}

% The above-mentioned four definitions are all specific to virtual-time constraints. 
Different choices of virtual-time constraints will lead to different feasible regions. A natural question is how to design a virtual-time constraint that induces an EFR as large as possible. To discuss this, we first need to define the maximum EFR.
\begin{definition}
\label{def: maximum EFR}
The maximum endlessly feasible region (shortened as maximum EFR), denoted as $\Xedlss$, is the union of EFRs under all possible virtual-time constraints.
\end{definition}
\noindent Maximum EFR is closely related to the notion of maximum control invariant set in MPC. The latter refers to the maximum set in which for any state, there exists a control sequence that keeps its successive states still in this set. The maximum EFR further requires that all states in it must satisfy the real-time constraint.
% In other words, a constraint-satisfying maximum control invariant set is a maximum EFR.
Either too strong or too weak a virtual-time constraint will induce an EFR $\Xedlsg$ smaller than $\Xedlss$. For the former case, less state is considered initially feasible and hence outside of $\Xedlsg$. For the latter case, one may take too aggressive actions at the beginning due to inadequate restriction of virtual-time constraint, leading to the infeasibility phenomenon.
Examples of these two cases can be found in Chapter \ref{sec: experiments}, where a pointwise constraint with $n=2$ is too weak, and a control barrier function constraint with $k=0.5$ is too strong.
\change{The following theorem provides an equivalent characterization of the maximum EFR (Definition \ref{def: maximum EFR})}, explicitly relating endless feasibility to constraint satisfaction in the real-time domain.

\begin{theorem}
\label{thm: maximum EFR equivalence}
For an arbitrary state $x$, the following two statements 1) and 2) are equivalent:
\begin{enumerate}
    \item[1)] $x\in\Xedlss$.
    \item[2)] $\exists\pi$, s.t. $h(x_{t+i})\le0, i=0,1,2,\dots,\infty, x_t=x$.
\end{enumerate}
\end{theorem}

\begin{proof}
First, we prove 1) $\Rightarrow$ 2).
According to Definition \ref{def: maximum EFR}, 
$$
\forall x\in\Xedlss, \exists\pi \text{ and } g, \st x_{t+i}\in\Xinitg,i=0,1,2,\dots,\infty.
$$
Since $\Xinitg\subseteq\Xcstr$, we have $x_{t+i}\in\Xcstr$, i.e., $h(x_{t+i})\le0$. 

Next, we prove 2) $\Rightarrow$ 1).
For an arbitrary state $x$ that satisfies 2), we can choose 
$$
g(x_{i|t})=h(x_{i|t}),i=0,1,2,\dots,\infty
$$
as the virtual-time constraint. In this case, we have $x\in\Xedlsg$. Since $\Xedlsg\subseteq\Xedlss$, we conclude that $x\in\Xedlss$.
\end{proof}

The above theorem demonstrates the importance of maximum EFR: it exactly contains all states that can be rendered safe in real-time domain. 
% That is to say, for any state $x\in\Xedlss$, there exists a policy $\pi$, such that all future states are in the constrained set $\Xcstr$. Since statement 2) in Theorem \ref{thm: maximum EFR equivalence} goes to infinity, the successive states are always in the maximum EFR.
We have a similar conclusion for EFR of a policy: if a state $x$ is in the EFR of a policy $\pi$, it can be rendered safe by this policy. This is formally stated in the following theorem, which theoretically justifies applying a policy in its EFR.

\begin{theorem}
\label{thm: EFR sufficient condition}
For any state $x$, statement 1) is a sufficient condition for statement 2):
\begin{enumerate}
    \item[1)] $\exists g, \st x\in\Xedlsg(\pi)$.
    \item[2)] $h(x_{t+i})\le0, i=0,1,2,\dots,\infty, x_t=x, u_t=\pi(x_t)$.
\end{enumerate}
\end{theorem}

\begin{proof}
According to Definition \ref{def: endless feasibility of policy},
$$
\forall x\in\Xedlsg(\pi), x_{t+i}\in\Xinitg, i=0,1,2,\dots,\infty,
$$
where $x_{t+i}$ are obtained by $\pi$. Since $\Xinitg\subseteq\Xcstr$, we have $x_{t+i}\in\Xcstr$, i.e., $h(x_{t+i})\le0$.
\end{proof}

Theorem \ref{thm: EFR sufficient condition} tells us that the EFR of a policy is where we can safely apply the policy in real-time domain. This is the reason why we need to find not only the policy but also its EFR when solving constrained OCPs. In literature, this EFR is found by identifying a safety certificate, e.g., control barrier function (CBF) \citep{ames2014control}, safety index (SI) \citep{liu2014control}, and Hamilton-Jacobi (HJ) reachability function \citep{mitchell2005time}. 
% In fact, initial feasibility is not what we ultimately care about because it only describes whether constraints of the current virtual-time OCP are satisfied. The role of initial feasibility is to help us define endless feasibility, which describes whether constraints of all successive virtual-time OCPs are satisfied. Our main concern is endless feasibility because it is related to long-term safety in real-time domain. Therefore, we care about how to calculate the EFR of a policy, especially the policy obtained from solving a virtual-time OCP. 
When solving a constrained OCP, we aim to find a policy with a maximum EFR. This is because the size of EFR not only determines the policy's working region but also affects its performance. If we want a policy to be optimal, its EFR must contain all optimal trajectories. A larger EFR has a greater chance of containing optimal trajectories and therefore increases the probability of finding the optimal policy.

\change{However, it is important to note that existing methods still face significant challenges in finding the maximum EFR.
For HJ reachability, the primary issue is scalability. While recent approaches using neural networks and RL offer promising paths to mitigate this \citep{fisac2019bridging, yu2022reachability}, they often lack formal guarantees. Methods that combine neural network verification with HJ analysis aim to provide these guarantees but often face scalability limitations again \citep{yang2025scalable}, making it an open problem.
For CBFs and SIs, these approaches do not in general provide methods to compute the maximum EFR. Synthesizing a CBF or SI that certifies the maximum EFR typically requires guidance from other tools, such as HJ reachability, to shape or tune its parameters \citep{yang2023synthesizing}, since these methods inherently rely on an a priori characterization of the EFR.}

\chapter{Some Properties of Feasible Regions}
\label{feasible region property}
In this chapter, we formally explore the properties of feasible regions, as well as their relationships and how they are informative for virtual-time constraint design. Specifically, we will pay attention to containment relationships and equivalence relationships of different feasible regions. 

\section{Containment relationships of feasible regions}
\label{sec: containment relationship}
The definitions of feasible regions can be categorized by two dimensions: a) initial ones or endless ones and b) with a given policy or without a given policy. The four combinations correspond to 1) IFR $\Xinitg$, 2) IFR of a policy $\Xinitg (\pi)$, 3) EFR $\Xedlsg$, and 4) EFR of a policy $\Xedlsg (\pi)$. Besides these, we also need to pay attention to the constrained set $\Xcstr$ and the maximum EFR $\Xedlss$. Among these six regions, $\Xcstr$, $\Xinitg$, and $\Xinitg (\pi)$ are related only to the current step, while $\Xedlsg$, $\Xedlsg (\pi)$, and $\Xedlss$ are related to infinite future steps, and finding them are our ultimate goals. Since the latter three regions are related to long-term feasibility, their identification is much more difficult than that of the former. The following theorem gives some containment relationships between these six regions so that we can let those easy-to-identify regions cast some light on our goal regions.
\begin{theorem}[Feasible region containment]
    \label{thm: Feasible region containment}~
    \begin{enumerate}
        \item $\Xedlsg \subseteq \Xinitg \subseteq \Xcstr$.
        \item $\forall\pi, \Xedlsg (\pi) \subseteq \Xinitg (\pi) \subseteq \Xinitg$.
        \item $\forall\pi, \Xedlsg (\pi) \subseteq \Xedlss$.
	\item $\forall\pi$, if $\Xinitg (\pi) = \Xinitg$, then $\Xedlsg \subseteq \Xedlsg (\pi)$.
    \end{enumerate}
\end{theorem}

\begin{proof}~
\begin{enumerate}
    \item The first containment relationship $\Xedlsg \subseteq \Xinitg$ can be directly concluded from Definition \ref{def: endless feasibility of state}. For the second one, according to Definition \ref{def: initial feasibility of state}, 
    $$
    \forall x \in\Xinitg, g(x)=g(x_{0|t})\le0.
    $$
    Since $g(x_{0|t})$ is not weaker than $h(x_t)$, we have $h(x)\le0$, i.e., $\forall x \in\Xcstr$. Therefore, $\Xinitg \subseteq \Xcstr$.
    
    \item This holds by Definition \ref{def: endless feasibility of policy} and Definition \ref{def: initial feasibility of policy}.
    
    \item As stated in Definition \ref{def: endless feasibility of policy}, 
    $$
    \forall x \in\Xedlsg (\pi), \pi \in \Pi_\mathrm{init}^g(x_{t+i}),i=0,1,2,\dots,\infty,
    $$
    where $x_t=x$. Thus, $x_{t+i} \in \Xinitg (\pi)$. Since $\Xinitg (\pi) \subseteq \Xinitg$, it follows that $x_{t+i} \in \Xinitg$. Thus, $x_{t+i} \in \Xcstr$, i.e., $h(x_{t+i})\le 0$. Recall Theorem \ref{thm: maximum EFR equivalence} and we reach $\Xedlsg (\pi) \subseteq \Xedlss$.
    
    \item Since $\Xinitg = \Xinitg (\pi)$, it holds that $\pi$ is initially feasible in all states in $\Xinitg$. That is to say, $\pi \in \Pi_\mathrm{init}^g(x)$,$\forall x \in\Xinitg$. Starting from any $x \in\Xedlsg$ and choosing 
    $$
    \pi_{t+i}=\pi \in \Pi_\mathrm{init}^g(x_{t+i}),i=0,1,2,\dots,\infty.
    $$
    By definition of endlessly feasibility, we have 
    $$
    x_{t+i} \in \Xinitg = \Xinitg (\pi) \text{ i.e., } \pi \in \Pi_\mathrm{init}^g(x_{t+i}).
    $$
    This means $\pi$ is endlessly feasible in $x$, i.e., $\forall x \in\Xedlsg (\pi)$. Thus, $\Xedlsg \subseteq \Xedlsg (\pi)$.
\end{enumerate}

\end{proof}

Figure \ref{fig: containment relationships}(a)-(d) illustrate the four containment relationships in Theorem \ref{thm: Feasible region containment} respectively. The light-colored circles correspond to virtual-time states, and the dark-colored ones stand for real-time states.
% Figure \ref{fig: containment relationships}(a) illustrates the first relationship $\Xedlsg \subseteq \Xinitg \subseteq \Xcstr$, which holds by definition. % since a state must first be initially feasible, then it could be endlessly feasible. This is also the case for the second one. 
% Figure \ref{fig: containment relationships}(b) illustrates the second relationship $\Xinitg(\pi) \subseteq \Xinitg$. Since any state in $\Xinitg(\pi)$ is initially feasible at least under $\pi$, $\Xinitg(\pi)$ is a subset of $\Xinitg$.
% For a state in $\Xinitg$, there exists a policy satisfying the virtual-time constraint, but it is not necessarily the given policy $\pi$. Intuitively, given a policy $\pi$ poses a greater restriction on the IFR, making the IFR of a policy become a subset of IFR. 
The proof of the first two relationships can be easily understood by definition, while that of the last two needs some additional explanation.
% Figure \ref{fig: containment relationships}(c) is an illustration of $\Xedlsg (\pi) \subseteq \Xedlss$. 
The left of Figure \ref{fig: containment relationships}(c) shows the zero-sublevel sets of $h(x)$ and $g(x)$, which correspond to the maximum EFR $\Xedlss$ and the EFR of a policy $\Xedlsg (\pi)$ on the right, respectively. To see how $\{x|h(x)\le0\}$ is related to $\Xedlss$, recall Theorem \ref{thm: maximum EFR equivalence} which states that $x \in\Xedlss$ is equivalent to the existence of a policy rendering $x$ safe in real-time domain. Likewise, $x \in\Xedlsg (\pi)$ means that policy $\pi$ renders $x$ safe in virtual-time domain, without violating virtual-time constraint $g(x)\le0$. Naturally, a larger valid state set ($\{x|h(x)\le0\}$ versus $\left\{x|g(x)\le0\right\}$) yields a larger EFR.
% Figure \ref{fig: containment relationships}(d) is an illustration of the statement ``If $\Xinitg (\pi) = \Xinitg$, then $\Xedlsg \subseteq \Xedlsg (\pi)$.'' 
The left of Figure \ref{fig: containment relationships}(d) corresponds to the condition $\Xinitg (\pi) = \Xinitg$, with the red region representing both $\Xinitg$ and $\Xinitg (\pi)$. This condition implies that $\pi \in \Pi_\mathrm{init}^g(x), \forall x \in\Xinitg$. The right picture illustrates how $\Xedlsg$ is a subset of $\Xedlsg (\pi)$. For a state in $\Xedlsg$, every initially feasible policy at every time step must lead to an initially feasible state. As illustrated in the figure, policy $\pi^\prime$ (in purple) and $\pi^{\prime\prime}$ (in red) are two arbitrarily chosen policies, and $\pi$ is also a possible choice for $\pi^\prime$ and $\pi^{\prime\prime}$ (and all future policies). This guarantees that $\pi$ renders any state in $\Xedlsg$ endlessly feasible. 

\begin{figure}[htbp]
    \centering
    \subfloat[]{
        \includegraphics[trim=18 0 18 0, clip, width=0.35\linewidth]{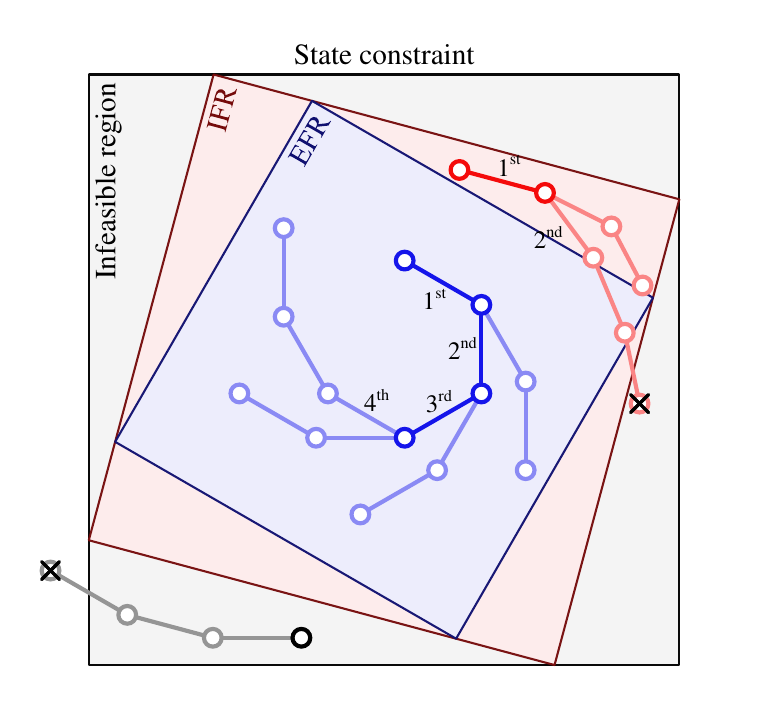}
    }
    \subfloat[]{
        \includegraphics[trim=42 0 -5 0, clip, width=0.35\linewidth]{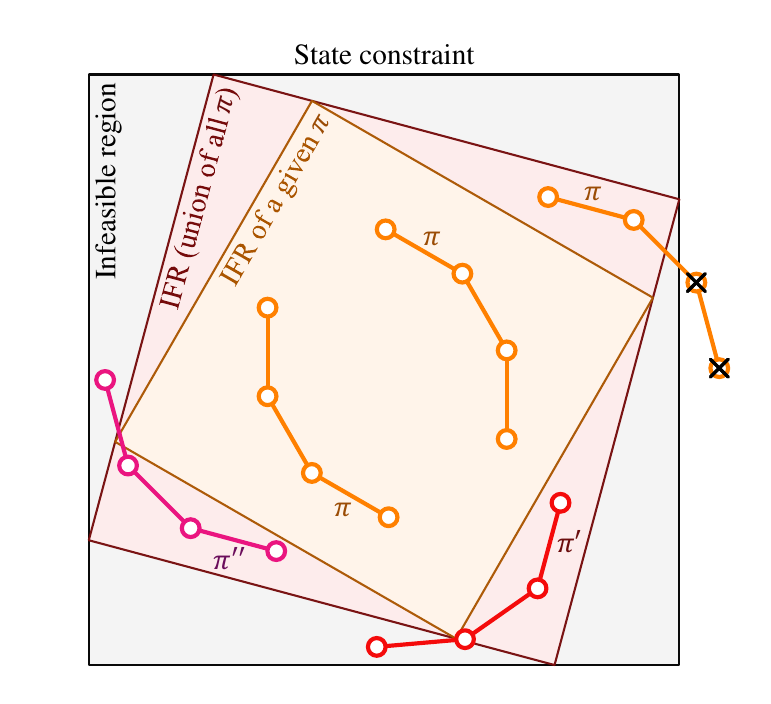}
    }
    \\
    \subfloat[]{
        \includegraphics[width=0.7\linewidth]{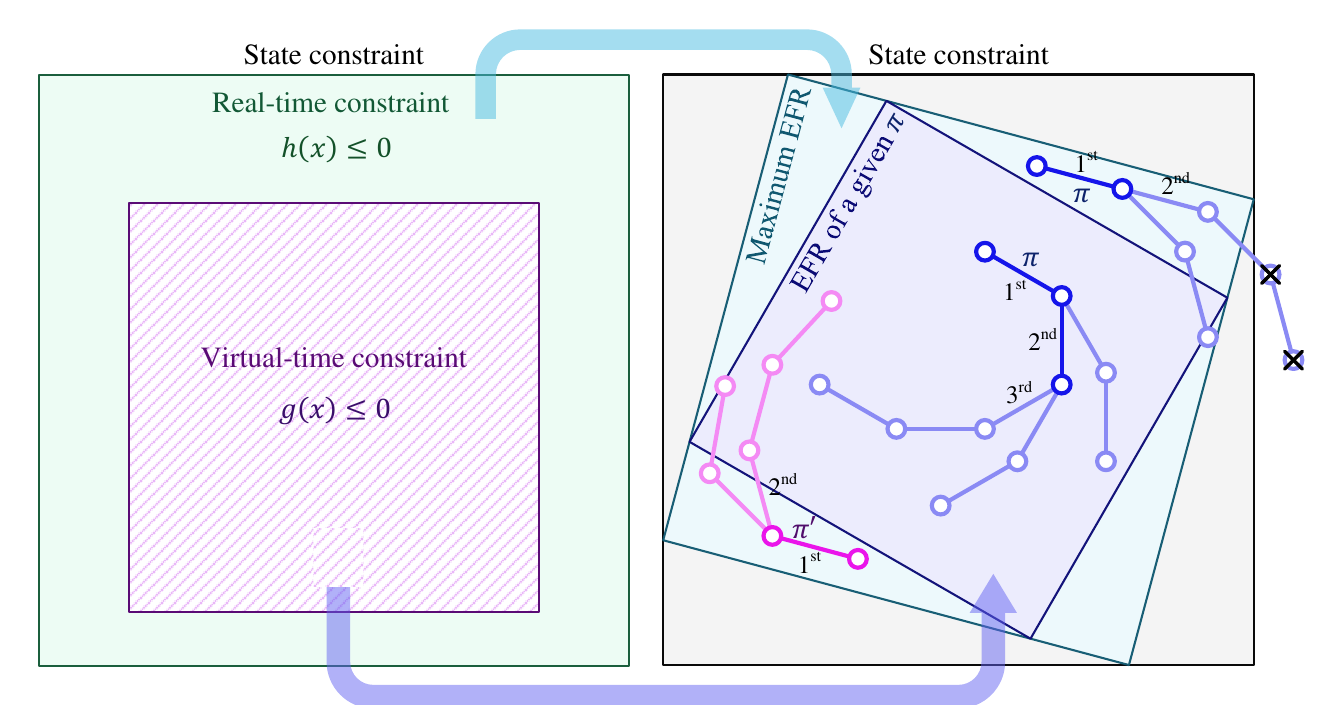}
    }
    \\
    \subfloat[]{
        \includegraphics[width=0.7\linewidth]{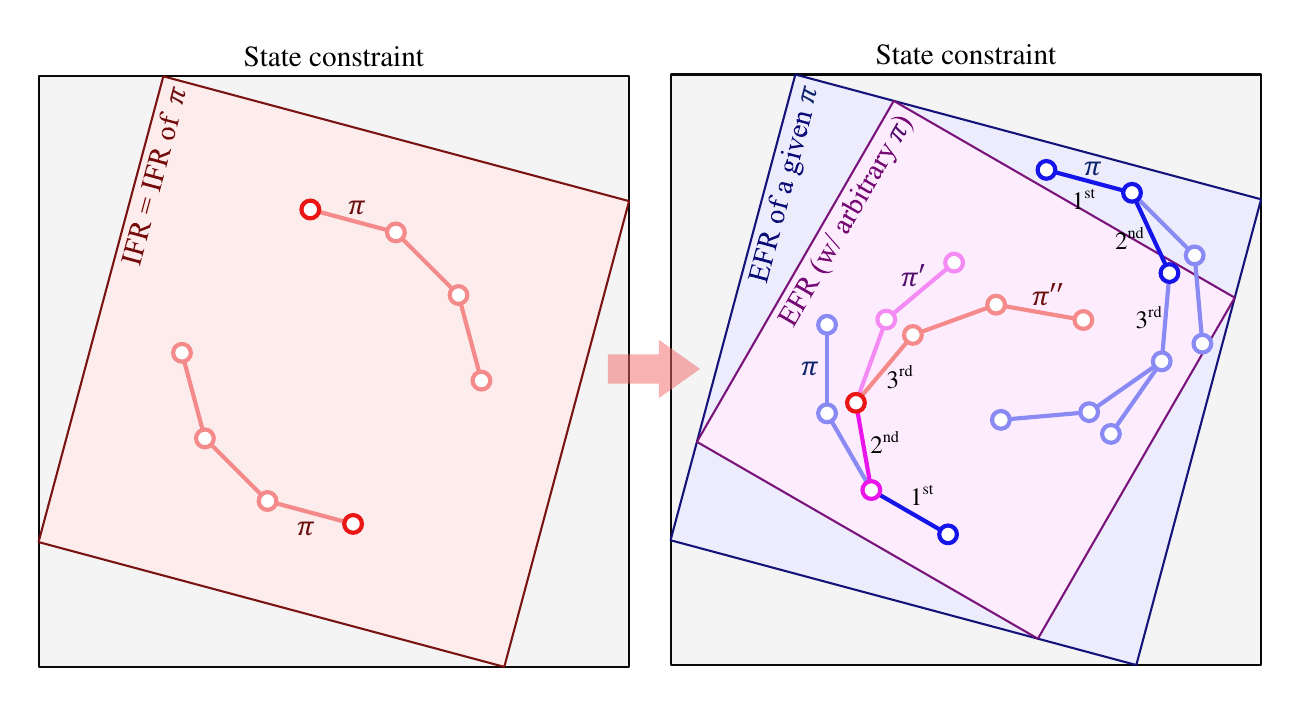}
    }
    \caption{Illustration of containment relationships.}
    \label{fig: containment relationships}
\end{figure}

Talking about $\Xinitg (\pi)$ and $\Xedlsg (\pi)$, the previous theorem only discusses the containment relationships for an arbitrary policy $\pi$. This policy $\pi$ may not be the solution to a virtual-time OCP. In safe RL, besides trying to identify the maximum EFR, one also seeks to find the optimal policy, i.e., $\pi^*$. Let us start from the IFR of $\pi^*$, i.e., $\Xinitg(\pi^*)$. For any state $x \in\Xinitg$, its corresponding virtual-time OCP must have feasible solutions, and $\pi^*$ is one of them. Therefore, $x$ must be initially feasible under $\pi^*$, i.e., $x \in\Xinitg(\pi^*)$. This leads to an important condition $\Xinitg(\pi^*) = \Xinitg$, and therefore a series of useful properties.

\begin{corollary}
\label{cor: containment relationship}~
\begin{enumerate}
    \item $\Xedlsg \subseteq \Xedlsg(\pi^*)\subseteq \Xinitg(\pi^*) = \Xinitg$. \label{cor: containment relationship-1}

    \item $\Xedlsg \subseteq \Xedlsg(\pi^*)\subseteq \Xedlss$. \label{cor: containment relationship-2}

    \item If $\Xedlsg = \Xinitg$, then $\Xedlsg = \Xedlsg(\pi^*) = \Xinitg(\pi^*) = \Xinitg \subseteq \Xedlss$. \label{cor: containment relationship-3}

    \item If $\Xedlsg = \Xedlss$, then $\Xedlsg = \Xedlsg(\pi^*) = \Xedlss \subseteq \Xinitg(\pi^*) = \Xinitg$. \label{cor: containment relationship-4}
\end{enumerate}
\end{corollary}

\begin{proof}~
\begin{enumerate}
    \item According to Theorem \ref{thm: Feasible region containment}(2), it holds that $\Xedlsg(\pi^*)\subseteq \Xinitg(\pi^*)\subseteq \Xinitg$. What left is to show that $\Xinitg \subseteq \Xinitg(\pi^*)$ and $\Xedlsg \subseteq \Xedlsg(\pi^*)$. Since
    $$
    \forall x \in\Xinitg, \pi^* \in \Pi_\mathrm{init}^g(x), \text{ i.e., } x \in\Xinitg(\pi^*),
    $$
    we have $\Xinitg \subseteq \Xinitg(\pi^*)$, which yields $\Xinitg(\pi^*) = \Xinitg$. Now we can choose $\pi=\pi^*$ in Theorem \ref{thm: Feasible region containment}(4) and conclude that $\Xedlsg \subseteq \Xedlsg(\pi^*)$. Thus, we have 
    $$
    \Xedlsg \subseteq \Xedlsg(\pi^*)\subseteq \Xinitg(\pi^*) = \Xinitg.
    $$
    \item This follows directly from Corollary \ref{cor: containment relationship}(1) and Theorem \ref{thm: Feasible region containment}(3).
    \item This follows directly from Corollary \ref{cor: containment relationship}(1) and (2).
    \item This follows directly from Corollary \ref{cor: containment relationship}(1) and (2).
\end{enumerate}
\end{proof}

Corollary \ref{cor: containment relationship}(1) reveals that the optimal policy's EFR is lower bounded by EFR and upper bounded by its own IFR. 
Corollary \ref{cor: containment relationship}(2) further clarifies that the optimal policy's EFR is a subset of the maximum EFR. 
Corollary \ref{cor: containment relationship}(3) gives a method for obtaining the optimal policy's EFR under a certain condition that IFR equals EFR. This condition also serves as a rule for designing virtual-time constraints. 
Corollary \ref{cor: containment relationship}(4) gives a condition when the optimal policy's EFR equals the maximum EFR: the EFR equals the maximum EFR. Combining with Corollary \ref{cor: containment relationship}(3), we arrive at an important rule for designing virtual-time constraints: $\Xedlsg = \Xinitg = \Xedlss$, i.e., EFR, IFR, and the maximum EFR are equal. Following this rule, we can guarantee that the optimal policy's EFR equals the maximum EFR.

\change{The sizes of feasible regions are largely determined by virtual constraints. As we mentioned in Chapter \ref{sec: infeasibility phenomenon}, weaker virtual-time constraints have worse feasibility in the long run. Here, we formally state and prove how the strength of virtual-time constraints affects the sizes of feasible regions. We begin with the following lemma.
\begin{lemma}
\label{lem: forward invariance of EFR}
For all $x\in\Xedlsg$ and $\pi\in\Pi_{\text{init}}^g(x)$, $x_{t+1}\in\Xedlsg$.
\end{lemma}
\begin{proof}
    This holds from the arbitrariness of policies and the infinity in Definition \ref{def: endless feasibility of state}.
\end{proof}
The following theorem states that a weaker virtual-time constraint leads to a smaller EFR.
\begin{theorem}
\label{thm: feasibility region monotonicity}
Given two virtual-time constraints $\text{VC}_1$ and $\text{VC}_2$, we term $\text{VC}_1$ is weaker than $\text{VC}_2$ if starting from any state $x$, a policy satisfying $\text{VC}_2$ also satisfies $\text{VC}_1$. Let $\text{VC}_h$ be a special virtual-time constraint that is exactly the same with the real-time constraint.
Suppose that $\text{VC}_1$ is weaker than $\text{VC}_2$, and $\text{VC}_2$ is weaker than $\text{VC}_h$, then $\Xedls^{\text{VC}_1} \subseteq \Xedls^{\text{VC}_2} \subseteq \Xedls^{\text{VC}_h} = \Xedlss = \Xinit^{\text{VC}_h} \subseteq \Xinit^{\text{VC}_2} \subseteq \Xinit^{\text{VC}_1}$.
% Then it holds that
%     \begin{enumerate}
%         \item Suppose that $\text{VC}_1$ is weaker than $\text{VC}_2$, and $\text{VC}_2$ is weaker than $\text{VC}_h$, then $\Xedls^{\text{VC}_1} \subseteq \Xedls^{\text{VC}_2} \subseteq \Xedls^{\text{VC}_h} = \Xedlss = \Xinit^{\text{VC}_h} \subseteq \Xinit^{\text{VC}_2} \subseteq \Xinit^{\text{VC}_1}$,
%         \item Suppose that $\text{VC}_h$ is weaker than $\text{VC}_1$, and $\text{VC}_1$ is weaker than $\text{VC}_2$, then $\Xinit^{\text{VC}_2} = \Xedls^{\text{VC}_2} \subseteq \Xinit^{\text{VC}_1} = \Xedls^{\text{VC}_1} \subseteq \Xinit^{\text{VC}_h} = \Xedls^{\text{VC}_h} = \Xedlss$.
%     \end{enumerate}
\end{theorem}
\begin{proof}
It holds by Definition \ref{def: initial feasibility of state}, Definition \ref{def: endless feasibility of state} and Theorem \ref{thm: maximum EFR equivalence} that $\Xedls^{\text{VC}_h} = \Xedlss = \Xinit^{\text{VC}_h} \subseteq \Xinit^{\text{VC}_2} \subseteq \Xinit^{\text{VC}_1}$. 
The left is to show that $\Xedls^{\text{VC}_1} \subseteq \Xedls^{\text{VC}_2} \subseteq \Xedls^{\text{VC}_h}$. We only proof the first containment and the second one follows the same procedure.
For all $x_t\in \Xedls^{\text{VC}_1}$, by Definition \ref{def: endless feasibility of state} and Lemma \ref{lem: forward invariance of EFR}, we have $\forall\pi_{t+i}\in\Pi_\mathrm{init}^{\text{VC}_1}(x_{t+i}),$ $x_{t+i+1}\in\Xedls^{\text{VC}_1}, i=0,1,2,\dots,\infty$.
Together with the facts that $\Pi_\mathrm{init}^{\text{VC}_2}(x_{t+i}) \subseteq \Pi_\mathrm{init}^{\text{VC}_1}(x_{t+i})$ and that $\Xedls^{\text{VC}_1} \subseteq \Xedlss \subseteq \Xinit^{\text{VC}_2}$, we can conclude that $x_t\in \Xinit^{\text{VC}_2}$ and $\forall\pi_{t+i}\in\Pi_\mathrm{init}^{\text{VC}_2}(x_{t+i}),$ $x_{t+i+1}\in\Xinit^{\text{VC}_2}, i=0,1,2,\dots,\infty$, which means $x_t\in \Xedls^{\text{VC}_2}$.
Thus, $\Xedls^{\text{VC}_1} \subseteq \Xedls^{\text{VC}_2}$.
% It holds by Definition \ref{def: initial feasibility of state}, Definition \ref{def: endless feasibility of state} and Theorem \ref{thm: maximum EFR equivalence} that $\Xinit^{\text{VC}_2} \subseteq \Xinit^{\text{VC}_1} \subseteq \Xinit^{\text{VC}_h} = \Xedls^{\text{VC}_h} = \Xedlss$. The left is to show that $\Xinit^{\text{VC}_1} = \Xedls^{\text{VC}_1}$ and $\Xinit^{\text{VC}_2} = \Xedls^{\text{VC}_2}$. We only proof the first equivalence. \\
% From Theorem \ref{thm: Feasible region containment}, we have $\Xedls^{\text{VC}_1} \subseteq \Xinit^{\text{VC}_1}$.
\end{proof}
\begin{remark}
It is worth distinguishing the case where virtual-time constraints are tighter than the real-time constraint $\text{VC}_h$. Suppose that $\text{VC}_h$ is weaker than $\text{VC}_1$, and $\text{VC}_1$ is weaker than $\text{VC}_2$ (i.e., $\text{VC}_2$ is the strictest). In this case, although $\Xinit^{\text{VC}_2} \subseteq \Xinit^{\text{VC}_1}$, the relationship $\Xedls^{\text{VC}_2} \subseteq \Xedls^{\text{VC}_1}$ does not necessarily hold. This is because $\Xedls$ requires the safety condition to hold for \textit{all} initially feasible policies. Since $\Pi_\mathrm{init}^{\text{VC}_2}(x) \subseteq \Pi_\mathrm{init}^{\text{VC}_1}(x)$, the fact that all policies in the smaller set $\Pi_\mathrm{init}^{\text{VC}_2}(x)$ maintain feasibility does not guarantee that the additional policies in $\Pi_\mathrm{init}^{\text{VC}_1}(x)$ also do so. Therefore, satisfaction of the endless feasibility condition under $\text{VC}_2$ does not imply satisfaction under $\text{VC}_1$.
\end{remark}
}

\section{Equivalence conditions of two feasible regions}
% The EFR is only determined by the dynamics model and its virtual-time constraint. 
The perfect design for a virtual-time constraint should result in an EFR that equals the maximum EFR, i.e., $\Xedlsg = \Xedlss$. However, $\Xedlsg$ is difficult to compute because it requires checking every initially feasible policy in every time step. To avoid this problem, we hope that $\Xedlsg = \Xinitg$ so that we only need to compute $\Xinitg$, which is much easier than computing $\Xedlsg$. In this case, if $\Xinitg = \Xedlss$, then $\Xedlsg(\pi^*) = \Xedlss$. The following theorem helps us understand what kind of condition $\Xedlsg = \Xinitg$ needs.

\begin{theorem}[Conditions for $\Xedlsg = \Xinitg$]
\label{thm: necessary cond and sufficient cond}~
\\
Necessary conditions:
\begin{enumerate}
    \item $\Xinitg \subseteq \Xedlss$.
    \item $\forall x_t \in \Xinitg$, $\exists u_t \in \mathcal{U},\st x_{t+1}\in\Xinitg$.
\end{enumerate}
Sufficient conditions (examples of virtual-time constraints):
\begin{enumerate}
    \item $h(x_{i|t})\le0,i=0,1,2,\dots,\infty$.
    \item $h(x_{0|t})\le0,x_{1|t}\in\Xedlss$.
\end{enumerate}
\end{theorem}

\begin{proof}
Necessary conditions:
\begin{enumerate}
    \item This is obvious by noting Definition \ref{def: maximum EFR}.
    \item Since $x_t\in\Xinitg$, it follows that $\Pi_\mathrm{init}^g (x_t )\neq\emptyset$. Considering that we also have $x_t\in\Xedlsg$, any $\pi \in\Pi_\mathrm{init}^g (x_t)$ and $u_t=\pi(x_t )$ will induce an initially feasible successive state.
\end{enumerate}

Sufficient conditions:
\begin{enumerate}
    \item $\forall x_t\in\Xinitg$ and $\pi \in\Pi_\mathrm{init}^g (x_t)$, we have 
    $$
    h(x_{t+i})=h(x_{i|t})\le0,i=0,1,\dots,\infty.
    $$
    Since this goes to infinity, $\pi$ must also be a feasible policy in $x_{t+1}$, i.e., $x_{t+1}\in\Xinitg$. By the arbitrariness of $x_t$ and $\pi$ , we can conclude that $x_t\in\Xedlsg$ and hence $\Xedlsg = \Xinitg$.
    \item $\forall x_t\in\Xinitg$ and $\pi \in\Pi_\mathrm{init}^g (x_t)$, we have $x_{t+1}=x_{1|t}\in\Xedlss$, and thus $x_{t+1}\in\Xinitg$. By the arbitrariness of $x_t$ and $\pi$, it holds that $\forall x \in\Xedlsg$. We can conclude that $\Xedlsg = \Xinitg$.
\end{enumerate}
\end{proof}

Theorem \ref{thm: necessary cond and sufficient cond} has important guiding significance for designing virtual-time constraints. 
% The necessary conditions should always be satisfied in constraint design. These conditions require that IFR is a subset of the maximum EFR and a control invariant set at the same time. Checking these conditions only involves computing IFR, which is much easier than computing EFR. 
One may argue that since the maximum EFR is unknown, the first necessary condition cannot be checked. In fact, we do not need to know the maximum EFR exactly. Instead, we can check if all states in IFR are endlessly feasible for a given policy. If this is true, the first necessary condition is satisfied.
\change{However, it remains intractable to check endless feasibility by definition in high-dimensional spaces, since we cannot traverse the state space. Instead, we can only prove or falsify endless feasibility by theoretical analysis. As we will demonstrate in Chapter \ref{constraint review}, this is achievable for established feasibility functions, such as CBF and Hamilton-Jacobi reachability function.}
The sufficient conditions in Theorem \ref{thm: necessary cond and sufficient cond} can be viewed as two examples of virtual-time constraints satisfying $\Xedlsg = \Xinitg$. When designing virtual-time constraints, one should try to satisfy a sufficient condition under the premise that the necessary conditions are satisfied. Note that other forms of sufficient conditions still exist that can ensure the equivalence of EFR and IFR. We can choose different ones according to the specific constrained OCP we aim to solve.

\chapter{Feasibility Function and Constraint Types}
\label{feasibility function}
Solving constrained OCPs involves finding both the optimal policy and its EFR. Whether the EFR equals the maximum EFR depends on the choice of virtual-time constraints. 
% The virtual-time constraints in previous sections, i.e., simplest constraint, pointwise constraint, and exponentially decaying constraint, are not suitable for finding the maximum EFR. Specifically, simplest constraint and exponentially decaying constraint do not satisfy the necessary conditions for the equivalence of EFR and IFR (Theorem \ref{thm: necessary cond and sufficient cond}). Pointwise constraint satisfies a sufficient equivalence condition but it contains an infinite number of constraints, which causes constrained OCPs computationally intractable. 
Therefore, constructing a proper virtual-time constraint is an indispensable task in constrained optimal control. 
In this chapter, we first introduce a tool for representing EFR called feasibility function, which helps us construct proper constraints that satisfy the equivalence conditions and enable us to find the maximum EFR.
Next, we review several commonly used virtual-time constraint formulations from the perspective of feasibility function. We point out that these virtual-time constraints are all constructed from feasibility functions and, therefore, satisfy desirable properties. Moreover, we use the tools introduced in the previous chapter to analyze the containment relationships of feasible regions under these virtual-time constraints.
% Clearly, feasibility function must be defined in the virtual-time domain in order to become a necessary element of virtual-time OCPs. To simplify narration, we often do not distinguish the differences of their notations in the real-time and virtual-time domains. The virtual-time subscript $x_{i|t}$ will be replaced with the real-time subscript $x_{t+i}$ in the rest of this chapter. Nevertheless, readers must keep in mind that constrained OCPs to be optimized must be defined in the virtual-time domain rather than the real-time domain.

\section{Feasibility function}
Feasibility functions are used for constructing virtual-time constraints and representing EFRs. A feasibility function is a mapping from the state space to a real number, i.e., $F:\mathcal{X}\to\mathbb{R}$. Through proper design, its zero-sublevel set $\mathrm{X}_F=\{x|F(x)\le0\}$ represents an EFR.
% The core of designing feasibility functions is to offer some kind of recursion properties so that by properly constructing the virtual-time constraint, $\mathrm{X}_F$ can become an EFR. 
There are two basic methods for designing feasibility functions. The first uses a control invariant set and the second resorts to constraint aggregation. A control invariant set is a region where there exists a policy that keeps all successive states still in this region. A constraint aggregation equivalently replaces the infinite-step real-time constraint with a single-step virtual-time constraint. These two types of feasibility functions result in two families of virtual-time constraints. 
% The former type restricts the state in a control invariant set while the latter type constrains an aggregation of infinite-step real-time constraints. 

For the first type, the zero-sublevel set of the feasibility function is chosen as a control invariant set, which is defined as follows.
\begin{definition}[Control invariant set (CIS)]
\label{def: Control invariant set}
A function $F:\mathcal{X}\to\mathbb{R}$ is a feasibility function if $\mathrm{X}_F\subseteq\Xcstr$ is a control invariant set, i.e., 
$$
\forall x\in\mathrm{X}_F, \exists u, \st x^\prime \in\mathrm{X}_F,
$$
where $x^\prime=f(x,u)$.
\end{definition}
\noindent When using this type of feasibility function, certain inequalities are used as virtual-time constraints, which require the next state to be kept in $\mathrm{X}_F$ starting from any state in $\mathrm{X}_F$. 
% That is to say, through the virtual-time constraint built upon it, the feasibility function $F$ equips the resulting policy with the recursive property we are looking for. But does this always work? In other words, how can we be sure that such a policy that satisfies the constraint and thus has the recursive property always exists? The answer is through the special property of the zero-sublevel set of $F$, i.e., control invariance.
Figure \ref{fig: CIS} gives an illustration of a feasibility function and its control invariant set. Any state $x$ inside the control invariant set can be still kept in this set under some action $u$, while a state $\tilde{x}$ outside the control invariant set may not be able to enter this set.

\begin{figure}
    \centering
    \includegraphics[width=0.45\linewidth]{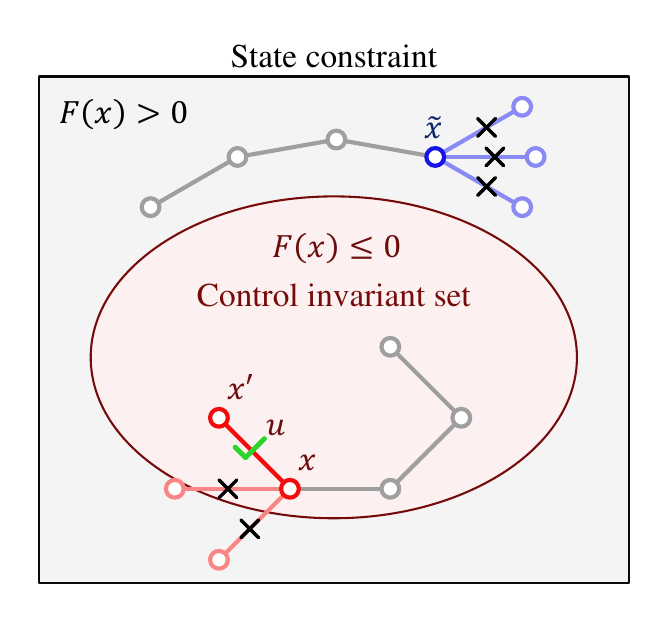}
    \caption{Feasibility function defined through control invariant set.}
    \label{fig: CIS}
\end{figure}

For the second type, the feasibility function is chosen as an aggregation function of the real-time constraint, which is defined as follows.
\begin{definition}[Constraint aggregation (CA)]
\label{def: Constraint aggregation}
A function $F:\mathcal{X}\to\mathbb{R}$ is a feasibility function if
$$
\exists\pi, \st \forall x\in\mathcal{X}, F(x)\le0\iff h(x_t)\le 0, t=0,1,\dots,\infty,
$$
where $x_0=x$, $x_{t+1}=f(x_t,\pi(x_t))$.
\end{definition}
\noindent Figure \ref{fig: CA} gives an illustration of constraint aggregation. For states $x_1$ and $x_3$, although they do not violate the state constraint at the current step, their future states leave the constrained set. Therefore, the feasibility function $F$ is positive on these two states. For state $x_2$, all of its successive states are inside the constrained set, thus its feasibility function is less than or equal to zero.
Definition \ref{def: Constraint aggregation} implies that $F$ is a function of a policy $\pi$. We can also denote it as $F^\pi$ to show its connection with $\pi$. 
% This equivalence compresses the infinitely many constraints into $F$ and hence equips it with a natural recursive property, i.e., once $F(x_t)\le0$, it follows that $F(x_{t+1})\le0$. However, this recursive property is not a free lunch. We have to make sure that the value of $F$ is practically available before it can be used to construct virtual-time constraints. 
One may recall that the state-value and action-value functions in RL also contain infinite future rewards, and we can compute them iteratively by bootstrapping. Similarly, feasibility functions of this type are usually formulated to satisfy self-consistency conditions so that they can be computed iteratively.

\change{The control invariant set method is most suitable for low-dimensional systems with known dynamics, as constructing such sets typically requires explicit knowledge of the system dynamics and remains computationally tractable only in lower dimensions. This approach provides strict, provable safety guarantees, but it may sacrifice optimality—especially when the resulting invariant set is overly conservative. It is thus well-suited for applications where safety is critical and some degree of optimality can be compromised.}

\change{In contrast, the constraint aggregation method can be applied to higher-dimensional systems and those with unknown dynamics. Feasibility functions designed via constraint aggregation generally satisfy self-consistency conditions, enabling them to be solved via fixed-point iteration. When parameterized using neural networks, such functions can be optimized iteratively via dynamic programming, much like value functions in reinforcement learning. A key advantage of this approach is its theoretical ability to recover the maximum feasible region in any system. However, this scalability comes at a cost: the zero-sublevel sets of these feasibility functions are not necessarily forward invariant, which implies that strong safety guarantees may not always be available. This behavior mainly stems from approximation and generalization errors when feasibility functions are represented by neural networks. Imperfect training, limited data, and finite network capacity prevent exact recovery of the true feasibility function.}

\change{The feasibility function has a natural connection to the value function defining the infinite-horizon backward avoidable tube in Hamilton-Jacobi reachability analysis. This reachability value function, whose zero-sublevel set is the maximum EFR, is a quintessential example of a feasibility function. However, the concept of a feasibility function is more general. Its zero-sublevel set is not restricted to the maximum EFR but can represent the infinite-horizon constraint-satisfying set of any policy. Moreover, its functional form is not confined to the specific ``maximum-over-time'' structure of the reachability solution. As we will show in the following two sections, the Hamilton-Jacobi value function is a specific instance within a broader and more flexible class of the feasibility function.}

\begin{figure}
    \centering
    \includegraphics[width=0.45\linewidth]{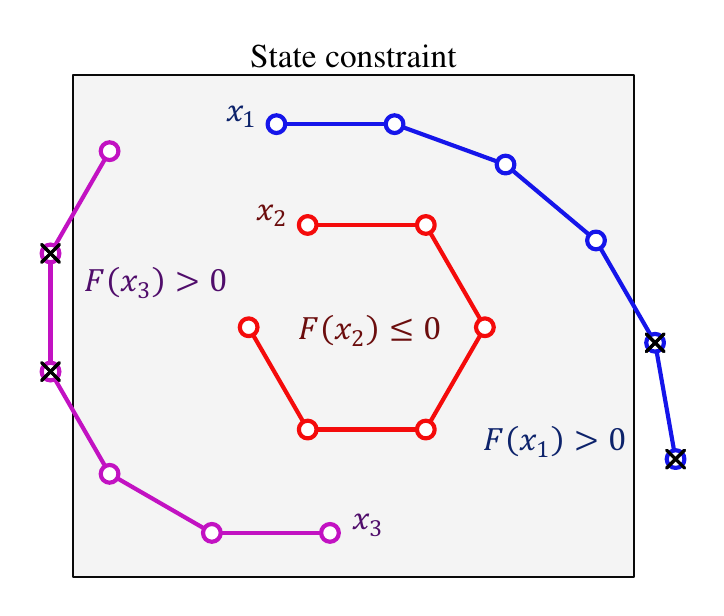}
    \caption{Feasibility function defined through constraint aggregation.}
    \label{fig: CA}
\end{figure}

\section{Type I: control invariant set}
This type of constraint restricts the state in a control invariant set (CIS) represented by the zero-sublevel set of a feasibility function. There are mainly two kinds of feasibility functions of this type: a) CBF and b) SI. They differ in the inequalities that guarantee the invariant property of their zero-sublevel sets.

\subsection{Control barrier function}
\label{sec: CBF}
CBF is a widely used feasibility function for synthesizing safe controllers in constrained OCPs \citep{ames2019control, ames2016control, cheng2019end, robey2020learning}. The control invariance of the zero-sublevel set of CBF is guaranteed by the fact that the time derivative of CBF is always non-positive on the boundary of this set. 
% As a state in the zero-sublevel set approaches the boundary, its CBF value may increase at first but will finally stop at some value not greater than zero. Therefore, the CBF will never take a value above zero as long as the initial state is in its zero-sublevel set. 
The formal definition of CBF in a discrete-time system is as follows.

\begin{definition}[Control barrier function]
\label{def: CBF}
A function $B:\mathcal{X}\to\mathbb{R}$ is a control barrier function (CBF) if it satisfies:
\begin{enumerate}
    \item[1)] $\forall x\in\mathcal{X}$, if $h(x)>0$, then $B(x)>0$,
    \item[2)] $\forall x\in\mathcal{X}, \exists u\in\mathcal{U}, \st B(x^\prime)-B(x)\le-\alpha(B(x))$,
\end{enumerate}
where $x^\prime=f(x,u)$ and $\alpha:\R\to\R$ is a strictly increasing function and $|\alpha(z)|\le|z|$ for all $z\in\R$.
\end{definition}

If a function satisfies the above two properties, it can be proved that its zero-sublevel set is a CIS \citep{ames2019control}. Therefore, according to Definition \ref{def: Control invariant set}, a CBF is a feasibility function itself, i.e., 
$$F(x)=B(x).$$
Property 2) in Definition \ref{def: CBF} is the most critical element for ensuring control invariance. 
% This property tells us that there exists an action such that the increment of CBF from the current step to the next step is upper bounded. 
When solving a constrained OCP, we must take it as a virtual-time constraint to ensure that the action satisfies this property. This leads to a one-step constraint imposed on the second state in virtual-time domain, i.e., $x_{1|t}$. In addition, $B(\cdot)\le0$ is imposed on the first state $x_{0|t}$. This is because we need to restrict the domain of the OCP to the CIS. This first-step virtual-time constraint is not weaker than the real-time state constraint because property 1) in Definition \ref{def: CBF} tells us that $B(\cdot)\le0$ implies $h(\cdot)\le0$. Combining these two constraints, we have
\begin{equation}
\label{CBF constraint}
\begin{aligned}
    g(x_{0|t})&=B(x_{0|t})\le0, \\
    g(x_{1|t})&=B(x_{1|t})-B(x_{0|t})+\alpha(B(x_{0|t}))\le0.
\end{aligned}
\end{equation}
The above constraints guarantee that the next state $x_{1|t}$ is still in the zero-sublevel set. If all states in the zero-sublevel set satisfy this constraint, the set is forward invariant, and the real-time state constraint is always satisfied.

In literature, a notion of high order CBF (HOCBF) \citep{xiao2019control, xiao2021high} is used to deal with high relative degree constraints. The relative degree of a function refers to its lowest order of time derivative where the action explicitly appears. When the constraint function $h$ has a high relative degree, directly choosing $h$ as a CBF has no effect on the action because constraint \eqref{CBF constraint} is irrelevant to the action. This necessitates HOCBF, which considers time derivatives of $h$ where the action explicitly appears \citep{xiao2019control}. Although not explicitly stated, Definition \ref{def: CBF} is compatible with HOCBF because $B$ does not have to be chosen as $h$ and can be constructed as some function that makes the action explicitly appear in constraint \eqref{CBF constraint}. 

Let us analyze the IFR and EFR under constraint \eqref{CBF constraint}. 
\change{
For IFR, we examine each state $x$ whether the corresponding virtual-time OCP starting from it has a solution. 
Regarding the first inequality in \eqref{CBF constraint}, since $x_{0|t}=x$, all states in the zero-sublevel set of $B$ naturally satisfy this condition. For notation simplicity, we denote this set as $\rmX_B=\{x\in\mathcal{X}|B(x)\le0\}$. 
Regarding the second inequality in \eqref{CBF constraint}, Definition \ref{def: CBF} guarantees that for every state $x_{0|t}=x\in\mathcal{X}$, there exists an action such that the resulting state $x_{1|t}$ satisfies this inequality. 
Thus, we can conclude that
$$\Xinit^g=\rmX_B.$$
}

For EFR, we look for the state where the OCPs in all successive states have solutions. Since the IFR is $\rmX_B$, the EFR must be a subset of $\rmX_B$. For any state in $\rmX_B$, its successive state is still in $\rmX_B$ if we take an action that satisfies \eqref{CBF constraint}. Therefore, the state can always be kept in $\rmX_B$ as long as an initially feasible policy is applied at every time step. We can conclude that the EFR equals $\rmX_B$, i.e.,
$$\Xedls^g=\rmX_B.$$

We can see that EFR equals IFR under constraint \eqref{CBF constraint}. This shows the benefit of feasibility function for designing virtual-time constraints, i.e., it ensures the equivalence of EFR and IFR. 
% This equivalence enables us to know the size of EFR by examining that of IFR. 
In a constrained OCP, we always aim to construct an EFR that equals the maximum EFR. 
% However, the size of EFR is difficult to obtain because it involves checking the constraint satisfaction in an infinite horizon. 
With the equivalence of EFR and IFR, we can achieve this by constructing an IFR that equals the maximum EFR. 
% For a virtual-time constraint designed by CBF, its IFR is the zero-sublevel set of the CBF, i.e., $\rmX_B$. 
The size of IFR depends on the design of CBF. 
% A better design results in a larger $\rmX_B$, which, in the best case, equals the maximum EFR. 
For a general CBF, we can only say that $\rmX_B$ is a subset of the maximum EFR, i.e.,
$$\rmX_B\subseteq\Xedls^*.$$
Therefore, the relationship among the feasible regions is as follows:
$$\Xinit^g=\Xedls^g\subseteq\Xedls^*.$$

In conclusion, a CBF is a feasibility function that represents EFR using a CIS, which is its zero-sublevel set. A CBF satisfies two properties: the first requires that the zero-sublevel set is constraint-satisfying, and the second ensures that the zero-sublevel set is control invariant. With a CBF, we can construct a two-step virtual-time constraint under which the IFR and EFR both equal the zero-sublevel set. However, this set may not equal the maximum EFR but only be a subset of it depending on the design of CBF. In systems with high-dimensional state and action spaces or systems without analytical models, it remains challenging to design a CBF with a CIS equal to or close to the maximum EFR.

\subsection{Safety index}
SI, originally proposed by \citet{liu2014control}, is another feasibility function that represents a CIS by its zero-sublevel set, which is called safe set. SI ensures control invariance in a similar way to CBF. In particular, it requires that once the system deviates from the safe set, the control policy will pull it back. This is achieved by setting the time derivative of SI to be negative outside the safe set. SI differs from CBF in that it explicitly considers the relative degree of system in construction. 
It is constructed as the linear combination of the constraint function, its nonlinear substitution, and its time derivatives. The formal definition of SI is as follows.

\begin{definition}[Safety index]
\label{def: safety index}
A function $\phi:\rmX\to\R$ is a safety index (SI) if it satisfies the following three properties:
\begin{enumerate}
    \item[1)] $\phi=h+k_1\dot{h}+\dots+k_n h^{(n)}$, where $k_1,\dots,k_n\in\R$,
    \item[2)] All roots of $1+k_1 s+k_2 s^2+\dots+k_n s^n=0$ are on the negative real line,
    \item[3)] The relative degree from $h^{(n)}$ to $u$ is one,
\end{enumerate}
where $n+1$ is the system order. Moreover, suppose $h^*$ defines the same set as $h$ does, i.e., $\{x\in\calX|h^*(x)\le0\}=\Xcstr$, then
$$\phi^*=\phi-h+h^*$$
is also an SI.
\end{definition}

\change{Property 3) ensures the SI has a relative degree of one. This means that the action has a direct influence on the SI. With the above three properties, it can be proved that for unbounded action space, the safe set is control invariant \citep{liu2014control}. For bounded action space $\calU$, a fourth condition is needed: the parameters $k_i, i=1,2,\dots,n$ are chosen such that there exist $u\in\calU$ for $\dot{\phi}\le0$ when $\phi=0$, which is a necessary and sufficient condition for control invariance of the safe set. The use of $h^*$ is to introduce nonlinearity for shaping the boundary of the feasible region.} With an SI, we can construct the following feasibility function:
$$
F(x)=\max\{\phi(x),h(x)\}.
$$
Here, including $h$ in the maximum operator is to ensure that the feasible region is a subset of the constrained set. When $h^*=h$ (i.e., the SI is linearly defined), the feasible region can be represented by a set of inequalities similar to that of HOCBF \citep{xiao2019control}. When $h^*\neq h$ (i.e., the SI is nonlinearly defined) and $n=1$ (i.e., the dynamic system is second order), the feasible region equals the zero-sublevel set of $F$. For $h^*\neq h$ and $n>1$, there is no conclusion regarding the exact representation of the feasible region with $F$.

The control invariance of the safe set is guaranteed by the inequality $\dot{\phi}\le0$ when $\phi=0$. This is a differential inequality and only applies to continuous-time systems. For discrete-time systems, we can convert it to a difference inequality through a forward Euler discretization:
% To make it applicable to discrete-time systems, we must convert it to a difference inequality. An intuitive conversion is to replace the derivative of $\phi$ with its discretization obtained by, for example, the forward Euler method, which results in the following inequality,
% $\phi(x_{1|t})-\phi(x_{0|t})\le0$ when $\phi(x_{0|t})=0$.
% However, this inequality only constrains the value of $\phi$ when $\phi(x_{0|t})$ is exactly zero, which is not enough to guarantee the control invariance of safe set in discrete-time systems. This is because even if $\phi(x_{0|t})<0$, it is possible that $\phi$ becomes positive after only one step at $x_{1|t}$. Therefore, we must also constrain the value of $\phi$ when $\phi(x_{0|t})<0$. Specifically, we require that $\phi(x_{1|t})<0$ so that $x_{1|t}$ is still in the safe set. When $\phi(x_{0|t})>0$, $x_{0|t}$ is already out of the safe set. In this case, $\phi$ should decrease fast enough so that the state can return to the safe set quickly. Specifically, we require that $\phi(x_{1|t})\le\phi(x_{0|t})-\eta$, where $\eta$ is a positive real number. 
% This inequality can be further relaxed to $\phi(x_{1|t})\le\max\bbr{\phi(x_{0|t})-\eta,0}$ because once $\phi(x_{1|t})\le0$, the state already enters the safe set. The above two inequalities can be written as the following inequality, 
$$\phi(x_{1|t})-\max\bbr{\phi(x_{0|t})-\eta,0}\le0,$$
% which acts as a virtual-time constraint of SI. 
Similar to the virtual-time constraint of CBF, SI also needs a first-step constraint that restricts the state in the safe set. Combining these two constraints, we arrive at the virtual-time constraint of SI,
\begin{equation}
\label{safety index constraint}
\begin{aligned}
    g(x_{0|t})&=\max\{\phi(x_{0|t}), h(x_{0|t})\}\le0, \\
    g(x_{1|t})&=\phi(x_{1|t})-\max\bbr{\phi(x_{0|t})-\eta,0}\le0.
\end{aligned}
\end{equation}
The first constraint above is not weaker than the real-time constraint. The second constraint requires that the SI decreases every time step by at least a value of $\eta$ when it is above zero. When the SI is below zero, it is required not to become positive. With a policy satisfying constraint \eqref{safety index constraint}, the zero-sublevel set of SI is forward invariant. 

The relationship between feasible regions under SI constraint is similar to that of CBF, so we omit the analysis here.
One may discover that SI is closely related to CBF in both feasibility function formulation and virtual-time constraint design. This is not a coincidence but rather originates from the construction of CBFs. While CBF guarantees control invariance of its zero-sublevel set by definition, it does not provide any guidance for its construction. HOCBF provides one design rule for obtaining a valid CBF, and SI can be considered as another.
However, we have seen that a limitation of these CIS-based methods is that they may not be able to obtain the maximum EFR. In contrast, the following feasibility function defined by constraint aggregation can lead to the maximum EFR.

\section{Type II: constraint aggregation}
Constraints \change{based on} constraint aggregation (CA) replace the infinite-step real-time constraints with a single-step virtual-time constraint by an aggregation function. There are mainly three kinds of CA-based feasibility functions: cost value function, Hamilton-Jacobi reachability function, and constraint decay function. Each kind of feasibility function further corresponds to two types of virtual-time constraints. One uses the feasibility function of the current policy, i.e., the policy to be optimized, while the other uses the feasibility function of a fixed policy, e.g., the policy from the last iteration.

\change{An advantage of this type of feasibility functions is that they provide a principled way to evaluate the feasible region of a given policy without a known system model. These feasibility functions are typically associated with a self-consistency condition that depends only on the policy to evaluate and on sampled transitions from the environment. This equation induces a contraction mapping, so its fixed point can be learned via standard RL temporal-difference methods without requiring a system model \citep{fisac2019bridging, yang2023feasible}. This enables us to solve these feasibility functions by fixed point iteration, Once the feasibility function is solved, its zero-sublevel set represents the feasible region of the policy.}

\subsection{Cost value function}
Cost value function (CVF) originates from a common formulation for safe RL, constrained Markov decision process (CMDP) \citep{altman1999constrained}, which augments a standard MDP with a cost function. Being the expected cumulative costs, CVF shares the same mathematical form and, therefore, the same properties, with state-value function $V^\pi (x)$ in RL. With a slight effort to define a cost function $c(x)$ with the constraint function $h(x)$, CVF can be well adapted to constrained OCPs. The definition of CVF is as follows.

\begin{definition}[Cost value function]
\label{def: cost value function}
For a constrained OCP with $h(x)$ as a constraint function, $x_t,t=0,1,2,\dots,\infty$ donating the state trajectory starting from state $x$ under policy $\pi$, the cost value function $F^\pi:\mathcal{X}\to\R$ is:
$$
F^\pi (x)=\sum_{t=0}^\infty\gamma^t c(x_t),
$$
where $c(x)=\pmb 1[h(x)>0]$ is the cost signal.
\end{definition}

To see that CVF is a valid feasibility function, note that $c(x)$ is non-negative for all $x$, so that $F^\pi (x)\le0$ is sufficient and necessary for $c(x_t )=0,\forall t=0,1,\dots,\infty$. The close relationship between CVF and state-value function implies some common properties. Specifically, the self-consistency condition and Bellman equation, which play important roles in solving state-value functions, have their counterparts for CVFs. For the sake of distinction, they are termed risky self-consistency condition: 
$$F^\pi (x)=c(x)+\gamma F^\pi (x^\prime),$$
and risky Bellman equation:
$$F^*(x)=c(x)+\gamma\min_u F^*(x^\prime).$$
Similar to the case for state-value function, the right-hand sides of the above two equations are contraction operators on a complete metric space. Hence, a CVF can be solved by fixed-point iteration, just like a state-value function. 

There are two ways to construct virtual-time constraints with a CVF. The first way is to constrain only the first state:
\begin{equation}
\label{CVF constraint 1}
    g(x_{0|t})=F^\pi (x_{0|t})\le0,
\end{equation}
and $x_{i|t},i=1,\dots,n$ are unconstrained. Here, the superscript $\pi$ of $F$ is the policy to be optimized. The dependence of constraint \eqref{CVF constraint 1} on $\pi$ is established through $F^\pi$, which is a function of $\pi$. Technically speaking, for a given $x_{0|t}$, $F^\pi (x_{0|t})$ is a functional of $\pi$. In practice, $F^\pi$ is difficult to compute directly and is usually approximated by importance sampling of trajectories collected by another policy, similar to the way of approximating value function in on-policy RL algorithms. 
% This form of virtual-time constraint is compatible with direct RL/ADP methods for constrained OCPs so that it can serve as an equivalent substitute for the real-time constraint. Through the aggregation of constraints introduced by the summation function, the horizon of virtual-time constraint reduces to only one step, the current step. 
The IFR under constraint \eqref{CVF constraint 1} is by definition
$$
\Xinit^g=\{x|\exists\pi,\st F^\pi (x)\le0\}.
$$
% No matter which initially feasible policy we apply at the current step, we will always reach a state in the IFR. 
Starting from $\Xinit^g$, as long as we apply an initially feasible policy, the successive states will always be in $\Xinit^g$.
This is because $F^\pi (x)\le0$ is already an infinite-horizon constraint-satisfying property and naturally holds for all future steps. 
% this policy must still be initially feasible in the next state, which is guaranteed by the fact that the CVF's being less than or equal to zero at a single state renders all future states safe. 
Therefore, the EFR is still
$$
\Xedls^g=\{x|\exists\pi,\st F^\pi (x)\le0\}.
$$
What's more, an excellent property of constraint \eqref{CVF constraint 1} is that
$$
\Xedls^g=\Xinit^g=\Xedls^*.
$$
% The first equation is obvious from the above analysis. 
To understand this, recall Theorem \ref{thm: maximum EFR equivalence} and that $F^\pi (x)\le0$ is equivalent to $h(x_t)\le0,t=0,1,\dots,\infty$. The equivalence of these three regions is satisfying, with which we can easily check endless feasibility by initial feasibility, and a constructed EFR is automatically the maximum.

The second way of constructing a virtual-time constraint is a little bit complex, involving both the first and the second states:
\begin{equation}
\label{CVF constraint 2}
    g(x_{i|t} )=F^k (x_{i|t} )\le0,i=0,1,
\end{equation}
and $x_{i|t},i=2,\dots,n$ are unconstrained. Unlike the former way, here $F^k$ is the CVF of a fixed policy $\pi_k$, which is irrelevant to the policy $\pi$ to be optimized. The dependence of constraint \eqref{CVF constraint 2} on $\pi$ is established through $x_{1|t}$, which is induced by $\pi$. 
Under this constraint, any $x$ satisfying $F^k (x)\le0$ is initially feasible, with $\pi_k$ being an initially feasible policy. Therefore,
% for a state $x$ to be initially feasible, the first requirement is that $g(x_{0|t})=F^k (x_{0|t})\le0$, where $x_{0|t}=x$. This implies that if we take $\pi=\pi_k$, we will have $F^k (x_{1|t})\le0$. Consequently, 
$$
\Xinit^g=\{x|F^k (x)\le0\}.
$$
Because of the same reason analyzed above, we still have
$$
\Xedls^g=\{x|F^k (x)\le0\}.
$$
Unfortunately, under constraint \eqref{CVF constraint 2}, the IFR and EFR depend on the policy $\pi_k$ and may be only subsets of the maximum EFR, i.e.,
$$
\Xedls^g=\Xinit^g\subseteq\Xedls^*.
$$
However, by properly choosing $\pi_{k+1}$ under the guide of $F^{\pi_k}(x)$, we can achieve a sequence of monotonically expanding EFRs, which converge to the maximum EFR \citep{yang2023feasible}.

\subsection{Hamilton-Jacobi reachability function}
\label{sec: HJR}
The Hamilton-Jacobi (HJ) reachability analysis, a technique from robust optimal control theory, is used to guarantee constraint satisfaction in a rigorous way. HJ reachability computes a backward reachable set of a system, i.e., the set of states from which trajectories can reach some given target set \citep{mitchell2005time, bansal2017hamilton}. For feasibility analysis, the target set is chosen as the \change{complement of the constrained set}, and the backward reachable set becomes the infeasible region.
HJ reachability function aggregates infinite-step constraints by taking the maximum of them. In other words, it considers the worst-case constraint violation in the entire trajectory or the ``closest distance'' to the constraint boundary if no violation is going to happen. The definition of HJ reachability function is as follows.

\begin{definition}[HJ reachability function]
\label{def: HJR}
For a constrained OCP with $h(x)$ as a constraint function, $x_t,t=0,1,2,\dots,\infty$ donating the state trajectory starting from state $x$ under policy $\pi$, the HJ reachability function $F^\pi: \calX\to\R$ is:
$$F^\pi(x)=\max_t h(x_t).$$
\end{definition}

% Figure \ref{fig: HJR} shows the relationship between $h(x)$ and $F^\pi(x)$. 
% The constraint function $h(x)$ indicates the safety of the current state. As its extension to the whole temporal domain, $F^\pi(x)$ does not describe the risk level of the current state, but the worst case in the future. It is the maximizer that attempts to find the most dangerous constraint point along the whole state trajectory. Therefore, 
It follows directly from the above definition that $F^\pi (x)\le0$ is equivalent to $h(x_t)\le0$ for $t=0,1,\dots,\infty$. 
Similar to CVF, HJ reachability function naturally satisfies a risky self-consistency condition:
$$F^\pi (x)=\max\{h(x), F^\pi(x^\prime)\},$$
and the optimal HJ reachability function satisfies a risky Bellman equation:
$$F^* (x)=\max\{h(x), \min_u F^*(x^\prime)\}.$$
In practical HJ reachability computation, a discount factor is introduced to the right-hand side of the above two equations to obtain contraction mappings for fixed point iteration \citep{fisac2019bridging}.

Similar to CVF, there are two ways to construct a virtual-time constraint with HJ reachability function. One is to use a first-step constraint:
\begin{equation}
\label{HJR constraint 1}
    g(x_{0|t})=F^\pi (x_{0|t})\le0,
\end{equation}
% and $x_{i|t},i=1,\dots,n$ are unconstrained. By aggregating constraints using the maximum function, the virtual-time constraint is replaced by the first-step constraint of HJ reachability function.
and the other is to use a two-step constraint:
\begin{equation}
\label{HJR constraint 2}
    g(x_{i|t})=F^k (x_{i|t})\le0,i=0,1,
\end{equation}
% and $x_{i|t},i=2,\dots,n$ are unconstrained. 
We omit the analysis of their corresponding feasible regions and containment relationships since it is similar to the case of CVF.

\subsection{Constraint decay function}
Constraint decay function (CDF) uses the remaining steps from the current state to the first constraint violation to aggregate constraints \citep{yang2024synthesizing}. Its basic idea is that if a state is infeasible, it will lead to a constraint violation in finite steps. 
% Otherwise, no constraint violation will happen in finite steps, i.e., the number of steps to constraint violation is infinite. 
% CDF is constructed as an exponential function with the number of steps to constraint violation as the exponent and a real number between zero and one as the base, which is 
CDF is formally defined as follows.

\begin{definition}[Constraint decay function]
\label{def: CDF}
For a constrained OCP, the constrained decay function $F^\pi:\calX\to\R$ is 
$$F^\pi (x)=\gamma^N(x),$$
where $0<\gamma<1$ is the discount factor, and $N$ is the number of steps to the first constraint violation starting from $x$ under $\pi$.
\end{definition}

As $N$ increases from zero to infinity, $F^\pi (x)$ decreases from one to zero. A larger value of CDF indicates that the current state is more dangerous in the sense that it is closer to constraint violation. A value of one means that the current state already violates the constraint, and a value of zero means that the constraint will never be violated. In other words, $F^\pi (x)\le0$ is a necessary and sufficient condition for infinite-horizon constraint satisfaction. Therefore, CDF is a valid CA-based feasibility function. CDF is similar to a notion called safety critic in safe RL literature \citep{thananjeyan2021recovery}. Their difference is that they are used in different kinds of systems and for different purposes. The safety critic is used in stochastic systems for estimating the probability of future constraint violation. In contrast, CDF is used in deterministic systems as a feasibility function for constructing virtual-time constraints and representing feasible regions.

Similar to other CA-based feasibility functions, CDF naturally satisfies a risky self-consistency condition:
$$F^\pi (x)=c(x)+(1-c(x))\gamma F^\pi (x^\prime),$$
and the optimal CDF satisfies a risky Bellman equation:
$$F^* (x)=c(x)+(1-c(x))\gamma\min_u F^*(x^\prime).$$
There are also two ways to construct virtual-time constraints with CDF. The first is a single-step one:
\begin{equation}
\label{CDF constraint 1}
    g(x_{0|t})=F^\pi (x_{0|t})\le0,
\end{equation}
and the second is a two-step one:
\begin{equation}
\label{CDF constraint 2}
    g(x_{i|t})=F^k (x_{i|t})\le0,i=0,1.
\end{equation}
Their corresponding feasible regions and containment relationships are similar to other CA-based feasibility functions and are thus omitted here.

\chapter{Review of Constraint Formulations}
\label{constraint review}
In this chapter, we review several commonly used virtual-time constraint formulations. In essence, they are applications of different feasibility functions.
\change{Existing surveys have extensively covered specific types in isolation, such as the theory and applications of CBFs \citep{ames2019control} or HJ reachability analysis \citep{bansal2017hamilton, ganai2024hamilton}. Meanwhile, other reviews on safe RL \citep{garcia2015comprehensive, gu2024review} and safe control \citep{brunke2022safe} discuss safety constraints from a broader, algorithmic perspective, often without dissecting the unified mathematical essence of the underlying feasibility functions. In contrast, we provide a synthesizing framework that moves beyond these isolated or high-level discussions. We categorize the constraint formulations according to their core feasibility function (CIS or CA), demonstrating how they all conform to the same general mathematical structure and feasibility analysis presented in our unified framework.}

\section{Control invariant set}
\subsection{Control barrier function}
The concept of CBF originated from the field of control theory and was first proposed by \citet{ames2014control} for continuous-time systems. It was later extended to discrete-time systems by \citet{agrawal2017discrete}. 
A large category of works solves OCPs with CBF constraints through online optimization, i.e., a single optimal action is computed each time the system arrives at a state.
\change{These works formulate the virtual-time OCP online at each step, taking the current state as the initial state, and constraining the action to ensure subsequent states remain in the invariant set of the CBF.}
\citet{ames2016control} unify CBF with control Lyapunov function (CLF) in the context of quadratic program (QP), in which performance objective is expressed by CLF and safety constraint is expressed CBF. The optimal control is obtained by solving the QP online and is ensured to keep the state in the invariant set of CBF. \citet{nguyen2016exponential} propose exponential CBFs that enforce strict satisfaction of high relative degree safety constraints for nonlinear systems. They also develop a method for designing exponential CBFs based on techniques from linear control theory.
\change{Another method to deal with high relative degree is HOCBF \citep{xiao2019control, xiao2021high}, which takes high-order time derivatives of the constraint function so that where the action explicitly appears in the CBF.
To address the over-conservativeness of CBFs, \citet{xiao2023barriernet} replace the set of hard constraints HOCBFs with a set of differentiable constraints. This is achieved by a barrier layer that can be combined with any neural network controller.}
\citet{taylor2020learning} use a learning-based method to reduce model uncertainty in order to enhance the safety of a CBF-cerified controller. Their approach iteratively collects data and updates the controller, ultimately achieving safe behavior.

Other works solve CBF-constrained OCPs in an offline manner using RL, i.e., they learn a policy that maps states to their corresponding optimal actions.
\change{These works integrate CBF constraints into the policy optimization of the RL training process, ensuring that the learned policy renders the zero-sublevel set of the CBF forward invariant.}
\citet{cheng2019end} ensure safety of a model-free RL controller by combining it with a model-based CBF controller and online learning of unknown system dynamics. The CBF controller both guarantees safety and guides the learning process by constraining the set of explorable policies. \citet{ohnishi2019barrier} propose a barrier-certified adaptive RL algorithm, which constrains policy in the invariant set of CBF and optimizes the action-value function in this set. Their solutions to barrier-certified policy optimization are guaranteed to be globally optimal under mild conditions. \citet{ma2021model} use a generalized CBF for systems with high relative degrees and mitigate the infeasibility of constrained policy optimization by an adaptive coefficient mechanism.
\change{\citet{marvi2021safe} augment the cost function with a CBF and solve the resulting problem using a model-free RL algorithm, yielding a controller with the ability of proactive safety planning.
\citet{emam2022safe} frame safety as a differentiable robust CBF layer in a model-based RL framework, within which the underlying reward-driven task is modularly learned independent of safety constraints.}

Most of the above works handcraft CBFs as functions with known forms, while some other works synthesize CBFs using neural networks. \citet{robey2020learning} leverage safe trajectories generated by an expert to optimize a CBF in control affine systems. The learned CBF enjoys provable safety guarantees under Lipschitz smoothness assumptions on system dynamics. \citet{qin2020learning} jointly learn multi-agent control policies and CBFs in a decentralized framework. They propose a spontaneous policy refinement method to further enforce CBF conditions during testing. \citet{yang2023model} propose a safe RL algorithm that learns both policy and CBF in a model-free manner. They extend the CBF invariant loss to a multi-step version, which balances bias and variance and enhances both safety and performance of the policy.
\change{To handle control input limits, \citet{liu2023safe} propose a learner-critic architecture, where the critic finds counterexamples of input saturation and the learner optimizes the neural CBF to eliminate them.
To address the conservativeness of learning-based CBFs, \citet{yang2024synthesizing} propose feasible region iteration that learns the maximum feasible region with a CDF to generate accurate feasibility labels for learning the CBF.}
 
\subsection{Safety index}
SI is commonly used as a safeguard for another (possibly unsafe) controller, e.g., an RL controller that solely maximizes reward performance. \citet{zhao2021model} propose a model-free RL algorithm that leverages an SI to ensure zero constraint violation during training. This is achieved by an implicit safe set algorithm, which searches for safe control only by querying a black-box dynamic function. \citet{ma2022joint} simultaneously synthesize an SI and learn a safe control policy with constrained RL. They learn the SI by minimizing the occurrence of energy increases, which does not rely on knowledge about a prior controller.
\change{To deal with systems with varying dynamics, \citet{yun2025safe} propose an SI adaptation method that guarantees forward invariance and finite convergence of the safe control law.
\citet{chen2024real} leverage determinant gradient ascend and derive a closed-form update to SI parameters once the dynamics are perturbed.
\citet{wei2022safe} propose a method for safety control synthesis with neural network dynamic models. The method first synthesizes an SI offline using evolutionary methods, and then solves a constrained optimization problem online based on an encoding of neural networks.}

In systems with control limits, there may be situations where it is impossible to find an action that satisfies the constraint of an SI. Some works focus on how to synthesize a valid SI under control limits. \citet{wei2022persistently} propose a control-limits-aware SI synthesis method for systems with bounded state-dependent uncertainties. They use convex semi-infinite programming to solve for a robust safe controller so that it is guaranteed to be realizable under control limits. \citet{zhao2023safety} propose a method for synthesizing SI in general systems with control limits. They prove that ensuring the existence of safe control on a safe set boundary is equivalent to sum-of-squares programming. \citet{zhao2023probabilistic} present an integrated dynamic model learning and safe control framework to safeguard any RL agent. They provide a design rule to construct an SI under control limits and a probabilistic safety guarantee under stochastic dynamic models.
\change{\citet{chen2024safety} study the problem of synthesizing feasible safety indices under state-dependent control spaces, and leverage Positivstellensatz to formulate the problem as nonlinear programming.}

\section{Constraint aggregation}
\subsection{Cost value function}
Plenty of existing works, mostly concentrating on the setting of CMDP, have applied CVF to handle constraints.
Generally, most works follow the standard notion of CMDPs, where the constraints are imposed directly on CVF itself.
\change{These works typically consider the cumulative constraint, which requires the sum of costs to be less than a given threshold. When the threshold is set to zero, the cumulative constraint is equivalent to requring all costs to be zero, and becomes the statewise constraint considered in this paper.}
\citet{chow2017risk} build a constraint based on the conditional value-at-risk of CVF and propose Lagrange multiplier methods for the constrained optimization problem.
\citet{ding2020natural} also employ the primal-dual approach but update the primal variable via natural policy gradient and the dual variable via projected sub-gradient. This work is extended to an entropy-regularized case in \citep{ying2022dual}.
An equivalent linear formulation of the objective and the CVF-based constraint \citep{altman1999constrained} is considered in \citep{bai2022achieving} and solved by a stochastic primal-dual algorithm.
To dampen the significant oscillation of state and cost value function during training when using Lagrange multiplier methods, \citet{stooke2020responsive} and \citet{peng2022model} use PID control to update the Lagrange multiplier for a stabler intermediate performance.
\citet{as2022constrained} use Bayesian world models to estimate an optimistic upper bound on task objective and pessimistic upper bounds on safety constraints. They use the augmented Lagrangian method to solve the constrained optimization problem based on these two bounds.

Besides Lagrangian-based methods, other approaches have also been proposed for solving OPCs with CVF constraints.
\citet{liu2022constrained} deal with instability issues of primal-dual style methods by introducing the Expectation-Maximization approach. By adopting a non-parametric variational distribution, the constrained optimization problem in the expectation step becomes convex and can be solved analytically.
\citet{liu2020ipo} introduce the interior-point method to augment the objective with logarithmic barrier functions composed with CVF.
\citet{chow2018lyapunovbased} propose a Lyapunov-based approach for transient MDPs which constructs Lyapunov functions w.r.t an undiscounted version of CVF.
\citet{achiam2017constrained} construct constrained optimization problems on the basis of a novel bound on the difference in cumulative rewards or costs between two policies and solve them with trust region methods. 
However, these optimization problems may be infeasible, which undermines the theoretical monotonicity. \citet{yang2020projection} address this by first applying an unconstrained trust region method and then projecting the policy back onto the constrained set.
Likewise, \citet{zhang2020first} also propose a two-stage algorithm that first searches for a solution of a constrained optimization problem in the non-parameterized policy space and then projects it back into the parametric one.
\citet{yu2022towards} propose a two-policy method where a safety editor, serving as an extension of a safety shield, is trained to reconcile possible constraint violation with minimal influence on the objective.
\change{\citet{liu2023constrained} study the offline safe RL problem and propose the constrained decision transformer. They train the policy in an autoregressive manner, making it adaptable to different thresholds of the cumulative cost constraints during deployment.
\citet{zhang2024cvar} consider the constraint of the conditional value-at-risk (CVaR) of cumulative costs, and propose the CVaR-CPO algorithm, which learns a distributional CVF to provide a quantile-based estimation of the CVaR constraint.}

\change{While the CVF is originally proposed for cumulative constraints in CMDP, the connection between CVF and state-wise constraints has been revealed by existing works. \citet{sootla2022saute} propose a framework that converts the cumulative costs constraint into state-wise constraints by augmenting the state with the remaining safety budget, preserving the Markov property.}
\citet{zhao2023state} introduces the framework of maximum MDP, which is an extension of CMDP that constrains the expected maximum state-wise cost along a trajectory. Under this framework, they propose state-wise constrained policy optimization algorithm, which provides guarantees for state-wise constraint satisfaction in expectation.

\subsection{Hamilton-Jacobi reachability function}
HJ reachability analysis computes the backward reachable set of a constrained system, which is its infeasible region, and safeguards controllers from entering this set.
\change{While both HJ reachability function and CVF are CA-type feasibility functions, HJ reachability function can better leverage continuous constraint functions such as distance to obstacles, while CVF are typically used with binary cost (constraint violation) signals. This makes HJ reachability function more suitable for accurately representing feasible regions in safety-critical control tasks with continuous state constraints.}
\citet{seo2019robust} use the reachable set obtained by HJ reachability analysis to safeguard receding horizon planning against unknown bounded disturbances. They approximate the reachable set using ellipsoidal parameterization and plan a robust trajectory that avoids risky regions under disturbance.
The exact computation of HJ reachability function requires solving an HJ partial differential equation (PDE) on a grid discretization of state space, resulting in an exponential computational complexity with respect to system dimension \citep{bansal2017hamilton}. Many efforts have been made to reduce the computational burden of HJ reachability.
\citet{rubies2019classification} approximates the optimal controller of HJ reachability problem in control-affine systems as a sequence of simple binary classifiers, thus avoiding storing a representation of HJ reachability function.
\citet{herbert2021scalable} propose several techniques, including decomposition, warm-starting, and adaptive grids, to speed up the computation of HJ reachability function. Their methods can update safe sets by one or more orders of magnitude faster than prior work.

Other works further accelerate computation by approximating HJ reachability function with neural networks. 
\citet{fisac2019bridging} introduce a time-discounted modification of HJ reachability function, which induces a contraction mapping and enables solving it by a fixed point iteration method. Their obtained reachability function approximates the maximum safe set and the safest policy.
Based on this work, \citet{yu2022reachability} further consider policy performance optimization in safe RL. They use HJ reachability function to construct virtual-time constraints and solve for the optimal safe policy with the Lagrange method.
\citet{bansal2021deepreach} develop a neural PDE solver for high-dimensional reachability problems. The computational requirements of their method do not scale directly with the state dimension but rather with the complexity of the underlying reachable tube.
\change{To ensure exploration safety under model uncertainty, \citet{yu2023safe} propose a distributional reachability certificate (DRC) and its Bellman equation to characterize robust persistently safe states, and build a safe RL framework to resolve constraints required by the DRC and its corresponding shield policy.
\citet{zheng2024safe} use a neural HJ reachability function to identify the maximum feasible region in offline RL. They decouple safety constraint adherence, reward maximization, and offline policy learning into three processes, where the optimal policy is derived in the form of weighted behavior cloning.}

\chapter{Examples of Classic Control Problems}
\label{sec: experiments}
In this chapter, we illustrate the proposed concepts, including feasible regions, containment relationships between them, and feasibility functions with \change{two classic constrained optimal control problems: emergency braking control and unicycle obstacle avoidance.}
\change{The example code used for all experiments in this chapter is available at our repository\footnote{https://github.com/yangyujie-jack/Feasibility-Tutorial}.}

\section{Emergency braking control}
\begin{figure}
    \centering
    \includegraphics[width=0.7\linewidth]{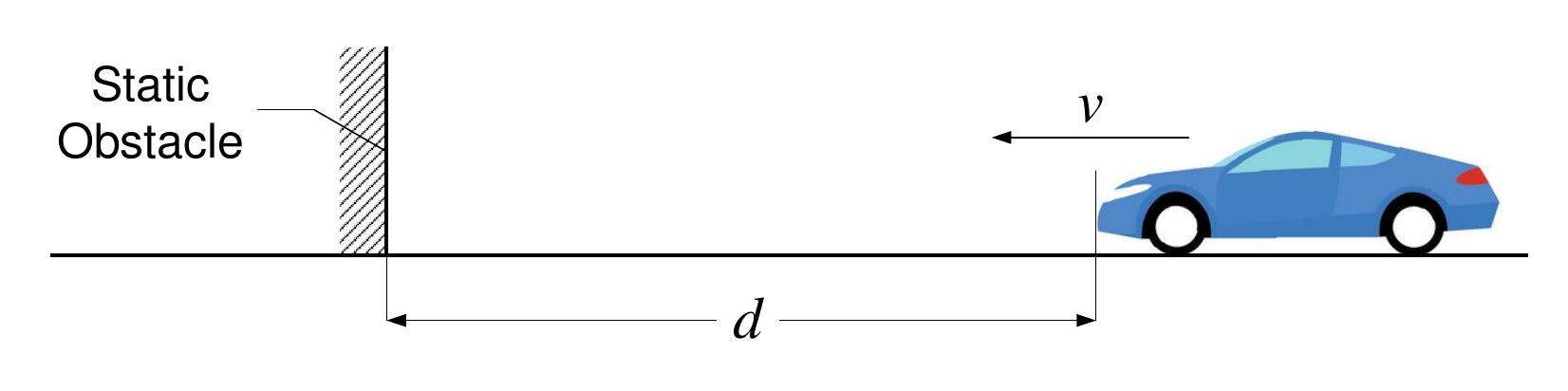}
    \caption{Emergency braking control scenario.}
    \label{fig: emergency braking}
\end{figure}
In emergency braking control, a vehicle is supposed to avoid crash with minimum effort. Figure \ref{fig: emergency braking} gives a schematic of this task.
The vehicle has a two-dimensional state $x_t=[d_t,v_t]^\top$ and a one-dimensional action $u_t=a_t$, where $d$ is the distance to the static obstacle, $v$ is the longitudinal velocity and $a$ is the longitudinal acceleration. It follows a simplified longitudinal dynamics:
\begin{equation}
    \begin{bmatrix}
        d_{t+1} \\
        v_{t+1}
    \end{bmatrix} =
    \begin{bmatrix}
        1 & -\Delta t \\
        0 & 1
    \end{bmatrix}
    \begin{bmatrix}
        d_t \\
        v_t
    \end{bmatrix} +
    \begin{bmatrix}
        0 \\
        \Delta t
    \end{bmatrix} a_t,
\end{equation}
where $\Delta t$ is the time step size, assigned with $0.1\mathrm{s}$. We assume a maximum braking deceleration $a_\mathrm{brk}=-10\mathrm{m/s^2}$ so that the action is bounded in $[a_\mathrm{brk},0]$. The reward function (or oppositely in control terminology, the utility function) is
$$
r(x_t,u_t)=-u_t^2=-a_t^2.
$$
The overall objective to be maximized is
\begin{equation}
\label{eq: objective}
    J=\sum_{i=0}^{\infty} \gamma^i r(x_{t+i},u_{t+i})=-\sum_{i=0}^{\infty} \gamma^i a_{t+i}^2.
\end{equation}
The safety constraint requires that the vehicle-to-obstacle distance be above zero at all times, requiring collision-free control:
\begin{equation}
    h(x_{t+i})=-d_{t+i}\le0, i=0,1,2,\dots,\infty.
\end{equation}
The objective \eqref{eq: objective} follows a general RL setting. To show the applicability of our proposed theoretical framework to both MPC and RL, we also involve an MPC controller in the experiments. The objective of the MPC controller is limited to a finite horizon with length $N=10$ and has no discounting, i.e., $\gamma=1$. \change{The objective for RL is infinite-horizon with $\gamma=0.99$.}

In the following analysis, we build four types of virtual-time constraints: (1) pointwise constraint, (2) CBF constraint, (3) SI constraint, and (4) HJ reachability constraint. Constraint (1) is a non-feasibility-function-based virtual-time constraint, constraints (2) and (3) are CIS-based, and constraint (4) is CA-based. 

\noindent \textbf{Pointwise} constraint is composed of a finite number of real-time constraints, i.e.,
\begin{equation}
    h(x_{i|t})=-d_{i|t}\le0, i=0,1,2,\dots,n,
\end{equation}
where $n<\infty$. Since this constraint is not constructed by a feasibility function, its feasible regions are not readily accessible, and its IFR and EFR may not be equivalent. 

\noindent \textbf{CBF} is handcrafted in the following form:
\begin{equation}
\label{eq: CBF}
    B(x_{i|t})=-d_{i|t}+kv_{i|t}^2,
\end{equation}
where $k$ is a tunable parameter. This form is obtained by considering laws of kinematics. Specifically, when braking with a constant acceleration, the distance-to-go is a quadratic function of the current velocity. This form can also be understood from the perspective of HOCBF. If directly choosing the constraint function $h$ as the CBF, its relative degree is two and we should use an HOCBF that considers the first order time derivative of $h$, which is exactly what \eqref{eq: CBF} does.
The corresponding CBF constraint applies to the first two steps in virtual-time domain:
\begin{equation}
\label{eq: CBF constraint}
\begin{aligned}
    &B(x_{0|t})\le0, \\
    &B(x_{1|t})-(1-\alpha)B(x_{0|t})\le0, \\
\end{aligned}
\end{equation}
where $\alpha$ is a constant assigned with $0.1$. Strictly speaking, the choice of $\alpha$ is not arbitrary when the action is constrained. It should be chosen such that property 2) in Definition \ref{def: CBF} is satisfied. Otherwise, $B$ may not be a valid CBF and its zero-sublevel set may not be control invariant. Here, we choose the value of $\alpha$ empirically without verifying property 2).

\noindent \textbf{SI} follows the form used by \citet{zhao2021model}:
\begin{equation}
    \phi(x_{i|t})=\sigma+d_\mathrm{min}^n-d_{i|t}^n+kv_{i|t},
\end{equation}
where $d_\mathrm{min}$ is the minimum allowable distance to obstacle and $\sigma,n,k$ are tunable parameters. In our experiments, we fix $d_\mathrm{min}=0, \sigma=0.12$ and adjust the values of $n,k$ to see their effects. $\sigma$ is set slightly larger than zero to avoid small amounts of constraint violation caused by numerical issues in optimization. The SI constraint applies to the first two steps in virtual-time domain:
\begin{equation}
\begin{aligned}
    &\phi(x_{0|t})\le0, \\
    &\phi(x_{1|t})-\max\{\phi(x)-\eta,0\}\le0,
\end{aligned}
\end{equation}
where $\eta$ is assigned with $0$. \citet{zhao2021model} point out that for collision-avoidance tasks with action limits, the following design rule ensures forward invariance of the safe set:
\begin{equation}
\label{eq: SI design rule}
\frac{n(\sigma+d_\mathrm{min}^n+kv_\mathrm{max})^{\frac{n-1}{n}}}{k}\le-\frac{a_\mathrm{brk}}{v_\mathrm{max}},
\end{equation}
where $v_\mathrm{max}$ is the maximum velocity.

\noindent \textbf{HJ reachability} function \change{is intrinsically associated with a policy, which is typically the optimal one derived from dynamic programming, but can also be a pre-specified policy}. Here, we consider a safety-oriented policy that takes the maximum braking deceleration at every time step. Its corresponding HJ reachability function is the negative minimum future vehicle-to-obstacle distance starting from the current state:
\begin{equation}
\label{eq: HJR}
    F(x_{i|t})=-d_{i|t}-\frac{v_{i|t}^2}{2a_\mathrm{brk}}.
\end{equation}
Its corresponding constraint also applies to the first two steps in virtual-time domain:
\begin{equation}
    F(x_{i|t})\le0,i=0,1.
\end{equation}
Note that the CBF \eqref{eq: CBF} is very similar in form to the HJ reachability function \eqref{eq: HJR}. In fact, they are equivalent when taking $k=-1/2\cdot a_\mathrm{brk}$ in \eqref{eq: CBF}. However, the equivalence of the two feasibility functions does not imply the equivalence of their virtual-time constraints. \change{
Specifically, the virtual-time constraint constructed by the CBF is more restrictive than that constructed by the HJ reachability function, since the second constraint in the CBF case requires $B(x_{1|t})\le (1-\alpha)B(x_{0|t})$ where $(1-\alpha)B(x_{0|t})$ may be a negative number, while that in the HJ reachability case only requires $F(x_{1|t})\le 0$.
}

We solve MPC and RL controllers under these four virtual-time constraints, resulting in eight combinations. For MPC, we use IPOPT to solve the constrained OCPs and compute the optimal actions online. \change{For RL, we adopt the approximate dynamic programming \citep{li2023reinforcement} algorithm, which trains a policy network and a value network.
Both networks are multilayer perceptrons with two hidden layers of 64 neurons each and ReLU activation functions.
The value network is updated by minimizing the mean squared error between its prediction and the target computed by one-step look-ahead. The loss function is
\begin{equation}
    \mathcal{L}_V(\phi) = \mathbb{E}\left[\left(V_{\phi}(x_t)-\left(r(x_t,u_t)+\gamma V_{\phi}(x_{t+1})\right)\right)^2\right],
\end{equation}
where $\phi$ are the parameters of the value network.
The policy network is updated by feasible policy improvement \citep{yang2023feasible}, which splits the state space into feasible region and infeasible region, and applies different update rules in these two regions.
Inside the feasible region, we maximize the value function:
\begin{equation}
    \mathcal{L}_{\pi,\text{in}}(\theta) = -\mathbb{E}\left[\mathbb{I}_{\Xedlsg}(x_t)\cdot (r(x_t,u_t)+\gamma V_{\phi}(x_{t+1}))\right],
\end{equation}
where $\theta$ are the parameters of the policy network, and $\mathbb{I}_{\Xedlsg}$ is the indicator function of the feasible region.
Outside the feasible region, we minimize the feasibility function:
\begin{equation}
    \mathcal{L}_{\pi,\text{out}}(\theta) = \mathbb{E}\left[(1-\mathbb{I}_{\Xedlsg}(x_t))\cdot F(x_{t+1})\right],
\end{equation}
where $F$ is the feasibility function used in the virtual-time constraint.
The total loss for policy network is
\begin{equation}
    \mathcal{L}_\pi(\theta) = \mathcal{L}_{\pi,\text{in}}(\theta) + \mathcal{L}_{\pi,\text{out}}(\theta).
\end{equation}
It is proved \citep{yang2023feasible} that this method guarantees monotonic expansion of the feasible region, monotonic improvement of the value function, and convergence to the optimal policy.
The optimizer for both networks is Adam, and the learning rate is $10^{-4}$. A mini-batch of size 256 is sampled for each update.
More implementation details can be found in our source code.}

To demonstrate the feasibility and control performance of different policies under different constraints, we visualize state trajectories and feasible regions of MPC and RL policies under the four virtual-time constraints. For MPC, we adjust the parameters of each virtual-time constraint to study their impact on the feasibility of the optimal policy. For RL, we fix the parameters and visualize trajectories and regions of intermediate policies at different iterations to study how they change. This demonstrates the capability of our theory to describe the feasibility of not only the optimal policy but also intermediate non-optimal policies during RL training.

Figure \ref{fig: MPC trajectory PW} shows state trajectories of MPC under pointwise constraints with different horizons. Each state trajectory is represented with a solid point followed by a series of hollow points, indicating the initial and successive states, respectively. The hollow point at the end of a trajectory is substituted with a cross if that state violates the real-time constraint indicated by the grey region. The red dashed line is the boundary of the maximum EFR computed analytically by considering the most cautious policy which always brakes with the maximum deceleration. When the constraint horizon $n$ is short, although starting from a state inside the maximum EFR, some trajectories still end up crashing due to the short-sighted control policy under pointwise constraints. This means that the EFR of the optimal policy under pointwise constraints is only a subset of the maximum EFR. As the horizon becomes longer, more states become endlessly feasible under the optimal policy, which coincides with the fact that the maximum EFR actually corresponds to a pointwise constraint with $n=\infty$, as stated in Theorem \ref{thm: maximum EFR equivalence}.

Figure \ref{fig: MPC region PW} shows feasible regions of MPC under pointwise constraints with different horizons. Since MPC can be viewed as an optimal policy, its EFR equals that of the optimal policy, and its IFR equals that of the constrained OCP (unrelated to any policy), i.e., $\Xinitg(\pi^*)=\Xinitg$. The feasible regions are obtained by running an episode from every state to test its feasibility. Results suggest that short-sighted pointwise constraints fail to recognize the inevitable collision at a distance and, hence, over-optimistically treat some states outside the maximum EFR as initially feasible. At the same time, pointwise constraints are weaker than real-time constraints, so the policy may take too aggressive actions early in the episode, making the EFR of the optimal policy smaller than the maximum EFR. The results also show that as the constraint horizon becomes longer, the EFR of the optimal policy enlarges while the IFR shrinks. When the horizon is long enough, e.g., $n=10$, both regions are almost the same as the maximum EFR. This again supports that with infinite-horizon pointwise constraints, which essentially equal real-time constraints, IFR becomes EFR, and they both equal the maximum EFR.

\begin{figure}[t]
    \centering
    \subfloat{
        \includegraphics[trim=10 20 10 20, width=0.22\linewidth]{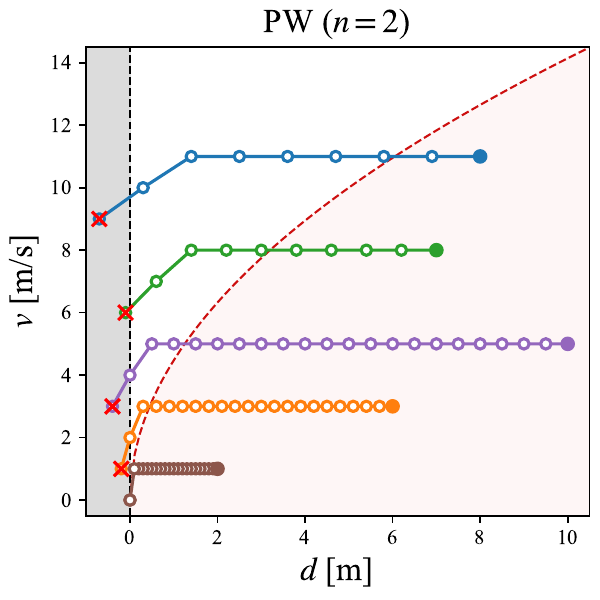}
    }
    \subfloat{
        \includegraphics[trim=10 20 10 20, width=0.22\linewidth]{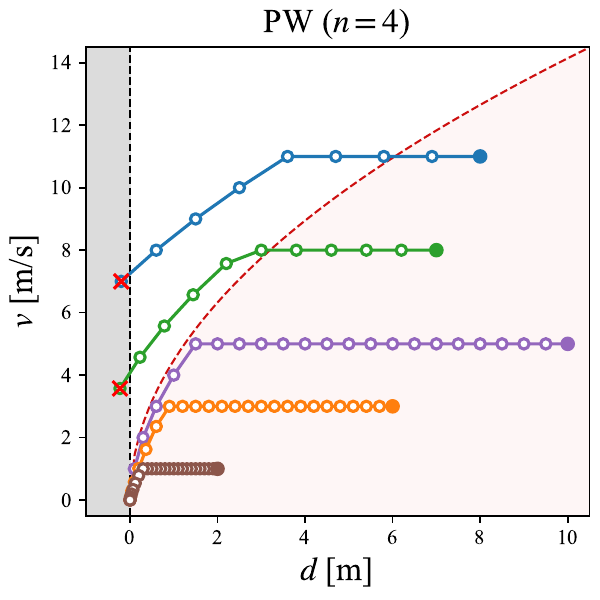}
    }
    \subfloat{
        \includegraphics[trim=10 20 10 20, width=0.22\linewidth]{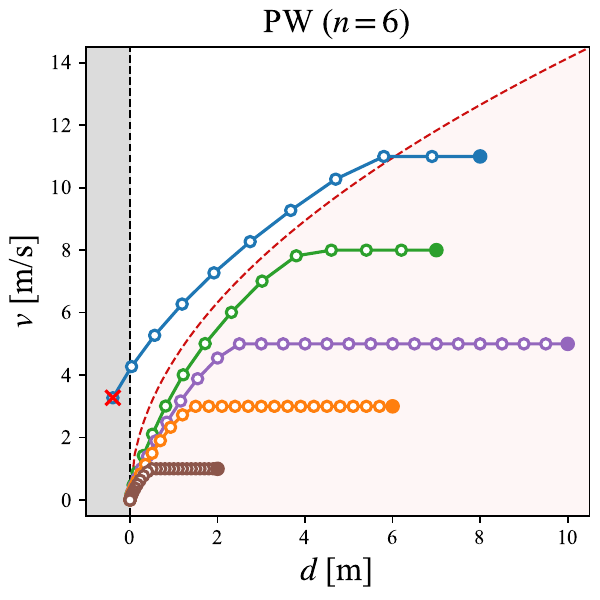}
    }
    \subfloat{
        \includegraphics[trim=10 20 10 20, width=0.22\linewidth]{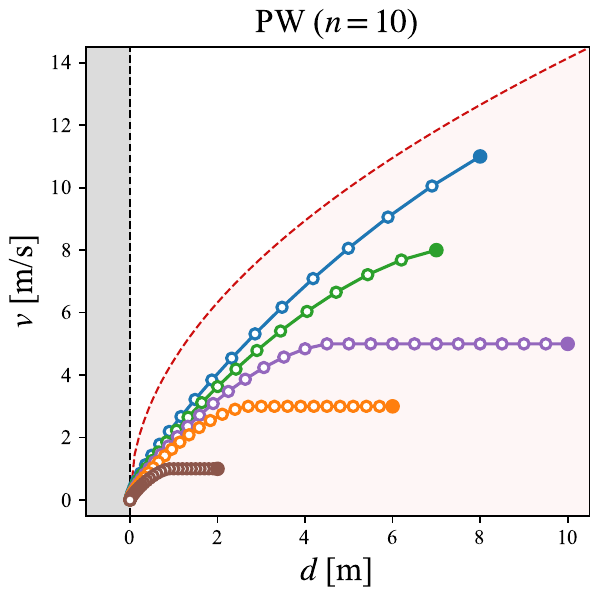}
    }
    \caption{State trajectories of MPC under pointwise constraints with different horizons. ``PW'' stands for ``pointwise''.}
    \label{fig: MPC trajectory PW}
\end{figure}

\begin{figure}[t]
    \centering
    \subfloat{
        \includegraphics[trim=10 20 10 20, width=0.22\linewidth]{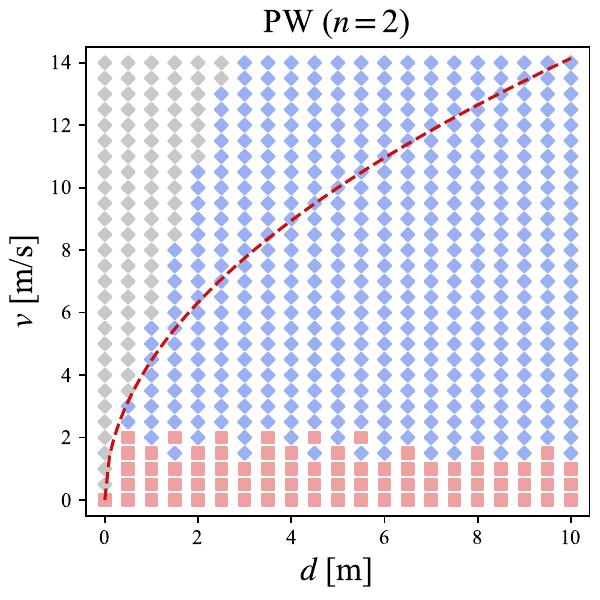}
    }
    \subfloat{
        \includegraphics[trim=10 20 10 20, width=0.22\linewidth]{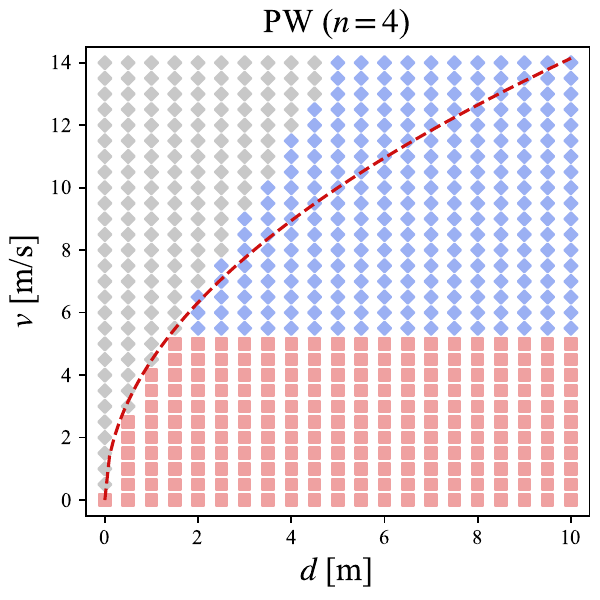}
    }
    \subfloat{
        \includegraphics[trim=10 20 10 20, width=0.22\linewidth]{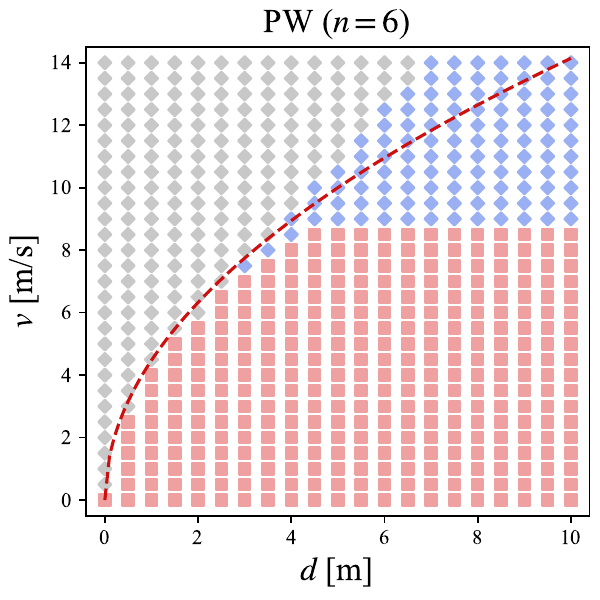}
    }
    \subfloat{
        \includegraphics[trim=10 20 10 20, width=0.22\linewidth]{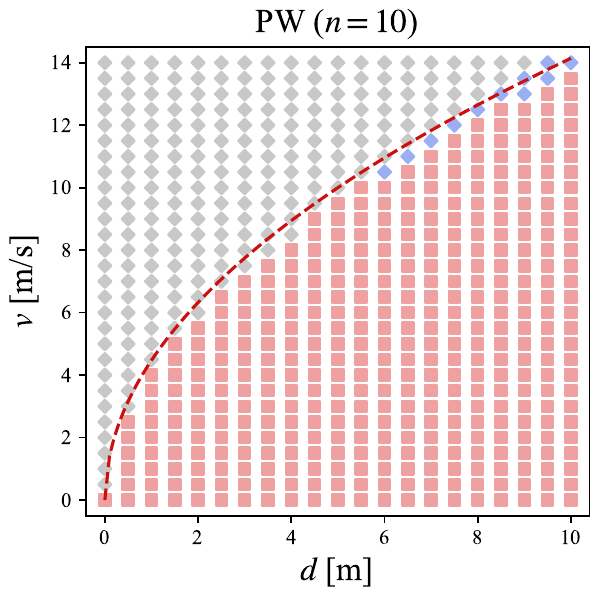}
    }
    \\
    \subfloat{
        \includegraphics[width=0.6\linewidth, trim=130 8 90 188, clip]{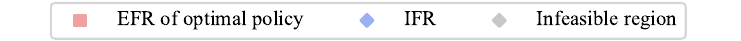}
    }
    \caption{Feasible regions of MPC under pointwise constraints with different horizons.}
    \label{fig: MPC region PW}
\end{figure}

Figure \ref{fig: RL trajectory PW} and \ref{fig: RL region PW} show state trajectories and feasible regions of RL under the same pointwise constraint at different iterations. At iteration 10, the feasible regions of the policy are small, and states with high velocities break out of the maximum EFR, resulting in constraint violations. Note that the EFR is smaller than the IFR because, in some states, constraint violation does not happen in the finite virtual horizon (10 steps) but is inevitable in the long run. Here, this infeasibility phenomenon is mainly due to inadequate policy training. As training proceeds, the feasible regions monotonically expand, and more state trajectories are included in the maximum EFR. At iteration 10000, the feasible regions are close to those of the optimal policy of MPC. These results indicate the effectiveness of feasible policy improvement \citep{yang2023feasible}, whose theoretical foundation is the feasibility analysis tools for intermediate non-optimal policies proposed in this paper.

\begin{figure}
    \centering
    \subfloat{
        \includegraphics[trim=10 20 10 20, width=0.22\linewidth]{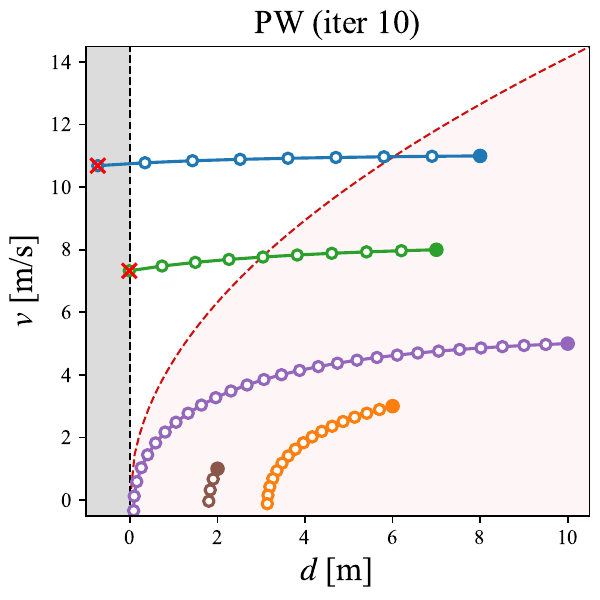}
    }
    \subfloat{
        \includegraphics[trim=10 20 10 20, width=0.22\linewidth]{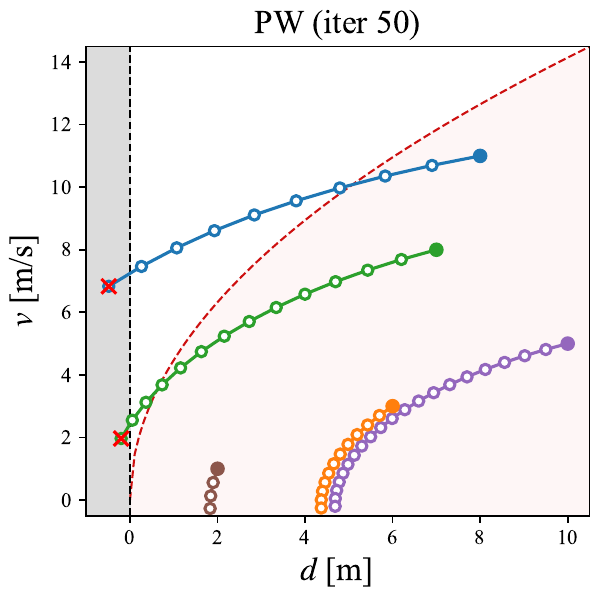}
    }
    \subfloat{
        \includegraphics[trim=10 20 10 20, width=0.22\linewidth]{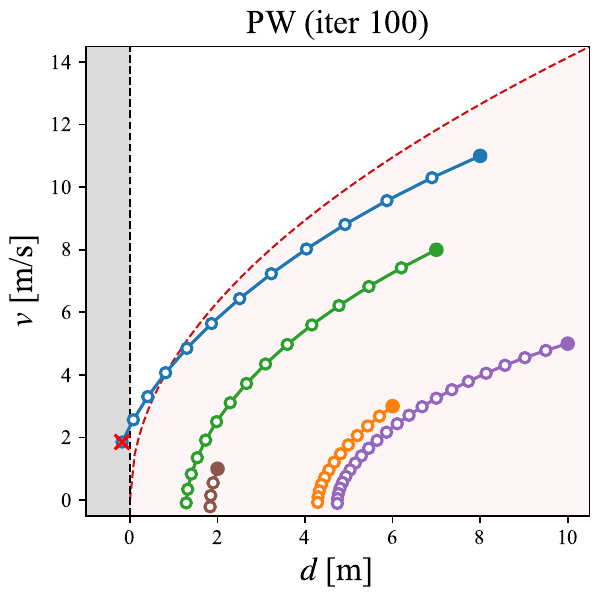}
    }
    \subfloat{
        \includegraphics[trim=10 20 10 20, width=0.22\linewidth]{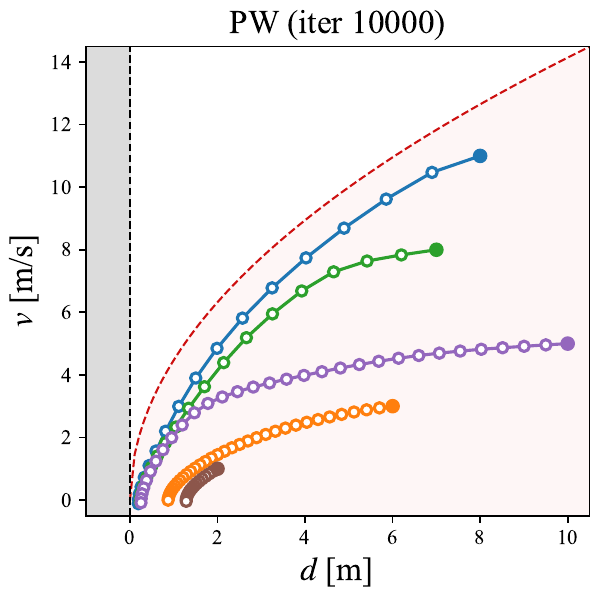}
    }
    \caption{State trajectories of RL under the same pointwise constraint ($n=10$) at different iterations.}
    \label{fig: RL trajectory PW}
\end{figure}

\begin{figure}
    \centering
    \subfloat{
        \includegraphics[trim=10 20 10 20, width=0.22\linewidth]{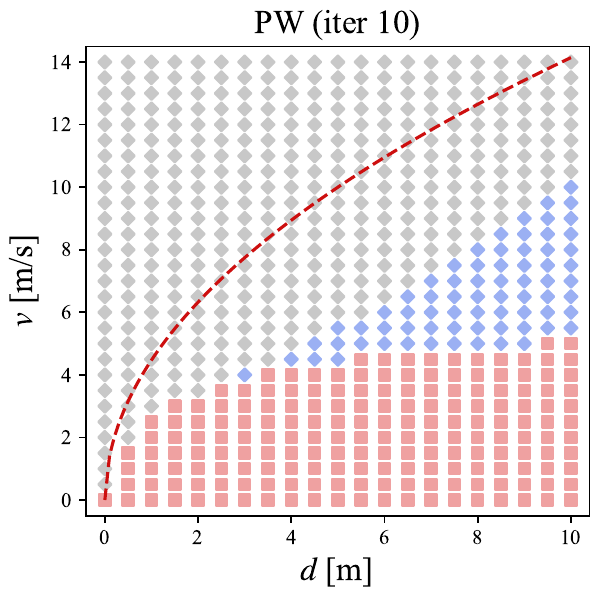}
    }
    \subfloat{
        \includegraphics[trim=10 20 10 20, width=0.22\linewidth]{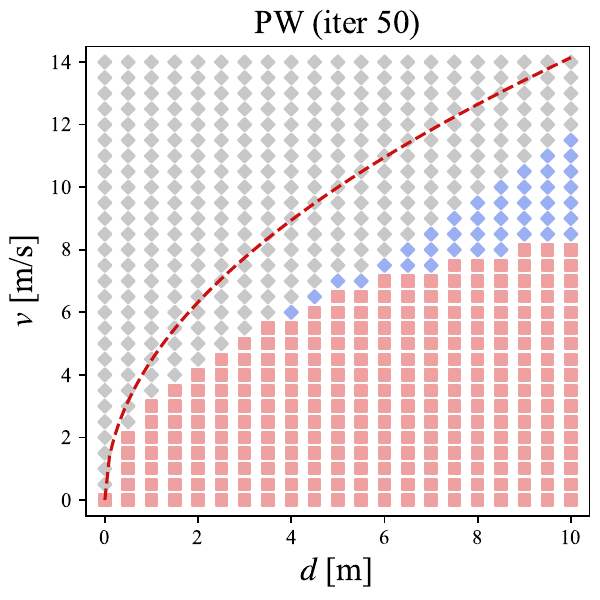}
    }
    \subfloat{
        \includegraphics[trim=10 20 10 20, width=0.22\linewidth]{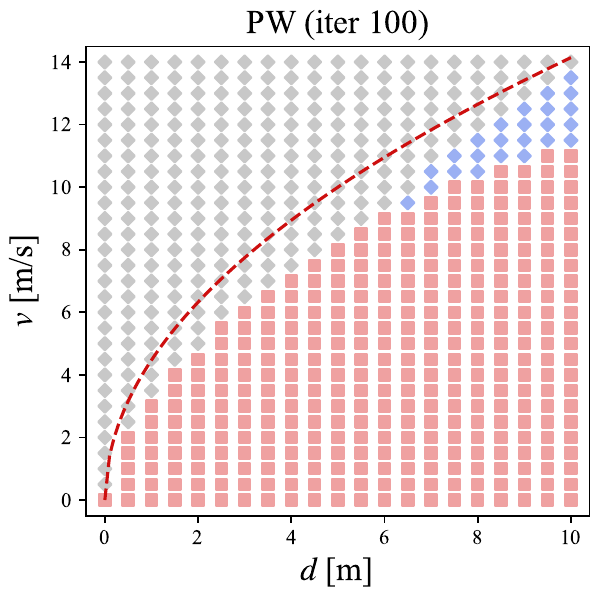}
    }
    \subfloat{
        \includegraphics[trim=10 20 10 20, width=0.22\linewidth]{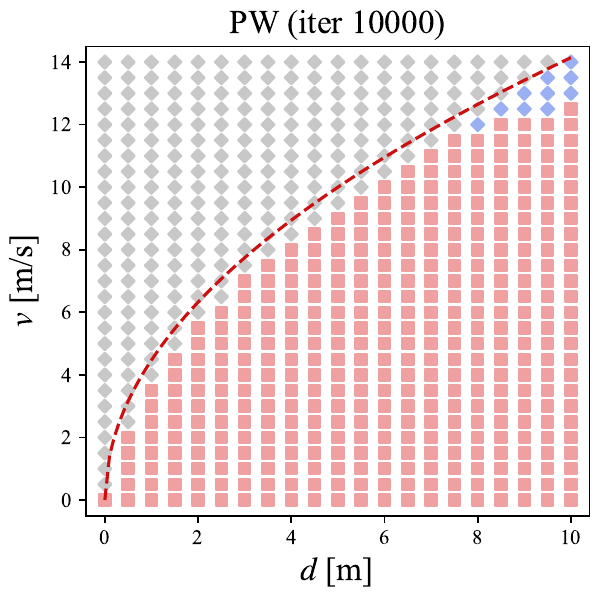}
    }
    \\
    \subfloat{
        \includegraphics[width=0.9\linewidth, trim=30 5 0 185, clip]{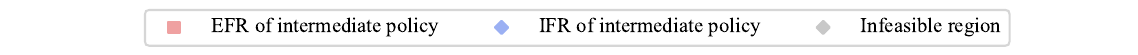}
    }
    \caption{Feasible regions of RL under the same pointwise constraint ($n=10$) at different iterations.}
    \label{fig: RL region PW}
\end{figure}

Figure \ref{fig: MPC trajectory CBF} and \ref{fig: MPC region CBF} show state trajectories and feasible regions of MPC under CBF constraints with different parameters. We observe that the EFR of the optimal policy is always identical to the IFR, which is a common feature of CIS-based virtual-time constraints and is consistent with the analysis in Section \ref{sec: CBF} and Corollary \ref{cor: containment relationship}(3). A smaller value of $k$ results in a larger feasible region, rendering more states safe in the long run. When $k=0.05$, the CBF actually equals the HJ reachability function \eqref{eq: HJR}, and its zero-sublevel set equals the maximum EFR. Another phenomenon is that as the feasible regions expand, restrictions on trajectories inside the feasible regions reduce. This can be seen by comparing the purple or orange trajectory when $k=0.5$ and $k=0.05$, in which the former decelerates immediately while the latter keeps its velocity for some time before decelerating. In other words, trajectories inside the feasible region become less conservative as $k$ decreases.

\begin{figure}
    \centering
    \subfloat{
        \includegraphics[trim=10 20 10 20, width=0.22\linewidth]{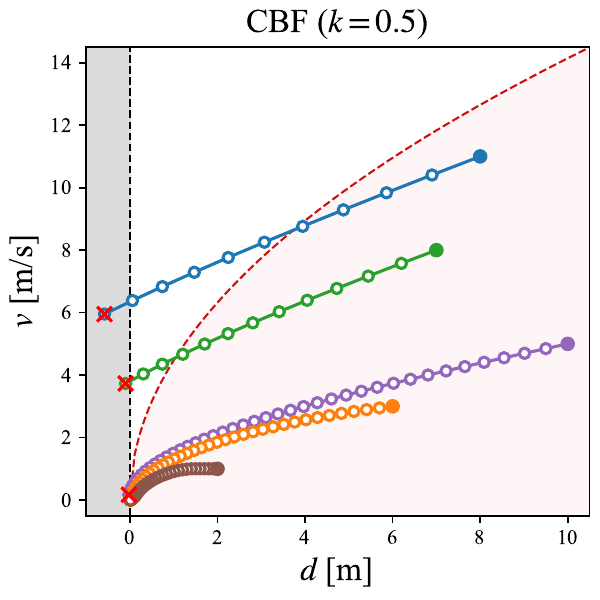}
    }
    \subfloat{
        \includegraphics[trim=10 20 10 20, width=0.22\linewidth]{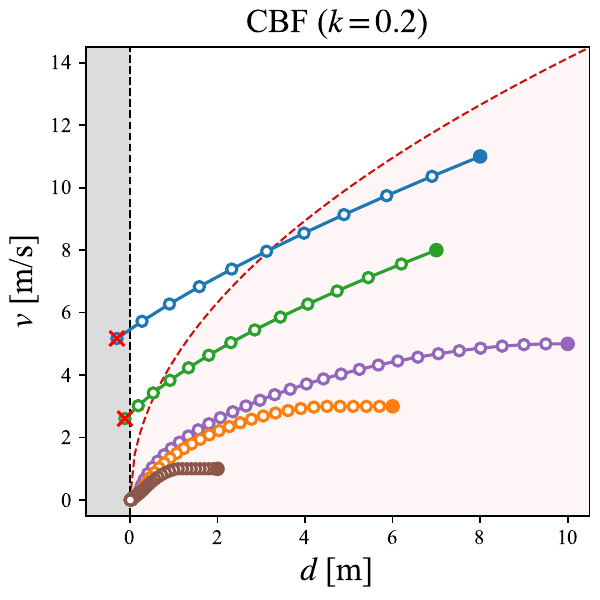}
    }
    \subfloat{
        \includegraphics[trim=10 20 10 20, width=0.22\linewidth]{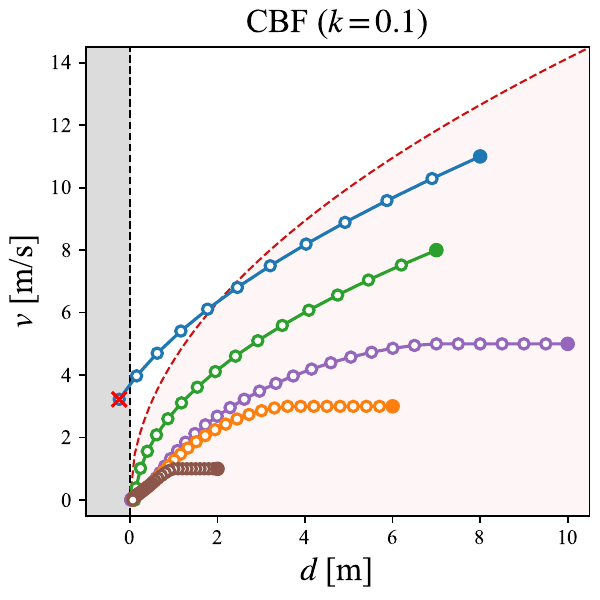}
    }
    \subfloat{
        \includegraphics[trim=10 20 10 20, width=0.22\linewidth]{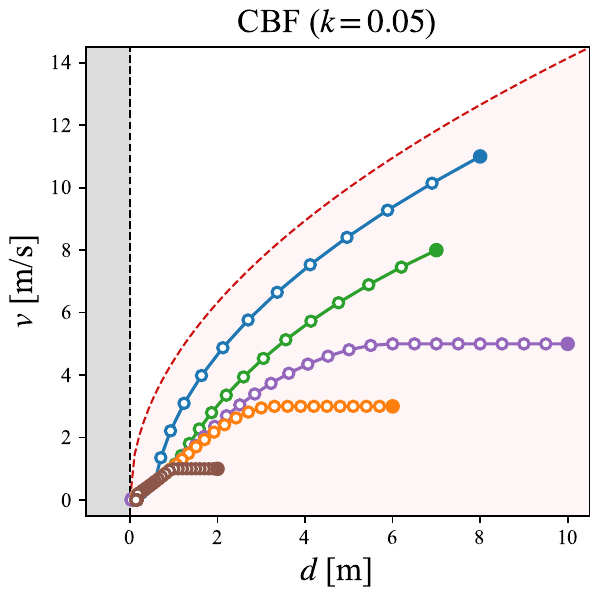}
    }
    \caption{State trajectories of MPC under CBF constraints with different parameters.}
    \label{fig: MPC trajectory CBF}
\end{figure}

\begin{figure}
    \centering
    \subfloat{
        \includegraphics[trim=10 20 10 20, width=0.22\linewidth]{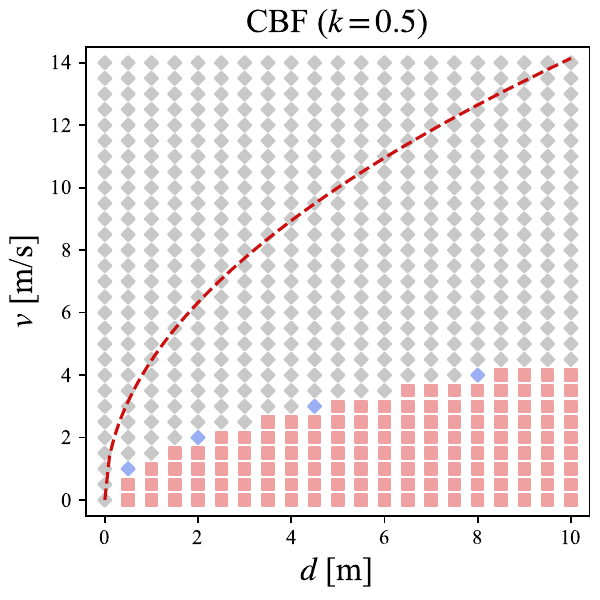}
    }
    \subfloat{
        \includegraphics[trim=10 20 10 20, width=0.22\linewidth]{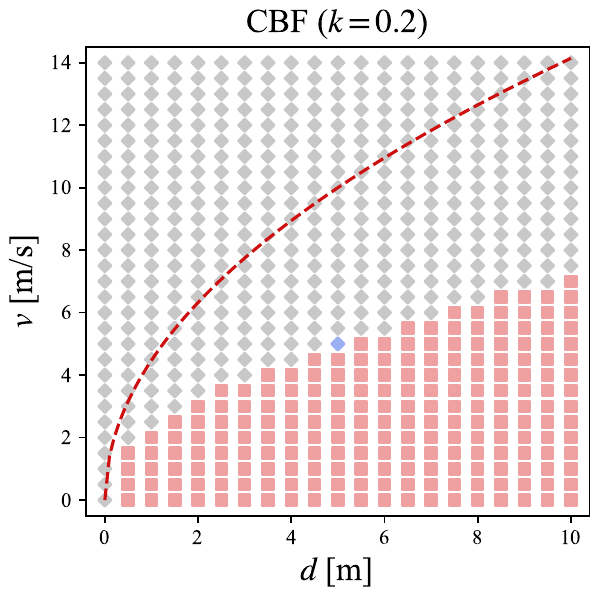}
    }
    \subfloat{
        \includegraphics[trim=10 20 10 20, width=0.22\linewidth]{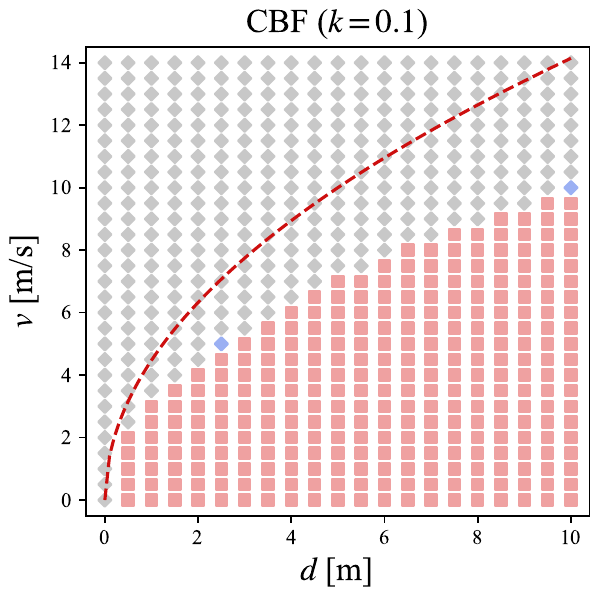}
    }
    \subfloat{
        \includegraphics[trim=10 20 10 20, width=0.22\linewidth]{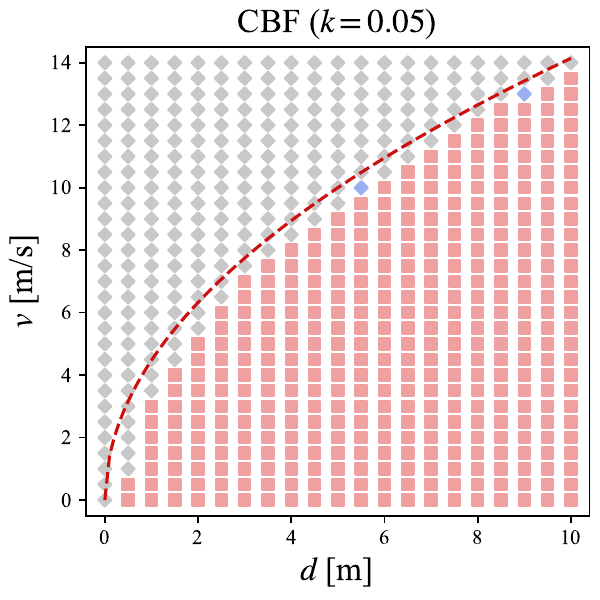}
    }
    \\
    \subfloat{
        \includegraphics[width=0.6\linewidth, trim=130 8 90 188, clip]{Figure/legend_mpc.pdf}
    }
    \caption{Feasible regions of MPC under CBF constraints with different parameters.}
    \label{fig: MPC region CBF}
\end{figure}

Figure \ref{fig: RL trajectory CBF} and \ref{fig: RL region CBF} show state trajectories and feasible regions of RL under the same CBF constraint at different iterations. The feasible regions monotonically expand during training and are close to the maximum EFR at iteration 10000. The state trajectories are either very conservative or infeasible at the beginning and gradually converge to the optimal behavior, i.e., drive the state to the lower left corner. It's worth noting that the EFR and IFR are nearly always the same during training (we believe their mismatch at iteration 10 is due to inadequate training), which is very different from the case of the SI and HJ reachability introduced below. This phenomenon is due to the second constraint in \eqref{eq: CBF constraint}, which not only requires the state to stay in the feasible region but also restricts the increasing rate of CBF.

\begin{figure}
    \centering
    \subfloat{
        \includegraphics[trim=10 20 10 20, width=0.22\linewidth]{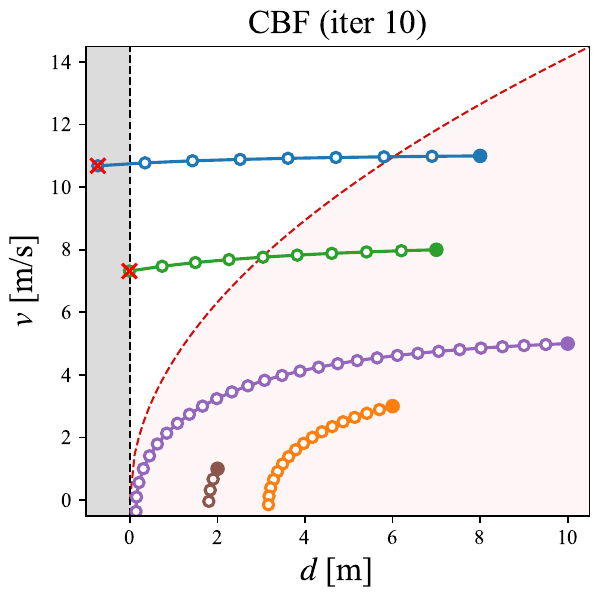}
    }
    \subfloat{
        \includegraphics[trim=10 20 10 20, width=0.22\linewidth]{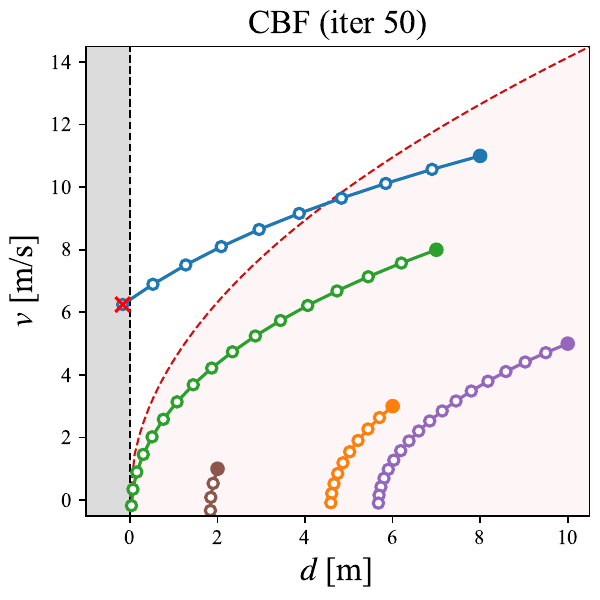}
    }
    \subfloat{
        \includegraphics[trim=10 20 10 20, width=0.22\linewidth]{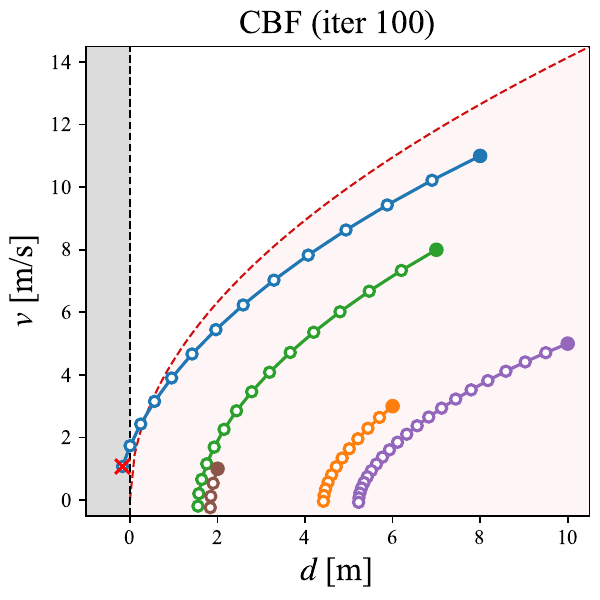}
    }
    \subfloat{
        \includegraphics[trim=10 20 10 20, width=0.22\linewidth]{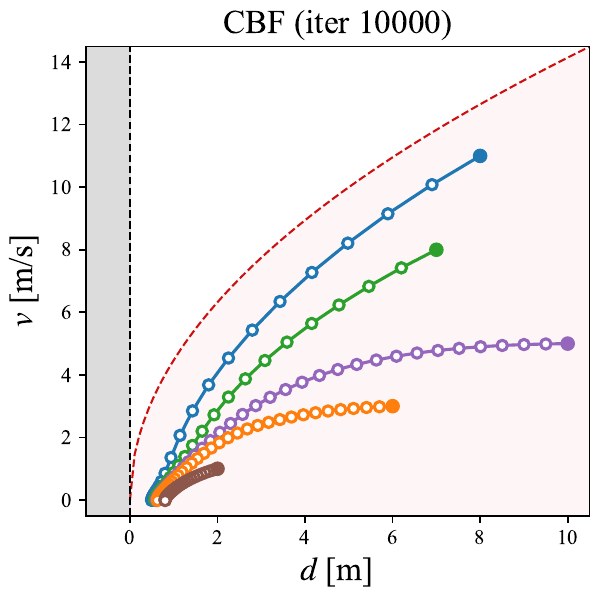}
    }
    \caption{State trajectories of RL under the same CBF constraint ($k=0.05$) at different iterations.}
    \label{fig: RL trajectory CBF}
\end{figure}

\begin{figure}
    \centering
    \subfloat{
        \includegraphics[trim=10 20 10 20, width=0.22\linewidth]{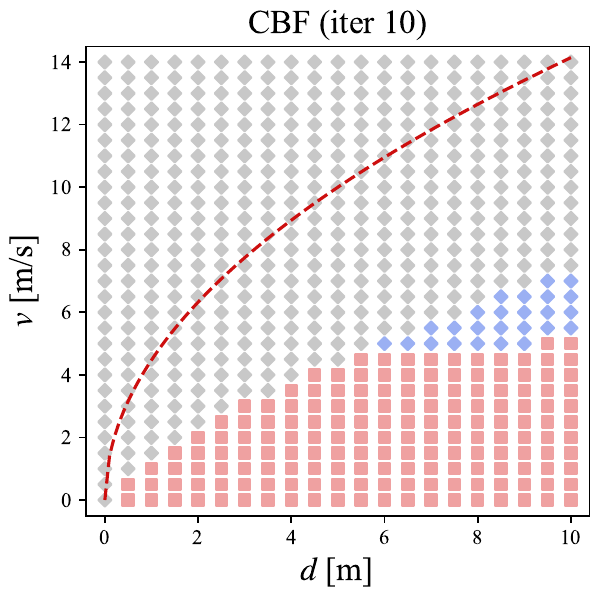}
    }
    \subfloat{
        \includegraphics[trim=10 20 10 20, width=0.22\linewidth]{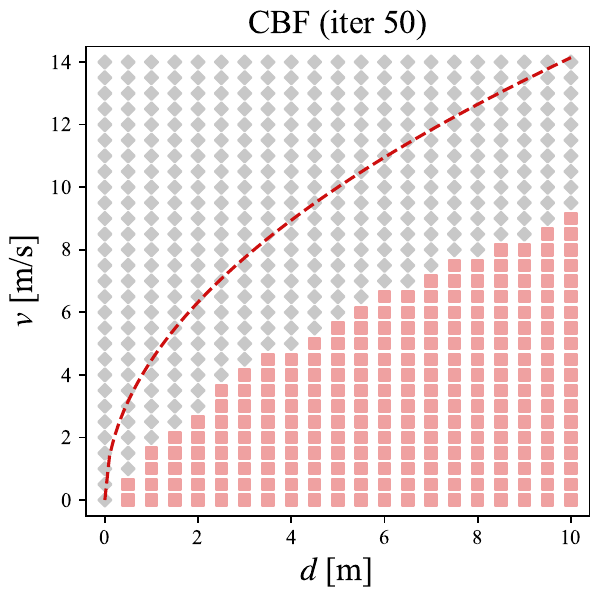}
    }
    \subfloat{
        \includegraphics[trim=10 20 10 20, width=0.22\linewidth]{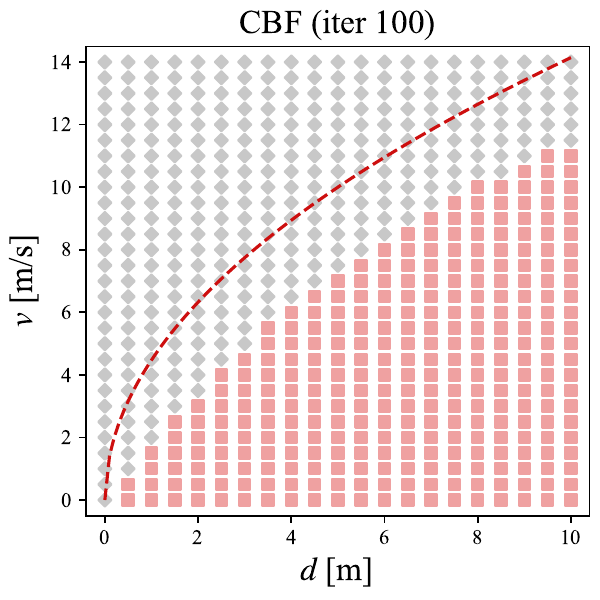}
    }
    \subfloat{
        \includegraphics[trim=10 20 10 20, width=0.22\linewidth]{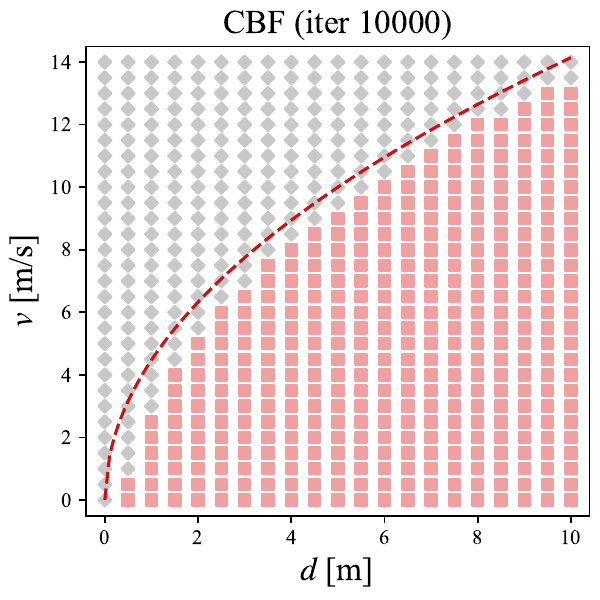}
    }
    \\
    \subfloat{
        \includegraphics[width=0.9\linewidth, trim=30 5 0 185, clip]{Figure/legend_rl.pdf}
    }
    \caption{Feasible regions of RL under the same CBF constraint ($k=0.05$) at different iterations.}
    \label{fig: RL region CBF}
\end{figure}

Figure \ref{fig: MPC trajectory SI} and \ref{fig: MPC region SI} show state trajectories and feasible regions of MPC under SI constraints with different parameters. The feasible regions' boundary shape accurately reflects the SI's function form. When $n=2$, the boundary is a quadratic function with respect to $d$. It becomes a linear function when $n=1$ and a square root function when $n=0.5$. When $n=0.5, k=0.23$, the zero-sublevel set of the SI equals the maximum EFR. Note that when $n=2,k=5$, the EFR is smaller than the IFR because with these parameters, the design rule \eqref{eq: SI design rule} is violated, and thus, forward invariance of the safe set is not guaranteed. Comparing Figure \ref{fig: MPC trajectory SI} with Figure \ref{fig: MPC trajectory CBF}, we observe a difference between SI and CBF: under the SI constraint, the policy does not decelerate until it reaches the boundary of EFR while under the CBF constraint, the policy starts to decelerate before coming close the boundary. This is because of the difference in their constraint design: the SI only requires the next state to stay in the EFR, while CBF further restricts its value-increasing rate.

\begin{figure}
    \centering
    \subfloat{
        \includegraphics[trim=10 20 10 20, width=0.22\linewidth]{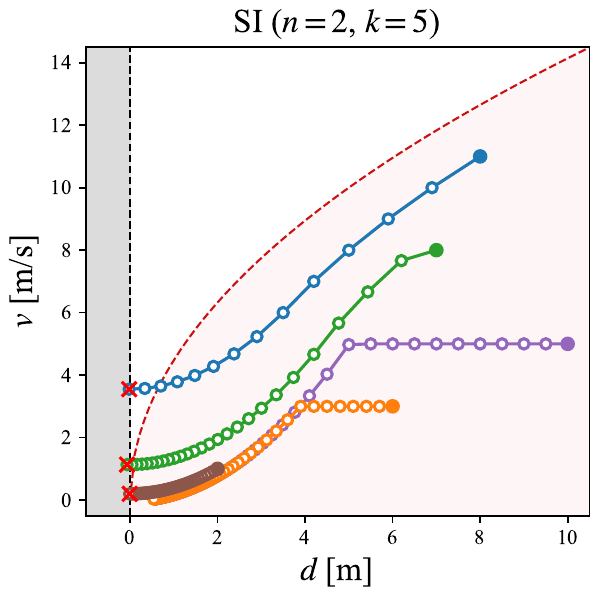}
    }
    \subfloat{
        \includegraphics[trim=10 20 10 20, width=0.22\linewidth]{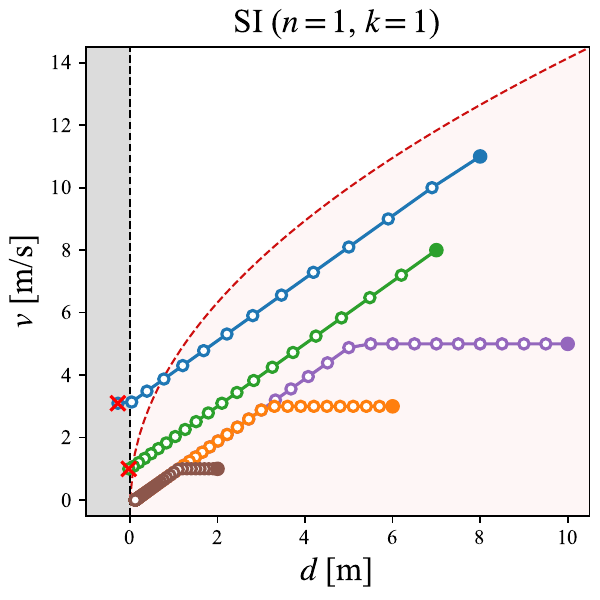}
    }
    \subfloat{
        \includegraphics[trim=10 20 10 20, width=0.22\linewidth]{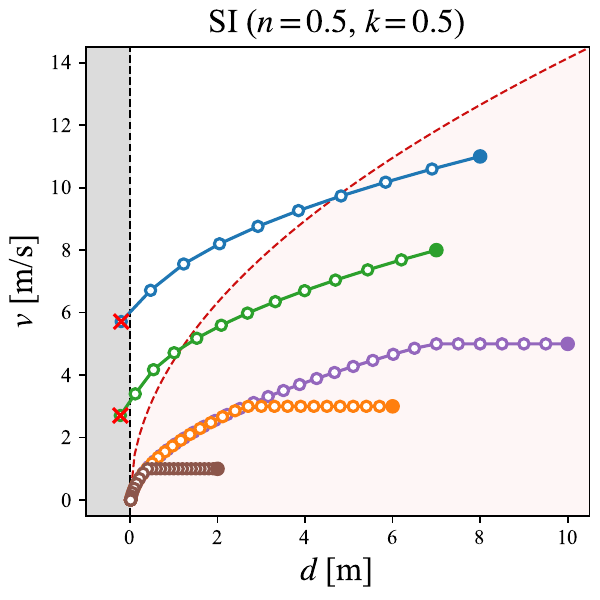}
    }
    \subfloat{
        \includegraphics[trim=10 20 10 20, width=0.22\linewidth]{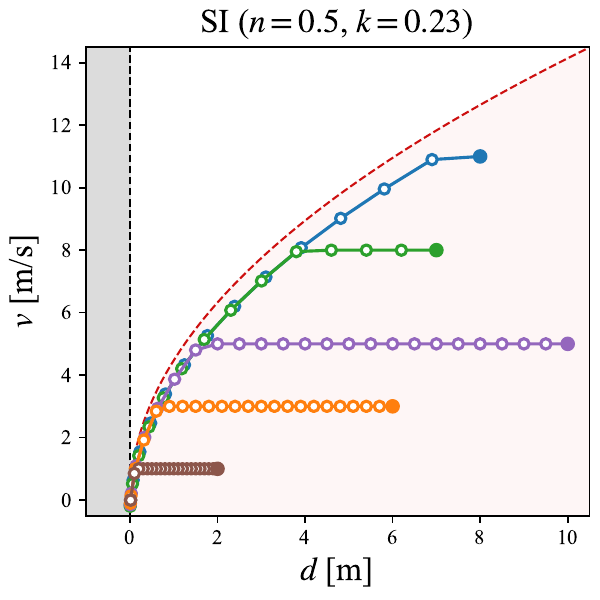}
    }
    \caption{State trajectories of MPC under SI constraints with different parameters.}
    \label{fig: MPC trajectory SI}
\end{figure}

\begin{figure}
    \centering
    \subfloat{
        \includegraphics[trim=10 20 10 20, width=0.22\linewidth]{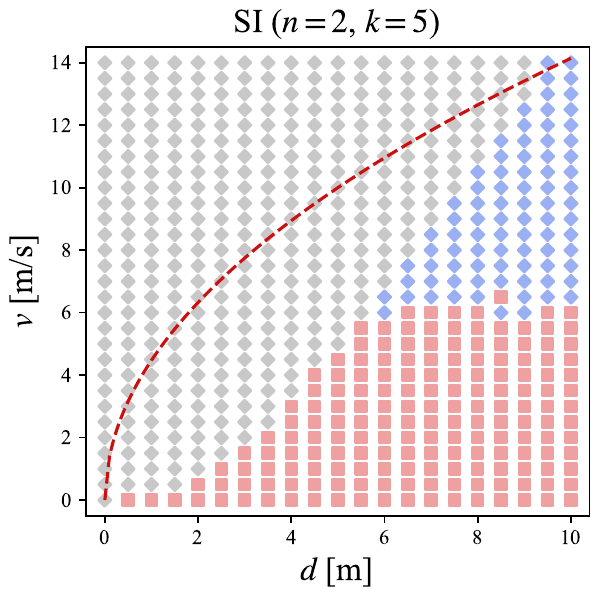}
    }
    \subfloat{
        \includegraphics[trim=10 20 10 20, width=0.22\linewidth]{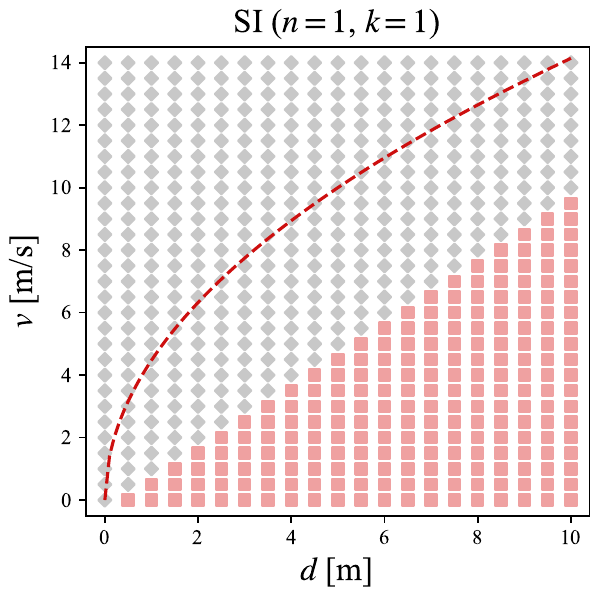}
    }
    \subfloat{
        \includegraphics[trim=10 20 10 20, width=0.22\linewidth]{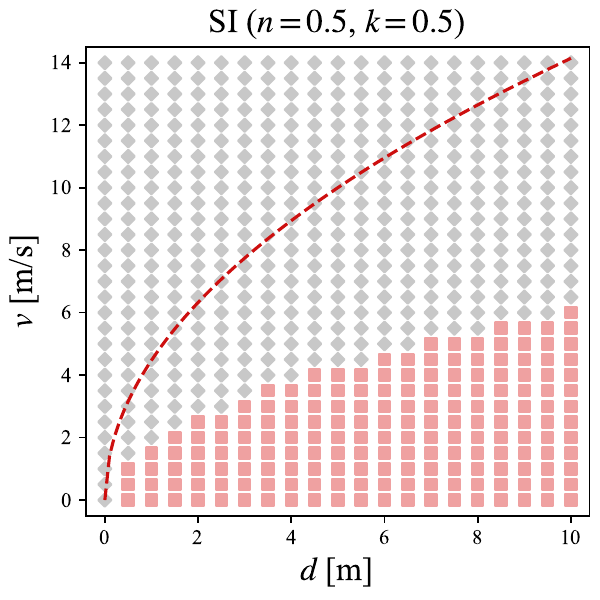}
    }
    \subfloat{
        \includegraphics[trim=10 20 10 20, width=0.22\linewidth]{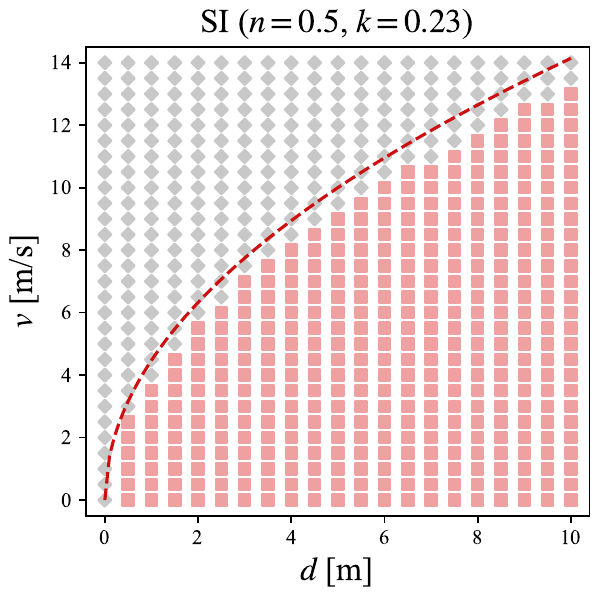}
    }
    \\
    \subfloat{
        \includegraphics[width=0.6\linewidth, trim=130 8 90 188, clip]{Figure/legend_mpc.pdf}
    }
    \caption{Feasible regions of MPC under SI constraints with different parameters.}
    \label{fig: MPC region SI}
\end{figure}

Figure \ref{fig: RL trajectory SI} and \ref{fig: RL region SI} show state trajectories and feasible regions of RL under the same SI constraint at different iterations. The EFR monotonically expands and becomes close to the maximum EFR at iteration 10000, while the state trajectories also gradually converge to those of the optimal policy. We observe that at an early stage of training, the IFR under SI constraint is already quite large, much larger than the EFR. This is different from the case of CBF constraint shown in Figure \ref{fig: RL region CBF}, where EFRs are close to IFRs. The reason is that the SI constraint only requires the next state to be in the safe set without restricting its value-increasing rate as CBF does. At iteration 10 in Figure \ref{fig: RL trajectory SI}, the states with high velocities are initially feasible as long as they do not leave the safe set in one step. Accordingly, states close to the safe set boundary are initially infeasible because their next states will leave the safe set.

\begin{figure}
    \centering
    \subfloat{
        \includegraphics[trim=10 20 10 20, width=0.22\linewidth]{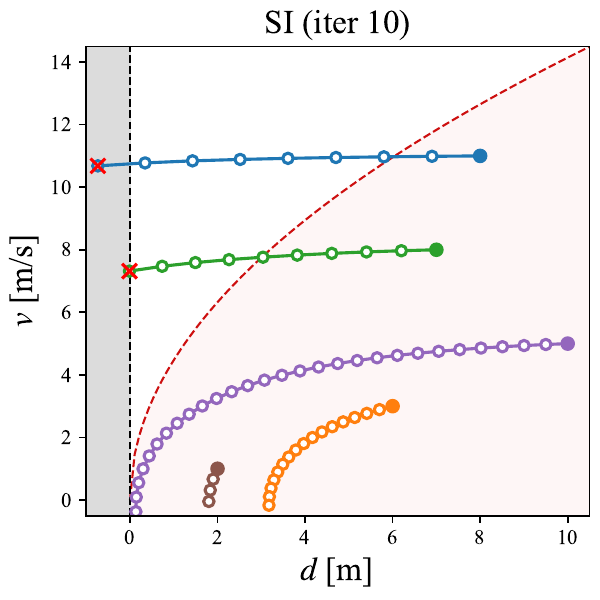}
    }
    \subfloat{
        \includegraphics[trim=10 20 10 20, width=0.22\linewidth]{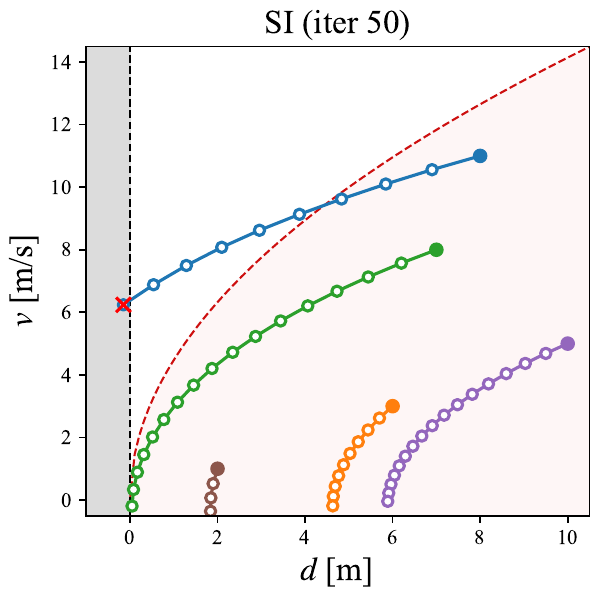}
    }
    \subfloat{
        \includegraphics[trim=10 20 10 20, width=0.22\linewidth]{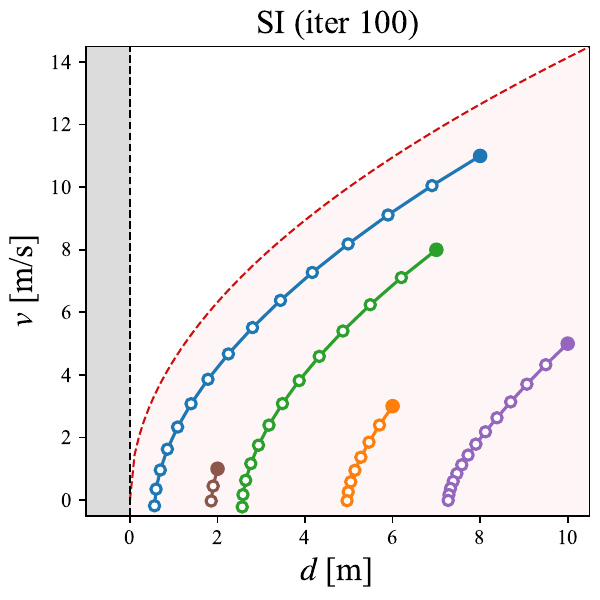}
    }
    \subfloat{
        \includegraphics[trim=10 20 10 20, width=0.22\linewidth]{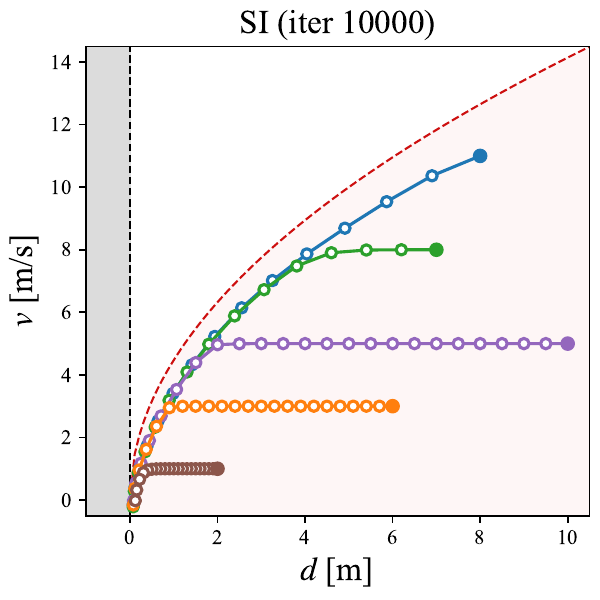}
    }
    \caption{State trajectories of RL under the same SI constraint ($n=0.5, k=0.23$) at different iterations.}
    \label{fig: RL trajectory SI}
\end{figure}

\begin{figure}
    \centering
    \subfloat{
        \includegraphics[trim=10 20 10 20, width=0.22\linewidth]{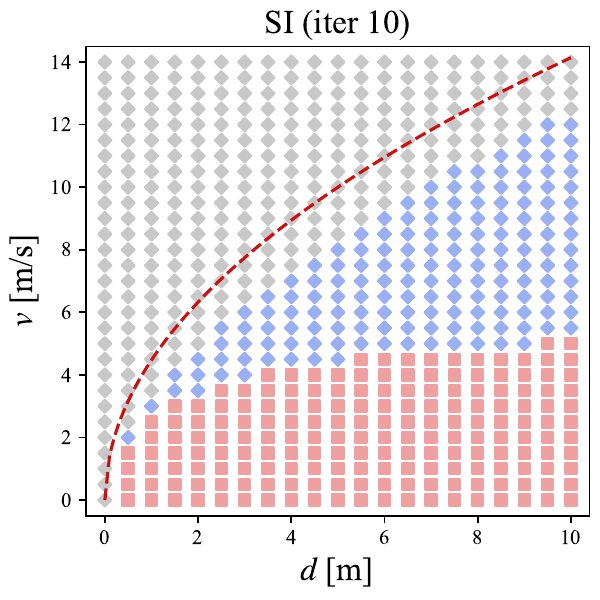}
    }
    \subfloat{
        \includegraphics[trim=10 20 10 20, width=0.22\linewidth]{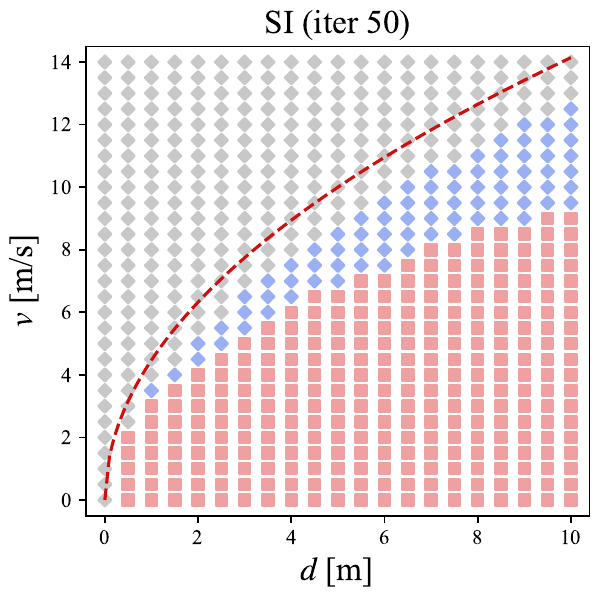}
    }
    \subfloat{
        \includegraphics[trim=10 20 10 20, width=0.22\linewidth]{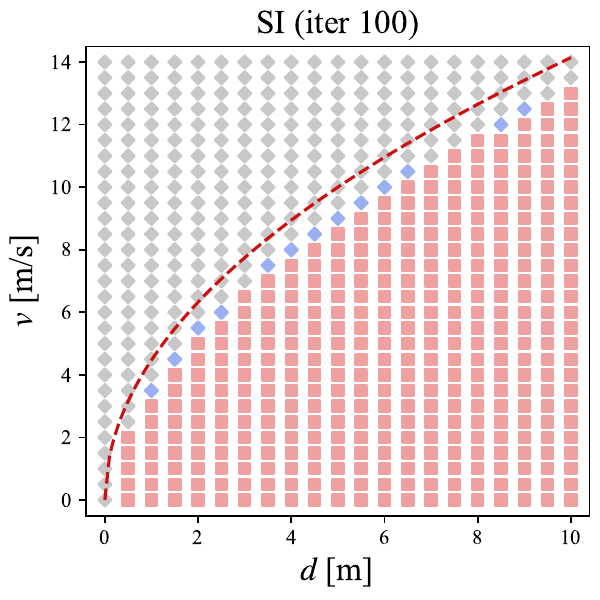}
    }
    \subfloat{
        \includegraphics[trim=10 20 10 20, width=0.22\linewidth]{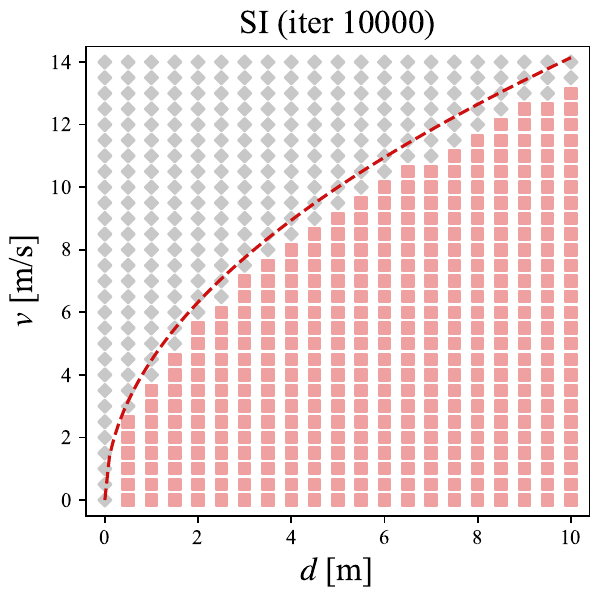}
    }
    \\
    \subfloat{
        \includegraphics[width=0.9\linewidth, trim=30 5 0 185, clip]{Figure/legend_rl.pdf}
    }
    \caption{Feasible regions of RL under the same SI constraint ($n=0.5, k=0.23$) at different iterations.}
    \label{fig: RL region SI}
\end{figure}

Figure \ref{fig: MPC trajectory region HJR} shows state trajectories and feasible regions of MPC under HJ reachability constraint. Being a CA-based virtual-time constraint, the HJ reachability constraint yields identical IFR and EFR of the optimal policy. In addition, since this constraint is built upon the reachability function of a safety-oriented policy, both its IFR and EFR equal the maximum EFR, which coincides with the analysis in Section \ref{sec: HJR}. This indicates that CA-based virtual-time constraints result in the maximum EFR as long as the optimal feasibility function is found. Comparing the right figure in Figure \ref{fig: MPC trajectory region HJR} and the last figure in Figure \ref{fig: MPC trajectory CBF}, we observe that the policy under HJ reachability constraint is less conservative than that under CBF constraint, even though the feasible regions of the two constraints are identical. The reason also lies in the fact that CBF constraint restricts the value-increasing rate while HJ reachability constraint does not.

\begin{figure}
    \centering
    \subfloat{
        \includegraphics[trim=10 20 10 20, width=0.22\linewidth]{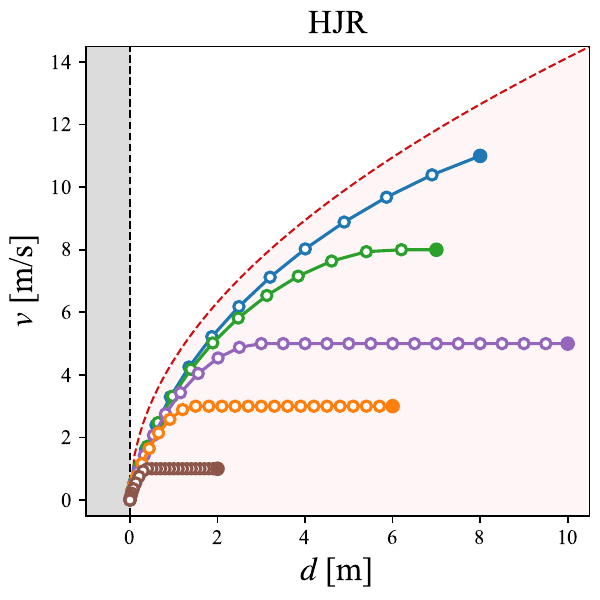}
    }
    \quad
    \subfloat{
        \includegraphics[trim=10 20 10 20, width=0.22\linewidth]{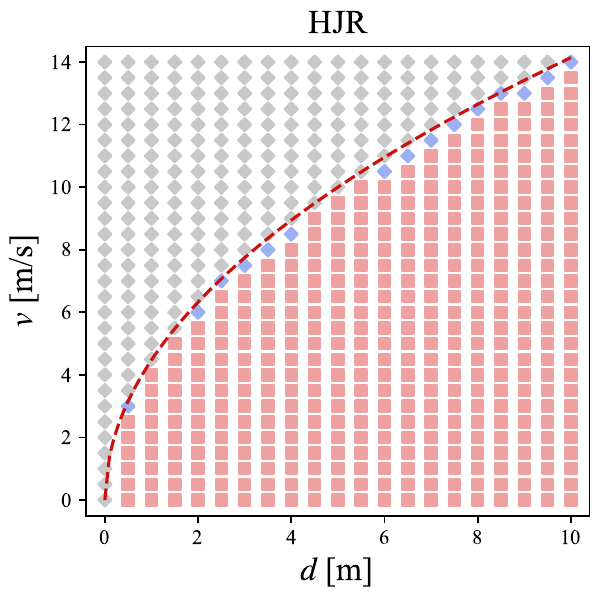}
    }
    \\
    \subfloat{
        \includegraphics[width=0.6\linewidth, trim=130 8 90 188, clip]{Figure/legend_mpc.pdf}
    }
    \caption{State trajectories and feasible regions of MPC under HJ reachability constraint. ``HJR'' stands for ``HJ reachability''.}
    \label{fig: MPC trajectory region HJR}
\end{figure}

Figure \ref{fig: RL trajectory HJR} and \ref{fig: RL region HJR} show state trajectories and feasible regions of RL under HJ reachability constraint at different iterations. As is the case in other constraints, the EFR monotonically expands to the maximum one, and the state trajectories converge to the optimal ones. Like SI, the IFR under HJ reachability constraint is much larger than the EFR at an early stage of training, which is also because of the minimum restrictiveness on the next state.

\begin{figure}
    \centering
    \subfloat{
        \includegraphics[trim=10 20 10 20, width=0.22\linewidth]{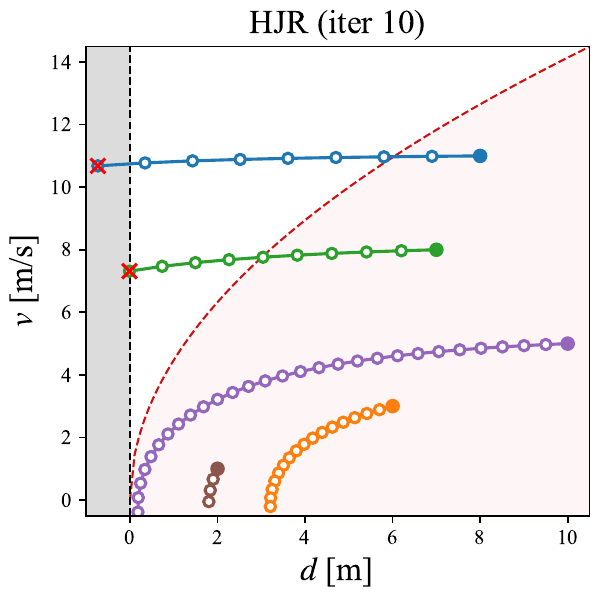}
    }
    \subfloat{
        \includegraphics[trim=10 20 10 20, width=0.22\linewidth]{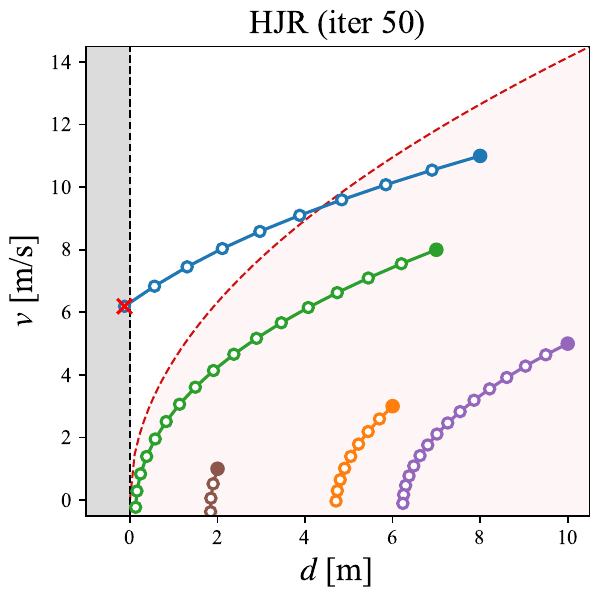}
    }
    \subfloat{
        \includegraphics[trim=10 20 10 20, width=0.22\linewidth]{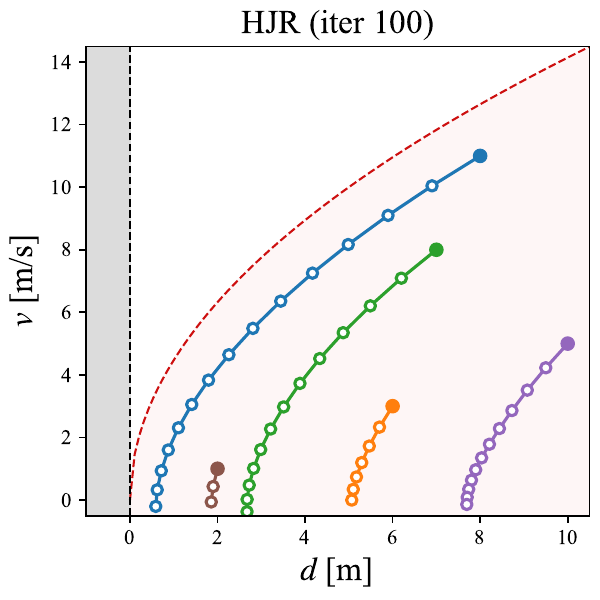}
    }
    \subfloat{
        \includegraphics[trim=10 20 10 20, width=0.22\linewidth]{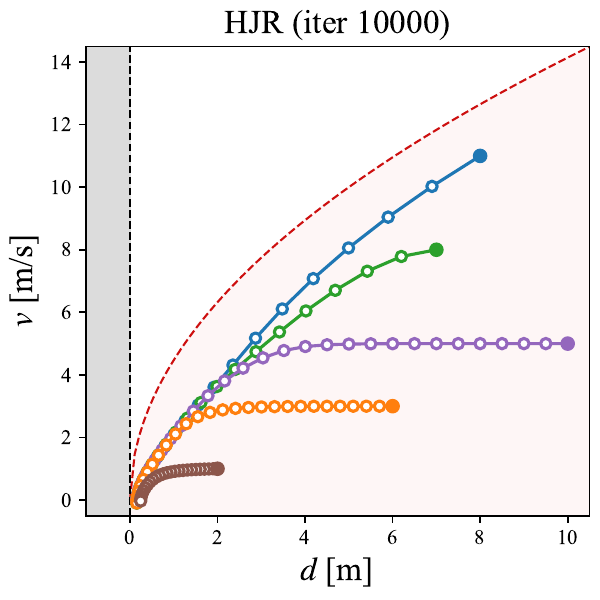}
    }
    \caption{State trajectories of RL under HJR constraint at different iterations.}
    \label{fig: RL trajectory HJR}
\end{figure}

\begin{figure}
    \centering
    \subfloat{
        \includegraphics[trim=10 20 10 20, width=0.22\linewidth]{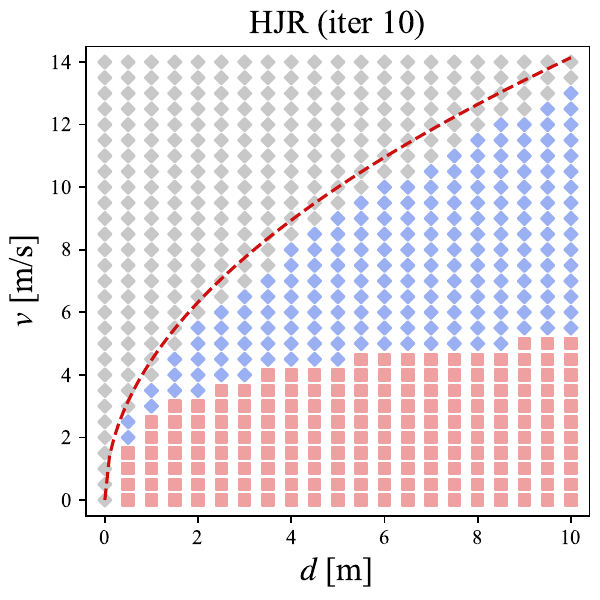}
    }
    \subfloat{
        \includegraphics[trim=10 20 10 20, width=0.22\linewidth]{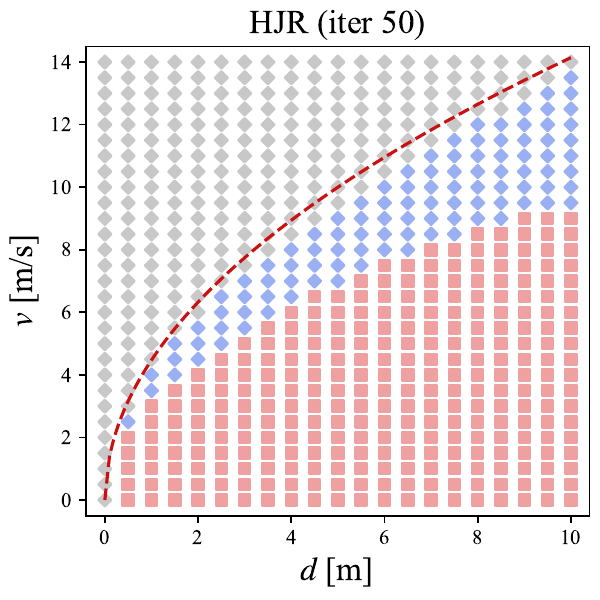}
    }
    \subfloat{
        \includegraphics[trim=10 20 10 20, width=0.22\linewidth]{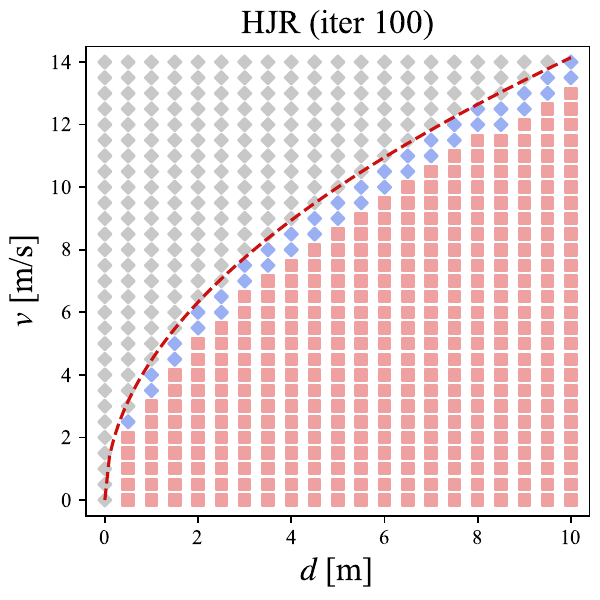}
    }
    \subfloat{
        \includegraphics[trim=10 20 10 20, width=0.22\linewidth]{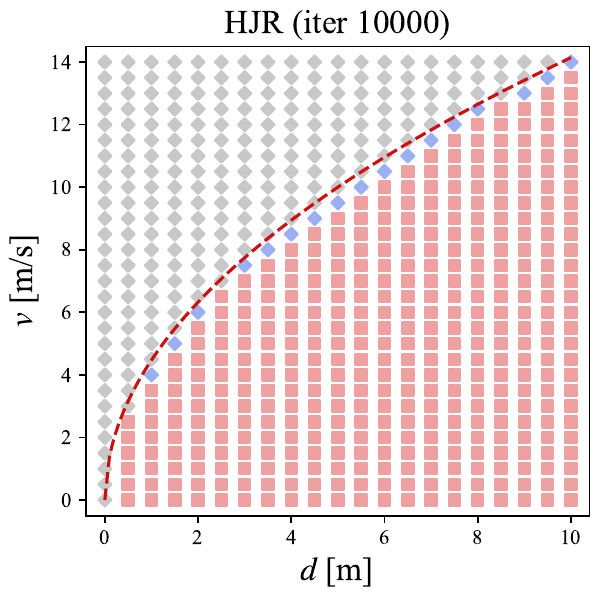}
    }
    \\
    \subfloat{
        \includegraphics[width=0.9\linewidth, trim=30 5 0 185, clip]{Figure/legend_rl.pdf}
    }
    \caption{Feasible regions of RL under HJR constraint at different iterations.}
    \label{fig: RL region HJR}
\end{figure}

\section{\change{Unicycle obstacle avoidance}}
\change{
In this task, a unicycle aims to move forward while avoiding a circular obstacle in the 2D plane, as shown in Figure \ref{fig: unicycle}.
\begin{figure}[h]
    \centering
    \includegraphics[width=0.4\linewidth]{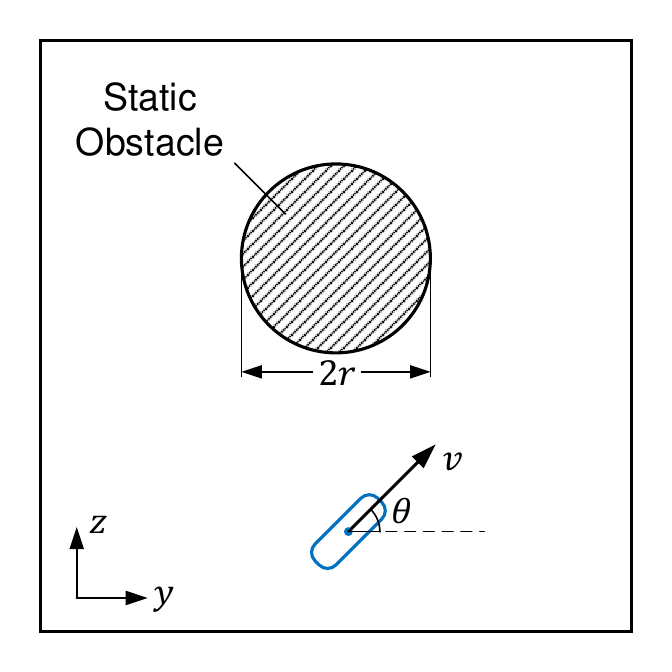}
    \caption{Unicycle obstacle avoidance scenario.}
    \label{fig: unicycle}
\end{figure}
The unicycle has a four-dimensional state $x_t=[y_t,z_t,v_t,\theta_t]^\top$ and a two-dimensional action $u_t=[a_t,\omega_t]^\top$, where $(y,z)$ is the coordinate on the 2D plane, $v$ is the velocity, $\theta$ is the heading angle, $a$ is the acceleration, and $\omega$ is the angular velocity. The system dynamics are given as follows:
\begin{equation}
\begin{aligned}
    y_{t+1} &= y_t + \Delta t \cdot v_t \cos \theta_t \\
    z_{t+1} &= z_t + \Delta t \cdot v_t \sin \theta_t \\
    v_{t+1} &= v_t + \Delta t \cdot a_t \\
    \theta_{t+1} &= \theta_t + \Delta t \cdot \omega_t
\end{aligned}
\end{equation}
The time step is $\Delta t = 0.1\mathrm{s}$.
The actions are bounded as $a \in [-1, 1] \ \mathrm{m/s^2}$ and $\omega \in [-\pi/4, \pi/4] \ \mathrm{rad/s}$.
The reward function is defined as
\begin{equation}
    r(x_t,u_t) = z_t,
\end{equation}
which encourages forward progress in the $z$-direction.
The constraint is defined to avoid collision with the obstacle:
\begin{equation}
    h(x_t) = r - \left\lVert [y_t, z_t]^\top - p_{\mathrm{obs}} \right\rVert_2 \le 0,
\end{equation}
where $r=0.5\mathrm{m}$ and $p_{\mathrm{obs}}=[0,0]^\top$ is the diameter and the center of the obstacle, respectively.
}

\change{
We consider three types of virtual-time constraints: (1) pointwise constraint, (2) CBF constraint, and (3) HJ reachability constraint. The pointwise constraint follows its standard formulation. The CBF is designed as
\begin{equation}
    B(x_{i|t}) = 0.7 - d_{i|t} - k \cdot \dot d_{i|t},
\end{equation}
where $k>0$ is a tunable parameter, and $d_{i|t}$ is the distance between the unicycle and the obstacle, with its rate of change computed as
\begin{equation}
    \dot d_{i|t} = v_{i|t}\cos(\theta_{i|t} - \arctan(z_{i|t}/y_{i|t})).
\end{equation}
The HJ reachability function cannot be computed analytically for this system. Instead, we use a neural network to approximate it and jointly train it with an RL policy using the feasible policy iteration \citep{yang2023feasible} framework. Training the HJ reachability function relies on the risky self-consistenty condition introduced in Section \ref{sec: HJR}. We follow the practice of \citet{fisac2019bridging}, which introduces a discount factor $\gamma\in(0,1)$ to the self-consistency condition to construct a contraction mapping:
\begin{equation}
    F^\pi(x) = (1-\gamma)h(x) + \gamma\max\{h(x),F^\pi(x')\}.
\end{equation}
Iteratively updating $F^\pi(x)$ using this equation converges to the fixed point, which approaches the true HJ reachability function as $\gamma\to1$ \citep{fisac2019bridging}.
}

\change{
In the following, we present the trajectories and feasible regions for MPC under the pointwise and CBF constraints, as well as the results for RL under the HJ reachability constraint.
To visualize the results on the 2D plane, we fix the initial velocity as $v_0=1\mathrm{m/s}$ and the initial heading angle as $\theta_0=\pi/2$.
Figure \ref{fig: unicycle MPC PW} shows the state trajectories and feasible regions of MPC under pointwise constraints with different horizons. Similar to the emergency braking task, short-sighted pointwise constraints lead to smaller EFRs, with the trajectories colliding with the obstacle. As the horizon increases, more trajectories become safe, the EFR expands, and the IFR approaches the EFR.
\begin{figure}[ht]
    \centering
    \subfloat{
        \includegraphics[trim=10 20 10 20, width=0.22\linewidth]{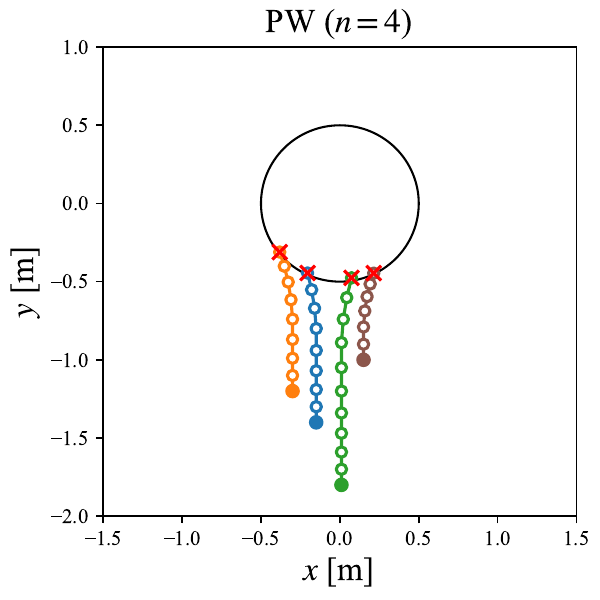}
    }
    \subfloat{
        \includegraphics[trim=10 20 10 20, width=0.22\linewidth]{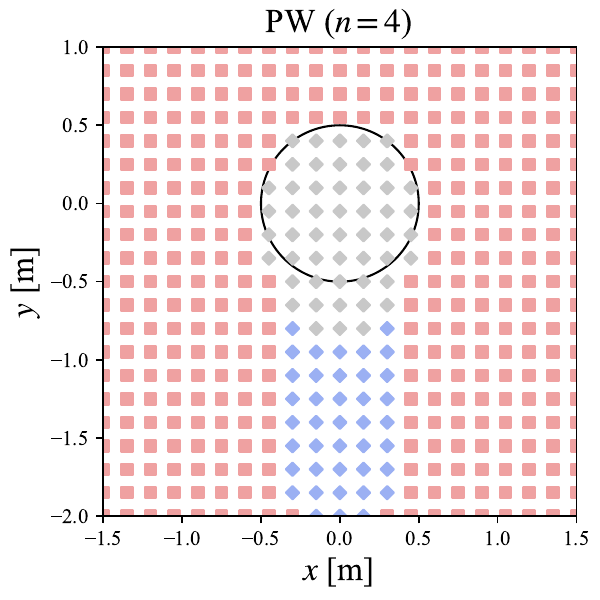}
    }
    \subfloat{
        \includegraphics[trim=10 20 10 20, width=0.22\linewidth]{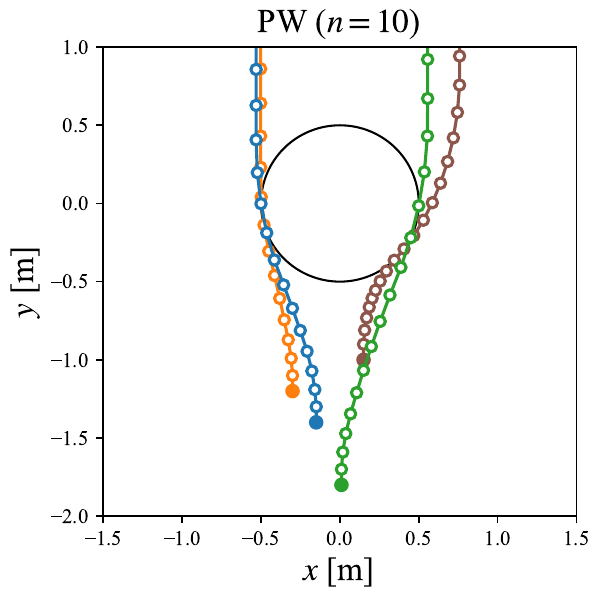}
    }
    \subfloat{
        \includegraphics[trim=10 20 10 20, width=0.22\linewidth]{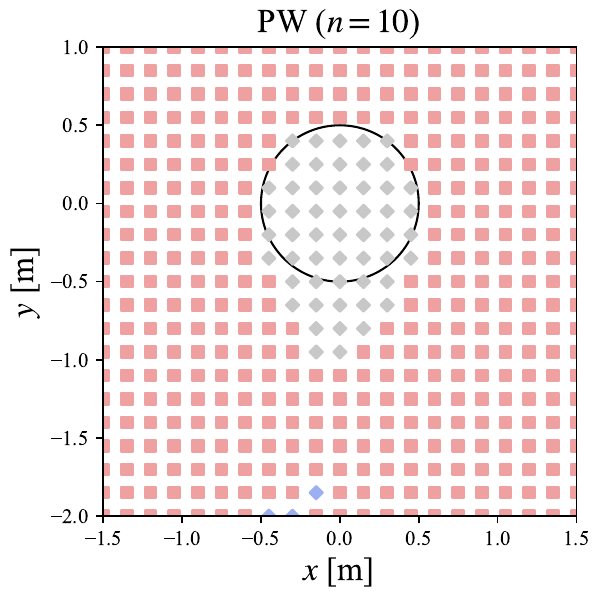}
    }
    \\
    \subfloat{
        \includegraphics[width=0.6\linewidth, trim=130 8 90 188, clip]{Figure/legend_mpc.pdf}
    }
    \caption{State trajectories and feasible regions of MPC under pointwise constraints.}
    \label{fig: unicycle MPC PW}
\end{figure}
}

\change{
Figure \ref{fig: unicycle MPC CBF} shows the state trajectories and feasible regions of MPC under CBF constraints with different parameters. As the parameter $k$ decreases, the EFR expands. It is worth noting that the IFR always equals the EFR under CBF constraints, consistent with our analysis in Section \ref{sec: CBF}. 
An interesting phenomenon is that although the initial states of some trajectories are infeasible, these trajectories can still avoid collision with the obstacle. This is because the virtual CBF constraint is more conservative than the real-time constraint, and violating the CBF constraint does not necessarily lead to collision in the real-time domain. This conservativeness can also be observed from the smaller EFR under CBF constraints compared to pointwise constraints with long horizons in Figure \ref{fig: unicycle MPC PW}.
\begin{figure}[ht]
    \centering
    \subfloat{
        \includegraphics[trim=10 20 10 20, width=0.22\linewidth]{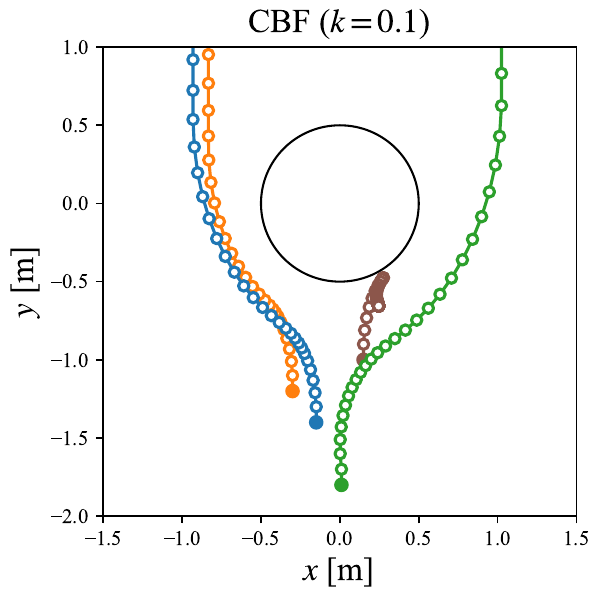}
    }
    \subfloat{
        \includegraphics[trim=10 20 10 20, width=0.22\linewidth]{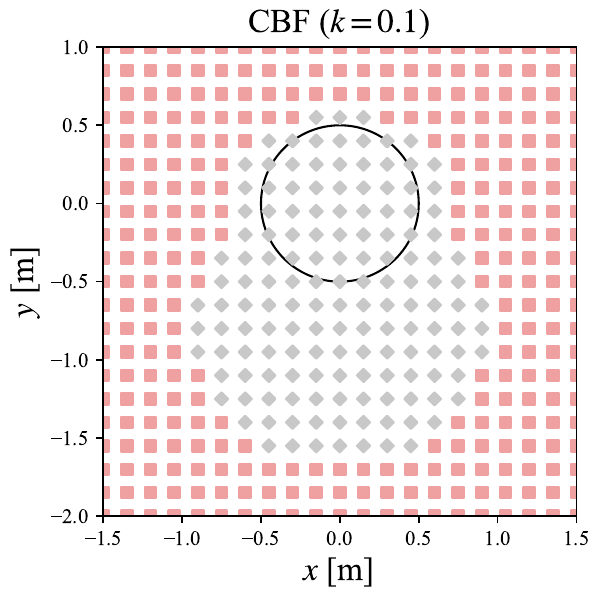}
    }
    \subfloat{
        \includegraphics[trim=10 20 10 20, width=0.22\linewidth]{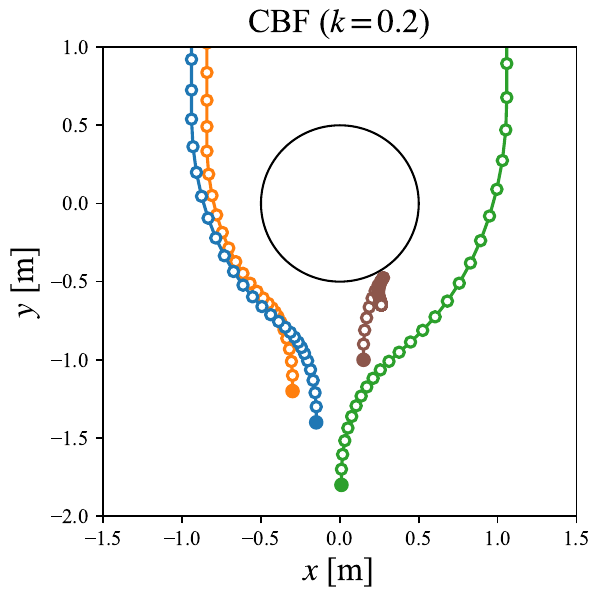}
    }
    \subfloat{
        \includegraphics[trim=10 20 10 20, width=0.22\linewidth]{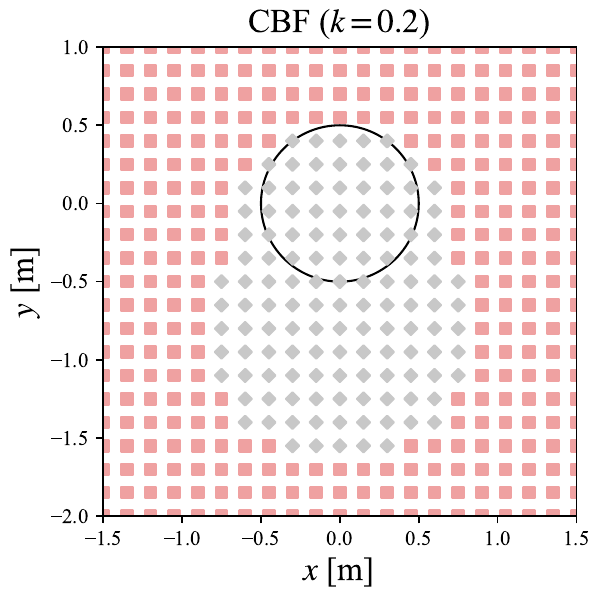}
    }
    \\
    \subfloat{
        \includegraphics[width=0.6\linewidth, trim=130 8 90 188, clip]{Figure/legend_mpc.pdf}
    }
    \caption{State trajectories and feasible regions of RL under HJR constraints.}
    \label{fig: unicycle MPC CBF}
\end{figure}
}

\change{
Figure \ref{fig: unicycle RL} shows the state trajectories and feasible regions of RL under the HJ reachability constraint at convergence. The EFR and IFR are almost identical, consistent with our analysis in Section \ref{sec: HJR}. The trajectories successfully avoid collision with the obstacle and the EFR is significantly larger than that under CBF constraints, close to that of the pointwise constraint with a long horizon. This demonstrates the effectiveness of the learning-based method for approximating the optimal feasibility function. It also shows the advantage of CA-type learning-based feasibility functions: they can achieve larger feasible regions with less conservativeness compared to CIS-type analytical feasibility functions while avoiding the computational complexity of long-horizon pointwise constraints.
\begin{figure}[ht]
    \centering
    \subfloat{
        \includegraphics[trim=10 20 10 20, width=0.22\linewidth]{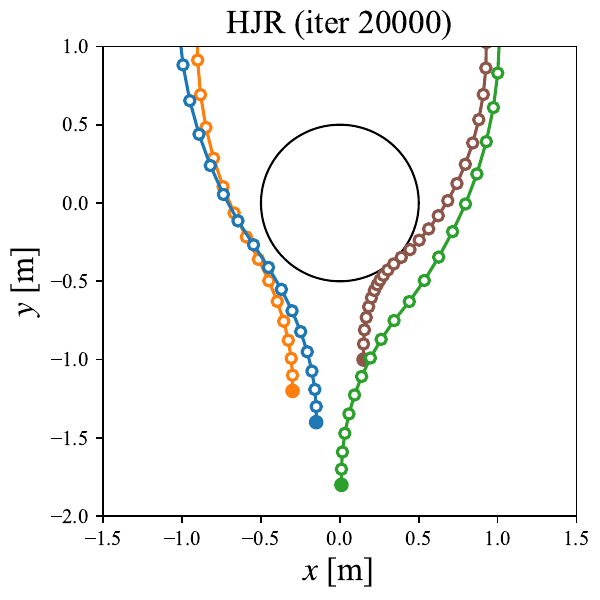}
    }
    \quad
    \subfloat{
        \includegraphics[trim=10 20 10 20, width=0.22\linewidth]{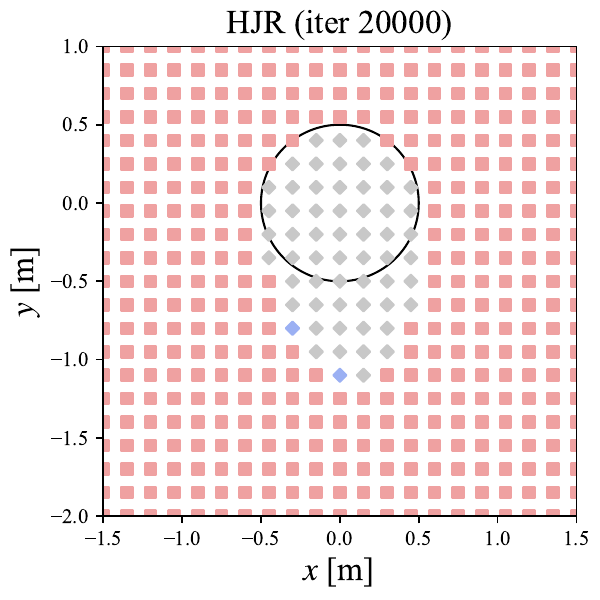}
    }
    \\
    \subfloat{
        \includegraphics[width=0.9\linewidth, trim=30 5 0 185, clip]{Figure/legend_rl.pdf}
    }
    \caption{State trajectories and feasible regions of MPC under CBF constraints.}
    \label{fig: unicycle RL}
\end{figure}
}

\chapter{Conclusion}
\label{conclusion}
This paper proposes a feasibility theory that applies to both MPC and RL by decoupling states, constraints, and policies in constrained OCPs. We reveal that feasibility depends on the combination of these three elements, extending the traditional viewpoint of OCP-specific feasibility theories in MPC. Based on separated definitions of initial and endless, state and policy feasibility, we analyze the containment relationships between different feasible regions, which enables us to describe feasibility under arbitrary combinations of states, constraints, and policies. After that, we provide virtual-time constraint design rules along with a practical design tool called feasibility function that helps achieve the maximum feasible region. The feasibility function is categorized into control invariant set and constraint aggregation, which lead to different kinds of virtual-time constraints. We review most existing constraint formulations and point out that they are essentially applications of feasibility functions in different forms. Finally, we demonstrate our feasibility theory in an emergency braking control task by visualizing initially and endlessly feasible regions of MPC and RL policies under different kinds of virtual-time constraints.

\begin{acknowledgements}
Yujie Yang, Zhilong Zheng, and Shengbo Eben Li are supported by the National Natural Science Foundation of China (Grant No. 92582205) and the Beijing Natural Science Foundation (Grant No. L257002).
\end{acknowledgements}

%BACKMATTER SEE DOCUMENTATION
\backmatter  % references, restarts sample

\printbibliography

@article{brunke2022safe,
  title={Safe learning in robotics: From learning-based control to safe reinforcement learning},
  author={Brunke, Lukas and Greeff, Melissa and Hall, Adam W and Yuan, Zhaocong and Zhou, Siqi and Panerati, Jacopo and Schoellig, Angela P},
  journal={Annual Review of Control, Robotics, and Autonomous Systems},
  volume={5},
  pages={411--444},
  year={2022},
  publisher={Annual Reviews}
}

@inproceedings{ravaioli2022safe,
  title={Safe reinforcement learning benchmark environments for aerospace control systems},
  author={Ravaioli, Umberto J and Cunningham, James and McCarroll, John and Gangal, Vardaan and Dunlap, Kyle and Hobbs, Kerianne L},
  booktitle={2022 IEEE Aerospace Conference (AERO)},
  pages={1--20},
  year={2022},
  organization={IEEE}
}

@article{guan2022integrated,
  title={Integrated decision and control: Toward interpretable and computationally efficient driving intelligence},
  author={Guan, Yang and Ren, Yangang and Sun, Qi and Li, Shengbo Eben and Ma, Haitong and Duan, Jingliang and Dai, Yifan and Cheng, Bo},
  journal={IEEE transactions on cybernetics},
  volume={53},
  number={2},
  pages={859--873},
  year={2022},
  publisher={IEEE}
}

@book{borrelli2017predictive,
  title={Predictive control for linear and hybrid systems},
  author={Borrelli, Francesco and Bemporad, Alberto and Morari, Manfred},
  year={2017},
  publisher={Cambridge University Press}
}

@article{zhang2016switched,
  title={Switched model predictive control of switched linear systems: Feasibility, stability and robustness},
  author={Zhang, Lixian and Zhuang, Songlin and Braatz, Richard D},
  journal={Automatica},
  volume={67},
  pages={8--21},
  year={2016},
  publisher={Elsevier}
}

@article{lofberg2012oops,
  title={Oops! I cannot do it again: Testing for recursive feasibility in MPC},
  author={L{\"o}fberg, Johan},
  journal={Automatica},
  volume={48},
  number={3},
  pages={550--555},
  year={2012},
  publisher={Elsevier}
}

@article{boccia2014stability,
  title={Stability and feasibility of state constrained MPC without stabilizing terminal constraints},
  author={Boccia, Andrea and Gr{\"u}ne, Lars and Worthmann, Karl},
  journal={Systems \& control letters},
  volume={72},
  pages={14--21},
  year={2014},
  publisher={Elsevier}
}

@article{gondhalekar2009controlled,
  title={Controlled invariant feasibility—a general approach to enforcing strong feasibility in MPC applied to move-blocking},
  author={Gondhalekar, Ravi and Imura, Jun-ichi and Kashima, Kenji},
  journal={Automatica},
  volume={45},
  number={12},
  pages={2869--2875},
  year={2009},
  publisher={Elsevier}
}

@article{gondhalekar2011mpc,
  title={MPC of constrained discrete-time linear periodic systems—A framework for asynchronous control: Strong feasibility, stability and optimality via periodic invariance},
  author={Gondhalekar, Ravi and Jones, Colin N},
  journal={Automatica},
  volume={47},
  number={2},
  pages={326--333},
  year={2011},
  publisher={Elsevier}
}

@book{altman1999constrained,
  title={Constrained Markov decision processes},
  author={Altman, Eitan},
  volume={7},
  year={1999},
  publisher={CRC press}
}

@article{chow2018lyapunovbased,
  title = {A Lyapunov-Based Approach to Safe Reinforcement Learning},
  author = {Chow, Yinlam and Nachum, Ofir and Duenez-Guzman, Edgar and Ghavamzadeh, Mohammad},
  year = {2018},
  journal = {Advances in neural information processing systems},
  volume = {31}
}

@article{ding2020natural,
  title = {Natural Policy Gradient Primal-Dual Method for Constrained Markov Decision Processes},
  author = {Ding, Dongsheng and Zhang, Kaiqing and Basar, Tamer and Jovanovic, Mihailo},
  year = {2020},
  journal = {Advances in Neural Information Processing Systems},
  volume = {33},
  pages = {8378--8390}
}

@inproceedings{bai2022achieving,
  title = {Achieving Zero Constraint Violation for Constrained Reinforcement Learning via Primal-Dual Approach},
  booktitle = {Proceedings of the AAAI Conference on Artificial Intelligence},
  author = {Bai, Qinbo and Bedi, Amrit Singh and Agarwal, Mridul and Koppel, Alec and Aggarwal, Vaneet},
  year = {2022},
  volume = {36},
  number = {4},
  pages = {3682--3689}
}

@inproceedings{liu2022constrained,
  title = {Constrained Variational Policy Optimization for Safe Reinforcement Learning},
  booktitle = {International Conference on Machine Learning},
  author = {Liu, Zuxin and Cen, Zhepeng and Isenbaev, Vladislav and Liu, Wei and Wu, Steven and Li, Bo and Zhao, Ding},
  year = {2022},
  pages = {13644--13668},
  publisher = {PMLR}
}

@inproceedings{ying2022dual,
  title = {A Dual Approach to Constrained Markov Decision Processes with Entropy Regularization},
  booktitle = {International Conference on Artificial Intelligence and Statistics},
  author = {Ying, Donghao and Ding, Yuhao and Lavaei, Javad},
  year = {2022},
  pages = {1887--1909},
  publisher = {PMLR}
}

@article{yu2022towards,
  title={Towards safe reinforcement learning with a safety editor policy},
  author={Yu, Haonan and Xu, Wei and Zhang, Haichao},
  journal={Advances in Neural Information Processing Systems},
  volume={35},
  pages={2608--2621},
  year={2022}
}

@inproceedings{liu2020ipo,
  title = {IPO: Interior-Point Policy Optimization under Constraints},
  shorttitle = {IPO},
  booktitle = {Proceedings of the AAAI Conference on Artificial Intelligence},
  author = {Liu, Yongshuai and Ding, Jiaxin and Liu, Xin},
  year = {2020},
  volume = {34},
  number = {04},
  pages = {4940--4947}
}

@inproceedings{stooke2020responsive,
  title = {Responsive Safety in Reinforcement Learning by Pid Lagrangian Methods},
  booktitle = {International Conference on Machine Learning},
  author = {Stooke, Adam and Achiam, Joshua and Abbeel, Pieter},
  year = {2020},
  pages = {9133--9143},
  publisher = {PMLR}
}

@article{peng2022model,
  title={Model-Based Chance-Constrained Reinforcement Learning via Separated Proportional-Integral Lagrangian},
  author={Peng, Baiyu and Duan, Jingliang and Chen, Jianyu and Li, Shengbo Eben and Xie, Genjin and Zhang, Congsheng and Guan, Yang and Mu, Yao and Sun, Enxin},
  journal={IEEE Transactions on Neural Networks and Learning Systems},
  year={2022},
  publisher={IEEE}
}

@inproceedings{yang2020projection,
  title={Projection-Based Constrained Policy Optimization},
  author={Yang, Tsung-Yen and Rosca, Justinian and Narasimhan, Karthik and Ramadge, Peter J},
  booktitle={International Conference on Learning Representations},
  year={2020},
}

@article{zhang2020first,
  title = {First Order Constrained Optimization in Policy Space},
  author = {Zhang, Yiming and Vuong, Quan and Ross, Keith},
  year = {2020},
  journal = {Advances in Neural Information Processing Systems},
  volume = {33},
  pages = {15338--15349}
}

@inproceedings{achiam2017constrained,
  title = {Constrained Policy Optimization},
  booktitle = {International Conference on Machine Learning},
  author = {Achiam, Joshua and Held, David and Tamar, Aviv and Abbeel, Pieter},
  year = {2017},
  pages = {22--31},
  publisher = {PMLR}
}

@article{chow2017risk,
  title={Risk-constrained reinforcement learning with percentile risk criteria},
  author={Chow, Yinlam and Ghavamzadeh, Mohammad and Janson, Lucas and Pavone, Marco},
  journal={The Journal of Machine Learning Research},
  volume={18},
  number={1},
  pages={6070--6120},
  year={2017},
  publisher={JMLR. org}
}

@inproceedings{ames2014control,
  title={Control barrier function based quadratic programs with application to adaptive cruise control},
  author={Ames, Aaron D and Grizzle, Jessy W and Tabuada, Paulo},
  booktitle={53rd IEEE Conference on Decision and Control},
  pages={6271--6278},
  year={2014},
  organization={IEEE}
}

@article{ames2016control,
  title={Control barrier function based quadratic programs for safety critical systems},
  author={Ames, Aaron D and Xu, Xiangru and Grizzle, Jessy W and Tabuada, Paulo},
  journal={IEEE Transactions on Automatic Control},
  volume={62},
  number={8},
  pages={3861--3876},
  year={2016},
  publisher={IEEE}
}

@inproceedings{nguyen2016exponential,
  title={Exponential control barrier functions for enforcing high relative-degree safety-critical constraints},
  author={Nguyen, Quan and Sreenath, Koushil},
  booktitle={2016 American Control Conference (ACC)},
  pages={322--328},
  year={2016},
  organization={IEEE}
}

@inproceedings{agrawal2017discrete,
  title={Discrete control barrier functions for safety-critical control of discrete systems with application to bipedal robot navigation.},
  author={Agrawal, Ayush and Sreenath, Koushil},
  booktitle={Robotics: Science and Systems},
  volume={13},
  pages={1--10},
  year={2017},
  organization={Cambridge, MA, USA}
}

@inproceedings{ames2019control,
  title={Control barrier functions: Theory and applications},
  author={Ames, Aaron D and Coogan, Samuel and Egerstedt, Magnus and Notomista, Gennaro and Sreenath, Koushil and Tabuada, Paulo},
  booktitle={2019 18th European control conference (ECC)},
  pages={3420--3431},
  year={2019},
  organization={IEEE}
}

@article{ohnishi2019barrier,
  title={Barrier-certified adaptive reinforcement learning with applications to brushbot navigation},
  author={Ohnishi, Motoya and Wang, Li and Notomista, Gennaro and Egerstedt, Magnus},
  journal={IEEE Transactions on robotics},
  volume={35},
  number={5},
  pages={1186--1205},
  year={2019},
  publisher={IEEE}
}

@inproceedings{cheng2019end,
  title={End-to-end safe reinforcement learning through barrier functions for safety-critical continuous control tasks},
  author={Cheng, Richard and Orosz, G{\'a}bor and Murray, Richard M and Burdick, Joel W},
  booktitle={Proceedings of the AAAI conference on artificial intelligence},
  volume={33},
  number={01},
  pages={3387--3395},
  year={2019}
}

@inproceedings{taylor2020learning,
  title={Learning for safety-critical control with control barrier functions},
  author={Taylor, Andrew and Singletary, Andrew and Yue, Yisong and Ames, Aaron},
  booktitle={Learning for Dynamics and Control},
  pages={708--717},
  year={2020},
  organization={PMLR}
}

@inproceedings{qin2020learning,
  title={Learning Safe Multi-agent Control with Decentralized Neural Barrier Certificates},
  author={Qin, Zengyi and Zhang, Kaiqing and Chen, Yuxiao and Chen, Jingkai and Fan, Chuchu},
  booktitle={International Conference on Learning Representations},
  year={2020}
}

@inproceedings{robey2020learning,
  title={Learning control barrier functions from expert demonstrations},
  author={Robey, Alexander and Hu, Haimin and Lindemann, Lars and Zhang, Hanwen and Dimarogonas, Dimos V and Tu, Stephen and Matni, Nikolai},
  booktitle={2020 59th IEEE Conference on Decision and Control (CDC)},
  pages={3717--3724},
  year={2020},
  organization={IEEE}
}

@inproceedings{ma2021model,
  title={Model-based constrained reinforcement learning using generalized control barrier function},
  author={Ma, Haitong and Chen, Jianyu and Eben, Shengbo and Lin, Ziyu and Guan, Yang and Ren, Yangang and Zheng, Sifa},
  booktitle={2021 IEEE/RSJ International Conference on Intelligent Robots and Systems (IROS)},
  pages={4552--4559},
  year={2021},
  organization={IEEE}
}

@article{yang2023model,
  title={Model-free safe reinforcement learning through neural barrier certificate},
  author={Yang, Yujie and Jiang, Yuxuan and Liu, Yichen and Chen, Jianyu and Li, Shengbo Eben},
  journal={IEEE Robotics and Automation Letters},
  volume={8},
  number={3},
  pages={1295--1302},
  year={2023},
  publisher={IEEE}
}

@inproceedings{liu2014control,
  title={Control in a safe set: Addressing safety in human-robot interactions},
  author={Liu, Changliu and Tomizuka, Masayoshi},
  booktitle={Dynamic Systems and Control Conference},
  volume={46209},
  pages={V003T42A003},
  year={2014},
  organization={American Society of Mechanical Engineers}
}

@inproceedings{ma2022joint,
  title={Joint synthesis of safety certificate and safe control policy using constrained reinforcement learning},
  author={Ma, Haitong and Liu, Changliu and Li, Shengbo Eben and Zheng, Sifa and Chen, Jianyu},
  booktitle={Learning for Dynamics and Control Conference},
  pages={97--109},
  year={2022},
  organization={PMLR}
}

@inproceedings{zhao2021model,
  title={Model-free safe control for zero-violation reinforcement learning},
  author={Zhao, Weiye and He, Tairan and Liu, Changliu},
  booktitle={5th Annual Conference on Robot Learning},
  year={2021}
}

@article{wei2022persistently,
  title={Persistently feasible robust safe control by safety index synthesis and convex semi-infinite programming},
  author={Wei, Tianhao and Kang, Shucheng and Zhao, Weiye and Liu, Changliu},
  journal={IEEE Control Systems Letters},
  volume={7},
  pages={1213--1218},
  year={2022},
  publisher={IEEE}
}

@inproceedings{zhao2023probabilistic,
  title={Probabilistic safeguard for reinforcement learning using safety index guided gaussian process models},
  author={Zhao, Weiye and He, Tairan and Liu, Changliu},
  booktitle={Learning for Dynamics and Control Conference},
  pages={783--796},
  year={2023},
  organization={PMLR}
}

@inproceedings{zhao2023safety,
  title={Safety index synthesis via sum-of-squares programming},
  author={Zhao, Weiye and He, Tairan and Wei, Tianhao and Liu, Simin and Liu, Changliu},
  booktitle={2023 American Control Conference (ACC)},
  pages={732--737},
  year={2023},
  organization={IEEE}
}

@inproceedings{fisac2019bridging,
  title={Bridging hamilton-jacobi safety analysis and reinforcement learning},
  author={Fisac, Jaime F and Lugovoy, Neil F and Rubies-Royo, Vicen{\c{c}} and Ghosh, Shromona and Tomlin, Claire J},
  booktitle={2019 International Conference on Robotics and Automation (ICRA)},
  pages={8550--8556},
  year={2019},
  organization={IEEE}
}

@inproceedings{yu2022reachability,
  title={Reachability constrained reinforcement learning},
  author={Yu, Dongjie and Ma, Haitong and Li, Shengbo and Chen, Jianyu},
  booktitle={International Conference on Machine Learning},
  pages={25636--25655},
  year={2022},
  organization={PMLR}
}

@article{thananjeyan2021recovery,
  title={Recovery rl: Safe reinforcement learning with learned recovery zones},
  author={Thananjeyan, Brijen and Balakrishna, Ashwin and Nair, Suraj and Luo, Michael and Srinivasan, Krishnan and Hwang, Minho and Gonzalez, Joseph E and Ibarz, Julian and Finn, Chelsea and Goldberg, Ken},
  journal={IEEE Robotics and Automation Letters},
  volume={6},
  number={3},
  pages={4915--4922},
  year={2021},
  publisher={IEEE}
}

@article{mitchell2005time,
  title={A time-dependent Hamilton-Jacobi formulation of reachable sets for continuous dynamic games},
  author={Mitchell, Ian M and Bayen, Alexandre M and Tomlin, Claire J},
  journal={IEEE Transactions on automatic control},
  volume={50},
  number={7},
  pages={947--957},
  year={2005},
  publisher={IEEE}
}

@inproceedings{bansal2017hamilton,
  title={Hamilton-jacobi reachability: A brief overview and recent advances},
  author={Bansal, Somil and Chen, Mo and Herbert, Sylvia and Tomlin, Claire J},
  booktitle={2017 IEEE 56th Annual Conference on Decision and Control (CDC)},
  pages={2242--2253},
  year={2017},
  organization={IEEE}
}

@inproceedings{bansal2021deepreach,
  title={Deepreach: A deep learning approach to high-dimensional reachability},
  author={Bansal, Somil and Tomlin, Claire J},
  booktitle={2021 IEEE International Conference on Robotics and Automation (ICRA)},
  pages={1817--1824},
  year={2021},
  organization={IEEE}
}

@inproceedings{herbert2021scalable,
  title={Scalable learning of safety guarantees for autonomous systems using Hamilton-Jacobi reachability},
  author={Herbert, Sylvia and Choi, Jason J and Sanjeev, Suvansh and Gibson, Marsalis and Sreenath, Koushil and Tomlin, Claire J},
  booktitle={2021 IEEE International Conference on Robotics and Automation (ICRA)},
  pages={5914--5920},
  year={2021},
  organization={IEEE}
}

@inproceedings{seo2019robust,
  title={Robust trajectory planning for a multirotor against disturbance based on hamilton-jacobi reachability analysis},
  author={Seo, Hoseong and Lee, Donggun and Son, Clark Youngdong and Tomlin, Claire J and Kim, H Jin},
  booktitle={2019 IEEE/RSJ International Conference on Intelligent Robots and Systems (IROS)},
  pages={3150--3157},
  year={2019},
  organization={IEEE}
}

@inproceedings{rubies2019classification,
  title={A classification-based approach for approximate reachability},
  author={Rubies-Royo, Vicen{\c{c}} and Fridovich-Keil, David and Herbert, Sylvia and Tomlin, Claire J},
  booktitle={2019 International Conference on Robotics and Automation (ICRA)},
  pages={7697--7704},
  year={2019},
  organization={IEEE}
}

@book{li2023reinforcement,
  title={Reinforcement learning for sequential decision and optimal control},
  author={Li, Shengbo Eben},
  year={2023},
  publisher={Springer}
}

@inproceedings{as2022constrained,
  title={Constrained Policy Optimization via Bayesian World Models},
  author={As, Yarden and Usmanova, Ilnura and Curi, Sebastian and Krause, Andreas},
  booktitle={International Conference on Learning Representations},
  year={2022}
}

@article{yang2023feasible,
  title={Feasible Policy Iteration},
  author={Yang, Yujie and Zheng, Zhilong and Li, Shengbo Eben},
  journal={arXiv preprint arXiv:2304.08845},
  year={2023}
}

@article{zhao2023state,
  title={State-wise Constrained Policy Optimization},
  author={Zhao, Weiye and Chen, Rui and Sun, Yifan and Wei, Tianhao and Liu, Changliu},
  journal={arXiv preprint arXiv:2306.12594},
  year={2023}
}

@inproceedings{xiao2019control,
  title={Control barrier functions for systems with high relative degree},
  author={Xiao, Wei and Belta, Calin},
  booktitle={2019 IEEE 58th conference on decision and control (CDC)},
  pages={474--479},
  year={2019},
  organization={IEEE}
}

@article{xiao2021high,
  title={High-order control barrier functions},
  author={Xiao, Wei and Belta, Calin},
  journal={IEEE Transactions on Automatic Control},
  volume={67},
  number={7},
  pages={3655--3662},
  year={2021},
  publisher={IEEE}
}

@ARTICLE{yang2024synthesizing,
  author={Yang, Yujie and Zhang, Yuhang and Zou, Wenjun and Chen, Jianyu and Yin, Yuming and Eben Li, Shengbo},
  journal={IEEE Transactions on Automatic Control}, 
  title={Synthesizing Control Barrier Functions With Feasible Region Iteration for Safe Reinforcement Learning}, 
  year={2024},
  volume={69},
  number={4},
  pages={2713-2720},
}

@article{scokaert1999suboptimal,
  title={Suboptimal model predictive control (feasibility implies stability)},
  author={Scokaert, Pierre OM and Mayne, David Q and Rawlings, James B},
  journal={IEEE Transactions on Automatic Control},
  volume={44},
  number={3},
  pages={648--654},
  year={1999},
  publisher={IEEE}
}

@article{pannocchia2011conditions,
  title={Conditions under which suboptimal nonlinear MPC is inherently robust},
  author={Pannocchia, Gabriele and Rawlings, James B and Wright, Stephen J},
  journal={Systems \& Control Letters},
  volume={60},
  number={9},
  pages={747--755},
  year={2011},
  publisher={Elsevier}
}

@article{xiao2023barriernet,
  title={Barriernet: Differentiable control barrier functions for learning of safe robot control},
  author={Xiao, Wei and Wang, Tsun-Hsuan and Hasani, Ramin and Chahine, Makram and Amini, Alexander and Li, Xiao and Rus, Daniela},
  journal={IEEE Transactions on Robotics},
  volume={39},
  number={3},
  pages={2289--2307},
  year={2023},
  publisher={IEEE}
}

@article{marvi2021safe,
  title={Safe reinforcement learning: A control barrier function optimization approach},
  author={Marvi, Zahra and Kiumarsi, Bahare},
  journal={International Journal of Robust and Nonlinear Control},
  volume={31},
  number={6},
  pages={1923--1940},
  year={2021},
  publisher={Wiley Online Library}
}

@article{emam2022safe,
  title={Safe reinforcement learning using robust control barrier functions},
  author={Emam, Yousef and Notomista, Gennaro and Glotfelter, Paul and Kira, Zsolt and Egerstedt, Magnus},
  journal={IEEE Robotics and Automation Letters},
  year={2022},
  publisher={IEEE}
}

@inproceedings{liu2023safe,
  title={Safe control under input limits with neural control barrier functions},
  author={Liu, Simin and Liu, Changliu and Dolan, John},
  booktitle={Conference on Robot Learning},
  pages={1970--1980},
  year={2023},
  organization={PMLR}
}

@inproceedings{chen2024safety,
  title={Safety index synthesis with state-dependent control space},
  author={Chen, Rui and Zhao, Weiye and Liu, Changliu},
  booktitle={2024 American Control Conference (ACC)},
  pages={937--942},
  year={2024},
  organization={IEEE}
}

@inproceedings{chen2024real,
  title={Real-time safety index adaptation for parameter-varying systems via determinant gradient ascend},
  author={Chen, Rui and Zhao, Weiye and Liu, Ruixuan and Zhang, Weiyang and Liu, Changliu},
  booktitle={2024 American Control Conference (ACC)},
  pages={3531--3536},
  year={2024},
  organization={IEEE}
}

@inproceedings{yun2025safe,
  title={Safe control of quadruped in varying dynamics via safety index adaptation},
  author={Yun, Kai S and Chen, Rui and Dunaway, Chase and Dolan, John M and Liu, Changliu},
  booktitle={2025 IEEE International Conference on Robotics and Automation (ICRA)},
  pages={7771--7777},
  year={2025},
  organization={IEEE}
}

@inproceedings{wei2022safe,
  title={Safe control with neural network dynamic models},
  author={Wei, Tianhao and Liu, Changliu},
  booktitle={Learning for Dynamics and Control Conference},
  pages={739--750},
  year={2022},
  organization={PMLR}
}

@inproceedings{sootla2022saute,
  title={Saut{\'e} rl: Almost surely safe reinforcement learning using state augmentation},
  author={Sootla, Aivar and Cowen-Rivers, Alexander I and Jafferjee, Taher and Wang, Ziyan and Mguni, David H and Wang, Jun and Ammar, Haitham},
  booktitle={International Conference on Machine Learning},
  pages={20423--20443},
  year={2022},
  organization={PMLR}
}

@inproceedings{zheng2024safe,
title={Safe Offline Reinforcement Learning with Feasibility-Guided Diffusion Model},
author={Yinan Zheng and Jianxiong Li and Dongjie Yu and Yujie Yang and Shengbo Eben Li and Xianyuan Zhan and Jingjing Liu},
booktitle={The Twelfth International Conference on Learning Representations},
year={2024},
}

@inproceedings{liu2023constrained,
  title={Constrained decision transformer for offline safe reinforcement learning},
  author={Liu, Zuxin and Guo, Zijian and Yao, Yihang and Cen, Zhepeng and Yu, Wenhao and Zhang, Tingnan and Zhao, Ding},
  booktitle={International conference on machine learning},
  pages={21611--21630},
  year={2023},
  organization={PMLR}
}

@article{zhang2024cvar,
  title={CVaR-constrained policy optimization for safe reinforcement learning},
  author={Zhang, Qiyuan and Leng, Shu and Ma, Xiaoteng and Liu, Qihan and Wang, Xueqian and Liang, Bin and Liu, Yu and Yang, Jun},
  journal={IEEE transactions on neural networks and learning systems},
  volume={36},
  number={1},
  pages={830--841},
  year={2024},
  publisher={IEEE}
}

@article{yu2023safe,
  title={Safe model-based reinforcement learning with an uncertainty-aware reachability certificate},
  author={Yu, Dongjie and Zou, Wenjun and Yang, Yujie and Ma, Haitong and Li, Shengbo Eben and Yin, Yuming and Chen, Jianyu and Duan, Jingliang},
  journal={IEEE Transactions on Automation Science and Engineering},
  volume={21},
  number={3},
  pages={4129--4142},
  year={2023},
  publisher={IEEE}
}

@article{ma2021feasible,
  title={Feasible actor-critic: Constrained reinforcement learning for ensuring statewise safety},
  author={Ma, Haitong and Guan, Yang and Li, Shegnbo Eben and Zhang, Xiangteng and Zheng, Sifa and Chen, Jianyu},
  journal={arXiv preprint arXiv:2105.10682},
  year={2021}
}

@article{yang2025scalable,
  title={Scalable synthesis of formally verified neural value function for hamilton-jacobi reachability analysis},
  author={Yang, Yujie and Hu, Hanjiang and Wei, Tianhao and Li, Shengbo Eben and Liu, Changliu},
  journal={Journal of Artificial Intelligence Research},
  volume={83},
  year={2025}
}

@article{yang2023synthesizing,
  title={Synthesizing control barrier functions with feasible region iteration for safe reinforcement learning},
  author={Yang, Yujie and Zhang, Yuhang and Zou, Wenjun and Chen, Jianyu and Yin, Yuming and Li, Shengbo Eben},
  journal={IEEE Transactions on Automatic Control},
  volume={69},
  number={4},
  pages={2713--2720},
  year={2023},
  publisher={IEEE}
}

@article{ganai2024hamilton,
  title={Hamilton-jacobi reachability in reinforcement learning: A survey},
  author={Ganai, Milan and Gao, Sicun and Herbert, Sylvia L},
  journal={IEEE Open Journal of Control Systems},
  volume={3},
  pages={310--324},
  year={2024},
  publisher={IEEE}
}

@article{garcia2015comprehensive,
  title={A comprehensive survey on safe reinforcement learning},
  author={Garc{\i}a, Javier and Fern{\'a}ndez, Fernando},
  journal={Journal of Machine Learning Research},
  volume={16},
  number={1},
  pages={1437--1480},
  year={2015}
}

@article{gu2024review,
  title={A review of safe reinforcement learning: Methods, theories and applications},
  author={Gu, Shangding and Yang, Long and Du, Yali and Chen, Guang and Walter, Florian and Wang, Jun and Knoll, Alois},
  journal={IEEE Transactions on Pattern Analysis and Machine Intelligence},
  year={2024},
  publisher={IEEE}
}

\end{document}